\documentclass[preprint,review,3p]{elsarticle}
\usepackage{natbib}
\usepackage{amsthm}

\newcommand{\cb}[1]{\ifmmode {\boldsymbol{#1}}\else ${\boldsymbol{#1}}$\fi}
\newcommand{\cp}[1]{\ifmmode {\mathcal{#1}}\else ${\mathcal{#1}}$\fi}

\theoremstyle{remark}
\newtheorem*{remark}{\textbf{Remark}}

\theoremstyle{definition}
\newtheorem{definition}{\textbf{Definition}}[section]

\newtheorem{theoreme}{Theorem}

\usepackage[T1]{fontenc}
\usepackage{psfrag,epsfig,graphics}
\usepackage{psfrag}
\usepackage{multirow}
\usepackage{multicol}
\usepackage{color}
\usepackage{colortbl}
\usepackage[centertags]{amsmath}
\usepackage{amsfonts}
\usepackage{amssymb}
\usepackage{newlfont}
\usepackage{graphicx}
\usepackage{float}
\usepackage{subcaption} 
\usepackage{multirow}
\usepackage{longtable}
\usepackage{color}
\usepackage{hyperref}
\hypersetup{
colorlinks=true,
linkcolor=blue,
urlcolor=cyan
}
\usepackage{soul}
\usepackage{comment}
\usepackage{colortbl}
\definecolor{Gray}{gray}{0.9}
\newcolumntype{a}{>{\columncolor{Gray}}c}

\usepackage[ruled]{algorithm2e}
\usepackage{lineno}
\definecolor{Gray}{gray}{0.9}

\usepackage{mathrsfs}
 \makeatletter
    \def\ps@pprintTitle{%
       \let\@oddhead\@empty
       \let\@evenhead\@empty
       \let\@oddfoot\@empty
       \let\@evenfoot\@oddfoot
    }
    \makeatother

\usepackage{setspace}
\journal{}

\begin{document}

\begin{frontmatter}

\title{Extended Hybrid Timed Petri Nets with Semi-Supervised Anomaly Detection for Switched Systems, Modelling and Fault Detection}
\author[a]{Fatiha Hamdi}
\author[b]{Abdelhafid Zeroual}
\author[c]{Fouzi Harrou}

\address[a]{Department of Electronics, Faculty of technology, LAAAS Laboratory, Batna 2 University, Batna, Algeria.}
\address[b]{Department of Electrical Engineering, Faculty of Science and Applied Sciences, Larbi Ben M’hidi University, Oum El Bouaghi, Algeria.
\\ E-mail: abdelhafid.zeroual@univ-oeb.dz}
\address[c]{Computer, Electrical and Mathematical Sciences and Engineering (CEMSE) Division, King Abdullah University of Science and Technology (KAUST), Thuwal, 23955-6900, Saudi Arabia.}

\begin{abstract}

Hybrid physical systems combine continuous and discrete dynamics, which can be simultaneously affected by faults. Conventional fault detection methods treat these dynamics separately, limiting their ability to capture interacting fault patterns. This paper presents a unified fault detection framework for hybrid dynamical systems by integrating a novel Extended Timed Continuous Petri Net (ETCPN) model with semi-supervised anomaly detection. The proposed ETCPN formalism advances existing Petri net models by introducing marking-dependent flow functions that natively encode the switching behavior of hybrid systems, enabling tighter and more intrinsic coupling between discrete and continuous dynamics than standard Hybrid Petri Nets. Methodologically, a mode-dependent hybrid observer is designed based on the ETCPN structure. The stability of the observer under arbitrary switching is rigorously guaranteed via Linear Matrix Inequalities (LMIs), which are solved offline to obtain the observer gains for each discrete mode. This observer generates residuals that reflect discrepancies between estimated and measured outputs. These residuals are then processed by semi-supervised detection methods—including One-Class SVM (OC-SVM), Support Vector Data Description (SVDD), and Elliptic Envelope (EE)—which are trained exclusively on fault-free data to eliminate the need for labeled fault examples. The framework is validated through comprehensive simulations covering three fault scenarios: discrete-event faults (e.g., mode blocking), continuous faults (e.g., sensor biases), and simultaneous hybrid faults. Results demonstrate high detection accuracy, rapid convergence, and robust performance across all fault types, with OC-SVM and SVDD achieving the best balance between detection rate and false alarms. The online execution of the framework involves lightweight computations, making it suitable for real-time deployment, with the main computational effort confined to the offline LMI design phase.

\end{abstract}


\end{frontmatter}
\section{Introduction}\label{sec:Introduction}

Modern engineering systems increasingly exhibit hybrid dynamic characteristics, which reflect the interaction between continuous and discrete state behaviours. This dual nature is prevalent in various technological domains, including autonomous transportation, advanced robotics, distributed energy networks, financial modelling, and biomechanical systems \cite{taghavian2024constrained, Zeroual2017ICEE-B}. Hybrid dynamic systems (HDS) present significant challenges in monitoring and control, particularly in anomaly detection and fault diagnostics~\cite{khorasgani2017structural, wang2012hybrid}. The stochastic nature of fault occurrence in HDS architectures further complicates diagnostic methodologies. Environmental perturbations and measurement uncertainties present a difficulty in distinguishing between normal and faulty operational states~\cite{shi2024novel}. Various diagnostic frameworks have been proposed to address these issues, with particular emphasis on model-based approaches~\cite{yang2010observer, diedrich2019model}. These methods assess discrepancies between predicted and actual system behaviours to infer potential faults. Integrated strategies that combine continuous state estimation with discrete mode reasoning have demonstrated considerable effectiveness in fault detection and isolation tasks~\cite{hofbaur2002mode}.

\medskip 

Recent advances in fault detection and isolation have led to significant progress in adaptive and hybrid observer architectures for complex dynamic systems. Liu et al.~\cite{liu2018new} proposed a novel approach that employed a bank of adaptive observers for nonlinear discrete time-varying systems with multiple fault scenarios. Their method transformed fault isolation into a reconstruction-based contribution analysis task to identify the most suitable monitor with the system’s current state. In the field of autonomous aerial vehicles, Lien et al.~\cite{lien2020adaptive} developed a dual observer framework for fault detection and estimation in quadrotor flight control systems. This framework achieves fault detection and isolation by systematically integrating observer outputs. For high-speed rail transportation, Yin et al.~\cite{yin2024dynamic} introduced an advanced hybrid observer paradigm for early slip detection based on adhesion coefficient estimation, enabling precise and timely fault identification. In the context of cyber-physical systems, Yan et al.~\cite{yan2024fault} proposed a hybrid observer scheme for simultaneous state and fault estimation, which incorporates dynamic correction terms that adapt to output availability. Stability and exponential convergence of the approach were theoretically ensured through Lyapunov-based analysis.

\medskip

Hybrid Dynamic Systems (HDS) are a class of dynamical systems characterized by the co-existence and interaction of continuous dynamics and discrete events~\cite{novelli2025identification, zeroual_piecewise_2017}. The system's state evolves through differential or difference equations until a discrete event, triggered by internal conditions or external commands, causes an instantaneous switch to a different set of governing equations. These transitions are often governed by switching laws or guards, which trigger transitions between different continuous behavioral modes. This powerful modeling paradigm is widely applied across numerous engineering fields, including robotics, automotive control, aerospace , traffic management~\cite{zeroual2015calibration}, and power systems \cite{zhu2025stability, harrou2018traffic}. Formal modeling of HDS is typically achieved through dynamical models with automata or discrete event formalisms \cite{chen2021hybrid}. Among these, Petri Nets have emerged as a powerful tool for discrete event modelling, providing a graphical representation of HDS interactions and offering strong potential for fault detection applications~\cite{lefebvre2017discussion}. Conceptualized initially for modelling flow systems dominated by continuous dynamics, such systems inherently involve events that trigger transitions between behavioural states. To address this complexity, Hybrid Petri Nets (HPNs) were developed as an extension, enabling the simultaneous representation of continuous and discrete behaviours~\cite{davrazos2007modeling, chen2014novel}. Early fault detection and diagnosis approaches proved effective for isolated faults affecting either continuous or discrete dynamics independently~\cite{renganathan2011observer}. However, the intrinsic hybrid nature of physical systems demands comprehensive frameworks capable of detecting faults that may simultaneously affect both continuous and discrete domains.

\medskip
This paper presents a unified modelling and fault detection framework for hybrid dynamic systems based on an advanced Extended Timed Continuous Petri Nets (ETCPNs) paradigm. The proposed approach is specifically designed to represent hybrid phenomena in switched discrete-time systems. It synthesises ordinary Petri Nets and Timed Continuous Petri Nets (TCPNs) to establish a comprehensive modelling architecture suited for discrete-time switching systems with linear time-invariant (LTI) dynamics. ETCPNs adopt a bipartite structure: discrete events and state transitions are represented using a standard Petri Net (PN) formalism to capture logical sequences and concurrency, while continuous inter-state dynamics governed by LTI difference equations are modelled through the TCPN formalism, enabling precise representation of time-dependent behaviours. The incorporation of TCPNs is particularly valuable for embedding temporal constraints, thereby supporting advanced temporal modelling essential for fault detection and diagnosis~\cite{hatte2024transforming}. In addition, the ETCPNs introduce marking-dependent flow functions, where the flow rate of a continuous transition is defined as a function of the discrete marking. This formulation natively encodes switching phenomena, as the current operational mode directly determines the active set of recurrent equations. Compared to the generic test-arc semantics typically used in HPNs, ETCPNs achieve a tighter and more intrinsic coupling between discrete and continuous subsystems, resulting in a more expressive and self-contained modeling framework. Building on this unified modelling framework, a fault detection scheme based on a hybrid observer (ETCPN-HO) is developed. The proposed observer integrates discrete PN and TCPN structures operating in synchronised coordination. The discrete observer estimates both discrete markings and firing vectors, enabling active mode identification for continuous state estimation. These state estimates, in turn, update the marking and firing vectors. Observer convergence conditions are formulated via Linear Matrix Inequalities (LMIs), ensuring asymptotic error convergence. To enhance fault detection capability, the framework incorporates a residual-based fault detection strategy. Residuals are generated as the discrepancy between estimated and measured states and analysed using advanced semi-supervised anomaly detection techniques, including One-Class SVM, Support Vector Data Description (SVDD), and Elliptic Envelope (EE).
These methods require only fault-free data for training, which avoids the need for labelled fault datasets and addresses class imbalance issues commonly encountered in supervised approaches. By the integration of residual analysis with semi-supervised detection, the ETCPN-HO framework offers accurate and interpretable identification of faults affecting both discrete and continuous system dynamics. Overall, this methodology represents a significant advancement, providing a scalable and data-efficient solution for hybrid discrete-continuous system fault detection.

\medskip
This paper is organised as follows. Section~\ref{sec2} reviews related work on Petri Net (PN) modelling for dynamic, hybrid, and fault detection systems. Section~\ref{sec3} presents the problem formulation and analysis of switched dynamic hybrid (SDH) systems, which established the theoretical background. Section~\ref{sec4} introduces the Extended Timed Continuous Petri Net (ETCPN) framework and its mathematical formulation. Section~\ref{sec5} details the fault detection methodology, including the design of the ETCPN-based hybrid observer, residual generation and fault detection methodology. Section~\ref{sec6} validates the proposed approach through simulation studies, demonstrating the performance of the ETCPN model, hybrid observer, residual analysis, and fault detection capabilities. Finally, Section~\ref{sec7} concludes the paper and discusses future research directions.

\medskip

\section{Related Works}\label{sec2}

Petri Nets (PNs) are graphical and mathematical modelling tools initially developed to capture the dynamic behaviour of discrete event systems~\cite{petri1962kommunikation, hamdi2024ICSC}. Their ability to visually and rigorously represent system states, events, and transitions has enabled widespread adoption across diverse application domains, including manufacturing~\cite{grobelna2021challenges, kahraman2010manufacturing}, power systems~\cite{vescio2015petri}, transportation and traffic~\cite{julvez2005modelling, liang2021modeling, chen2014novel}, and biology~\cite{heiner2008petri}. This combination of intuitive graphical representation and formal mathematical foundation has established PNs as a versatile modelling approach for complex systems. To address evolving modelling requirements, several PN extensions have been introduced, such as Colored Petri Nets, Timed Petri Nets, and Stochastic Petri Nets~\cite{mecheraoui2021petri, ali2024modeling}. These extensions have further broadened the applicability of PNs, enabling more detailed and accurate representations of complex and heterogeneous system behaviours.

\medskip 
Hybrid Petri Nets (HPNs) represent an evolutionary extension of standard Petri Nets, designed to integrate discrete and continuous dynamics for more accurate modelling of complex systems. While traditional PNs depict instantaneous events through transitions and discrete states via places, HPNs extend this formalism by incorporating continuous behaviours in transitions and enabling the continuous evolution of place markings over time~\cite{david2010discrete}. This dual capability provides a formal graphical representation well suited for mixed-behaviour systems. HPNs have been applied across various domains, leading to several structural extensions. For example, Alla and David~\cite{alla1998continuous} demonstrated the utility of HPNs in describing and analysing manufacturing and logistics systems. Demongodin and Alla~\cite{demongodin1998differential} further extended HPNs to model batch and modular systems. Wang et al.~\cite{wang2005hybrid} proposed an HPN framework in which the continuous part models production dynamics and the discrete part captures delivery and order processes in manufacturing networks. Dotoli et al.~\cite{cavone2018hybrid} introduced a first-order HPN to model a baking system, effectively detecting production failures such as waste and chokepoints. Beyond manufacturing, HPNs have been adopted in other sectors. Di et al.~\cite{di2004modelling} applied HPNs to model urban traffic networks, integrating continuous traffic flows and discrete traffic light control. In the energy domain, Mishra et al.~\cite{mishra2023multi} utilised multi-agent HPNs to model and control direct current power in microgrids, representing distributed energy resources continuously while managing switching operations discretely. In the biological sciences, Brinkrolf et al.~\cite{brinkrolf2021vanesa} developed an extended HPN framework for the VANESA simulator to model and simulate complex biological processes. Despite their broad applicability, HPNs exhibit a significant limitation in their treatment of switching logic, which governs changes in continuous dynamics. In most implementations, this switching behavior is managed externally through a separate state machine or imperative code~\cite{david2010discrete}. Such external management decouples the control logic from the PN model and undermines the self-contained nature of the PN formalism. Furthermore, the interaction between the discrete and continuous subsystems in standard HPNs is often limited to basic enabling or disabling of continuous transitions based on discrete markings (generally implemented via inhibitor or test arcs)~\cite{ghomri2007modeling}. While effective for simple mode changes, this mechanism lacks the expressiveness required to accurately model complex, mode-dependent continuous behaviors, particularly in systems with tightly coupled hybrid dynamics~\cite{peleg2005using}. 

\medskip 

Petri Net (PN) structures have been extended to address fault detection challenges across various systems through numerous methodologies and extensions. Köhler et al.~\cite{kohler2023fault} proposed a diagnosis framework for discrete event manufacturing systems using signal-interpreted PNs, enabling fault detection and isolation through residual generation between fault-free and controlled system models. In the context of cyber-physical systems, He et al.~\cite{he2018modeling} developed a detection control strategy for aircraft fuel tank systems, integrating agent-oriented network modelling to represent both discrete and continuous dynamics. For power distribution networks, Jiang et al.~\cite{jiang2024petri} introduced a PN-based method to monitor power changes and detect broken input lines. These studies demonstrate the adaptability of PNs for diagnosing specialised technical domains. In addition, various analytical methodologies have emerged from discrete PN (DPN) extensions to address fault detection in complex systems. Alzalab et al.~\cite{alzalab2021trust} proposed a novel approach through the combination of coloured DPNs with a trust-based model. This incorporation aims to enhance fault detection accuracy by the comparison of estimated faulty states with nominal conditions. Arichi et al.~\cite{arichi2022fault} employed partially observed DPNs to design fault detection and identification algorithms based on algebraic observers and comparative analysis. Further advancements were made by Coquand et al.~\cite{coquand2023diagnosabilization}, who investigated diagnosability properties using a timed PN framework to analyse the temporal aspects of fault detection. Moreover, labelled PNs have been applied for online fault diagnosis in discrete event systems, which has demonstrated robust real-time performance~\cite{de2022online}. Collectively, these methodologies highlight the versatility of PN extensions in addressing the complexities of fault detection and diagnosis across diverse engineering applications.

\medskip

While existing PN-based approaches and their extensions have demonstrated significant potential for modeling and fault detection across diverse application domains, most remain constrained to either purely discrete or loosely integrated hybrid representations. The proposed ETCPN paradigm addresses these limitations by introducing two key innovations: marking-dependent flow functions and the explicit integration of timing semantics into continuous transitions. 
In ETCPNs, the flow rate of a continuous transition is not constant nor externally managed; instead, it is explicitly defined as a function of the discrete marking. This formulation natively encodes the core behavior of switched systems, where the active set of continuous differential equations is determined by the current discrete operational mode. Consequently, ETCPNs establish a tighter and more intrinsic coupling between discrete and continuous subsystems compared to the more generic “test arc” semantics in HPNs. Furthermore, ETCPNs extend the capabilities of TCPNs by embedding a discrete PN layer to formally capture the dynamics of discrete events that govern mode switches, thereby establishing a true hybrid modeling framework. By integrating the switching logic directly into the firing semantics, ETCPNs provide a unified, self-contained, and expressive formalism for modeling complex hybrid dynamic systems.

Beyond the Petri net formalism, switched systems, a prominent class of hybrid dynamical systems, have gained significant attention in fault diagnosis research over the past decade. Diagnostic methodologies for these systems can be broadly categorized into model-based and data-driven paradigms. A dominant model-based approach employs a multi-model observer design, where a bank of observers is designed for each system mode. For instance, Zhang et al. \cite{zhang2019robust} developed a robust fault detection method for discrete-time switched systems using multiple Lyapunov functions, enabling real-time fault identification and management. Kazemi and Montazeri \cite{kazemi2018new} and Du et al. \cite{du2022actuator} applied $H_\infty$ filtering techniques to achieve accurate state estimation and fault diagnosis in power inverters under disturbances. Extending this model-based paradigm, Alharabi et al. \cite{el2025line} proposed a framework robust to unknown inputs and stochastic noise, even with an unknown switching signal. Complementary to these model-based techniques, data-driven and classification-based (typically requiring labeled fault data) methods have emerged. For instance, Yahia et al. \cite{yahia2025actuator} introduced a classification-based estimation method that integrates real-time data processing to improve fault tolerance. Similarly, application-specific studies, such as the fault location algorithm for drive systems by Song and Kim \cite{song2018practical} or adaptive parameter identification for DC-DC converters by Li et al. \cite{li2019sensor, li2020robust}, demonstrate the value of tailored solutions. While these methods are effective in their respective contexts, they present notable limitations. Model-based approaches using frameworks like hybrid automata can become intractable for systems with high concurrency and complex resource sharing due to state explosion and reachability issue \cite{waez2013survey}. Meanwhile, data-driven methods require extensive labeled fault data for training, which is often unavailable or impractical to obtain in real-world industrial settings. Furthermore, they can lack interpretability, making it difficult to trace faults back to specific physical components or logical process sequences.

The limitations of existing methods, ranging from the modeling constraints of HPNs and the complexity of alternative formalisms to the dependency of data-driven approaches on labeled fault data, highlight a clear need for a more integrated and practical solution. This work addresses these challenges by introducing a comprehensive ETCPN-based modelling and fault detection framework. For the modelling task, ETCPN introduces marking-dependent flow functions, which natively encode switching behavior by defining continuous transition rates as functions of discrete markings. This creates a tighter discrete-continuous coupling than HPNs and offers a more compact, natural representation for systems with parallel processes compared to Hybrid Automata, thereby avoiding state-space explosion and reachability problems. For robust residual generation, the framework integrates a stable hybrid observer designed using LMIs, providing formal stability guarantees. This synergy between a transparent, process-oriented model and a rigorous model-based observer offers a powerful foundation for fault diagnosis. For fault detection, the framework leverages semi-supervised anomaly detection techniques, eliminating the dependency on labeled fault data that plagues many supervised fault detection approaches. This combination of an interpretable model with efficient detection enables actionable diagnostics without the need for extensive fault archives. By unifying intuitive modeling, stable estimation, and practical detection, the proposed ETCPN framework bridges critical gaps, offering a unified, scalable, and interpretable solution for fault detection in complex hybrid dynamic systems.

\section{Problem statement and materials}\label{sec3}
Fault detection and localisation are essential for maintaining the safe and reliable operation of industrial and real-world systems. Hybrid Dynamic Systems (HDS) offer an effective framework for representing the inherent complexities of such systems. Among various analytical tools, Petri Nets (PNs) are particularly notable for their robustness, providing both graphical and mathematical modelling capabilities. Their dual functionality enables comprehensive analysis of system behaviour, fault identification, and assurance of operational reliability.

\begin{figure}[h!]
  \centering
  \includegraphics[width=12cm]{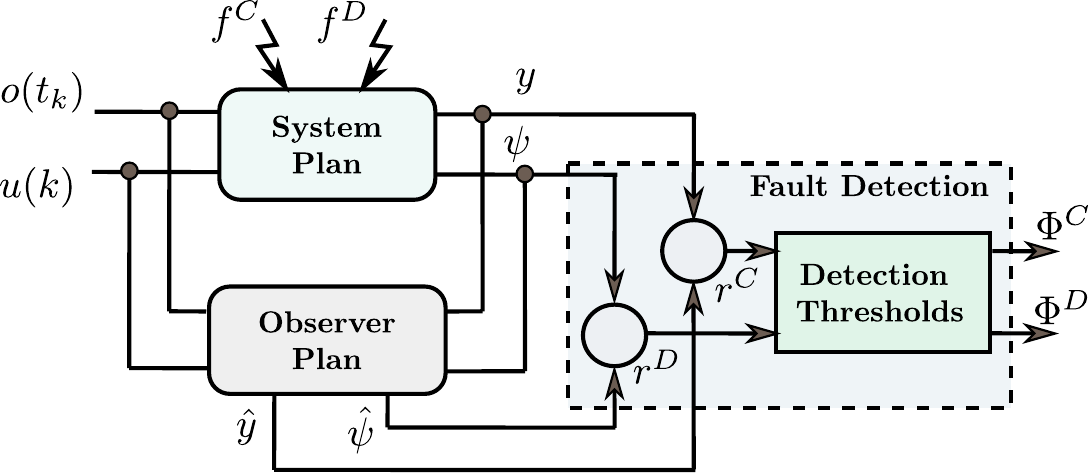}
  \caption{Schematic representation of the fault detection process for Hybrid Dynamic Systems (HDS). The system and observer plans generate outputs $(y, \psi)$ and $(\hat{y}, \hat{\psi})$, respectively. Residuals $r^C$ and $r^D$ are computed and compared against detection thresholds to identify faults affecting continuous ($f^C$) and discrete ($f^D$) dynamics, producing fault indicators $\Phi^C$ and $\Phi^D$.}
  \label{fig1}
\end{figure}

\medskip
Figure~\ref{fig1} illustrates the fault detection process specifically designed for Hybrid Dynamic Systems (HDS). In this framework, $f^C$ and $f^D$ represent faults affecting the continuous-time dynamics and discrete-event components, respectively. Fault detection is achieved by analysing the residuals generated through comparison between the actual and estimated system outputs. Specifically, the actual discrete-time outputs and discrete events, denoted by $(y, \psi)$, are compared with their estimated counterparts $(\hat{y}, \hat{\psi})$. Both the system and its observer are modelled using Petri Nets. As a result, the residual vector consists of two distinct components:
\begin{itemize}
  \item The discrete-time residual vector $r^C(k)$ for each discrete-time subsystem:
  \begin{equation}
    r^C(k) = y(k) - \hat{y}(k) {\color{red}.}
  \end{equation}
  
  \item The discrete-event residual vector $r^D(k)$ for each mode:
  \begin{equation}
    r^D(k) = \psi(k) - \hat{\psi}(k) {\color{red}.}
  \end{equation}
\end{itemize}

\noindent As shown in Figure~\ref{fig1}, $u(k)$ and $o(k)$ represent the inputs to the Hybrid Dynamic System (HDS), while $\Phi^C$ and $\Phi^D$ denote the fault indicators associated with the discrete-time subsystem and discrete-event components, respectively.

\noindent To effectively address fault detection in such systems, the following objectives must be achieved:

\begin{enumerate}

  \item Develop a Petri Net (PN) model capable of detecting malfunctions or irregular behaviours in the HDS, while ensuring the model remains robust and accurate.

  \item Design a PN-based framework that integrates both the continuous and discrete components of the Switched Dynamic Hybrid (SDH) system, under both nominal and faulty conditions.

\end{enumerate}

\subsection{Hybrid Dynamical System Class}

To address this aspect, consider the following switched hybrid dynamical system described by:

\begin{eqnarray}\label{sds}
x(k+1) & = & A_q x(k) + B_q u(k) \label{sds1}{\color{red}\,,}\\
y(k) & = & C_q x(k) \label{sds2}{\color{red}\,,}
\end{eqnarray}

\noindent
 
\noindent where $x(k) \in \mathbb{R}^n$ is the state vector at time step $k$, $u(k) \in \mathbb{R}^p$ and $y(k) \in \mathbb{R}^r$ are the input and output vectors, respectively. The system matrices $A_q \in \mathbb{R}^{n \times n}$, $B_q \in \mathbb{R}^{n \times p}$, and $C_q \in \mathbb{R}^{r \times n}$ are mode-dependent, corresponding to the active discrete mode $q(k) \in \mathcal{Q} = \{1, 2, \ldots, Q\}$, where $\mathcal{Q}$ is the finite set of discrete modes and $Q \in \mathbb{N}$ represents the number of discrete modes (or subsystems).

The evolution of the discrete mode $q(k)$ is governed by a state-dependent switching signal. The state space $\mathbb{R}^n$ is partitioned into $Q$ non-overlapping regions $\{\Omega_q\}_{q \in \mathcal{Q}}$, such that $\bigcup_{q=1}^Q \Omega_q = \mathbb{R}^n$ and $\Omega_q \cap \Omega_{q'} = \emptyset$ for $q \neq q'$, where $q$ and $q'$ are two different modes. The system operates in mode $q$ if and only if the state resides in the corresponding region{\color{red}.}
\begin{equation}
q(k) = q \quad \text{if and only if} \quad x(k) \in \Omega_q ~{\color{red}.}
\end{equation}
The switching signal $s(k): \mathbb{N} \to \mathcal{Q}$ is formally defined as the sequence of active modes, i.e., $s(k) = q(k)$. For analytical convenience, we define a Boolean indicator function $S_q(k)$ for each mode that explicitly signals its activation:

\begin{equation}
S_q(k) =
\begin{cases}
1 & \text{if } q(k) = q \quad (\text{i.e., } x(k) \in \Omega_q), \\
0 & \text{otherwise}.
\end{cases}
\end{equation}

These indicator functions satisfy $\sum_{q=1}^Q S_q(k) = 1$, ensuring the system is in exactly one mode at any time step $k$. Figure~\ref{fig2} illustrates a system with two modes ($Q=2$). 
The regions and corresponding switching signals are given by:

\[
\begin{aligned}
\Omega_1 &= \{ x \in \mathbb{R}^n \,|\, x \geq g_1 \}, \quad & S_1(k) = 1 \text{ if } x(k) \geq g_1, \\
\Omega_2 &= \{ x \in \mathbb{R}^n \,|\, x \leq g_2 \}, \quad & S_2(k) = 1 \text{ if } x(k) \leq g_2,
\end{aligned}
\]

\noindent
where $g_1$ and $g_2$ are constant thresholds that define the switching boundaries.

\begin{figure}[h!]
  \centering
  \includegraphics[width=6cm]{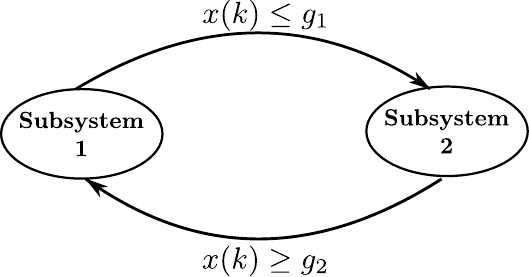}
  \caption{Illustration of subsystem switching logic in a hybrid dynamical system. Transitions between Subsystem~1 and Subsystem~2 are governed by the switching signal $S_q$ and conditions $g_1$ and $g_2$.}
  \label{fig2}
\end{figure}

To better reflect real-world conditions and account for potential system faults, the model in Equations~\eqref{sds1} and \eqref{sds2}  is reformulated as follows:

\begin{eqnarray}\label{sdsf}
x(k+1) & = & A_q x(k) + B_q u(k) + F^{x}_{q} f(k)\,, \\
y(k) & = & C_q x(k) + F^{y}_{q} f(k) {\color{red}\,,}
\end{eqnarray}

\noindent
where $F^{x}_{q}$ and $F^{y}_{q}$ are known fault distribution matrices of appropriate dimensions, representing how faults influence the state dynamics and output, respectively. This formulation explicitly incorporates additive fault dynamics, allowing for faults to affect both the state evolution and the measured outputs of the system.

\begin{figure}[h!]
  \centering
  \includegraphics[width=6.5cm]{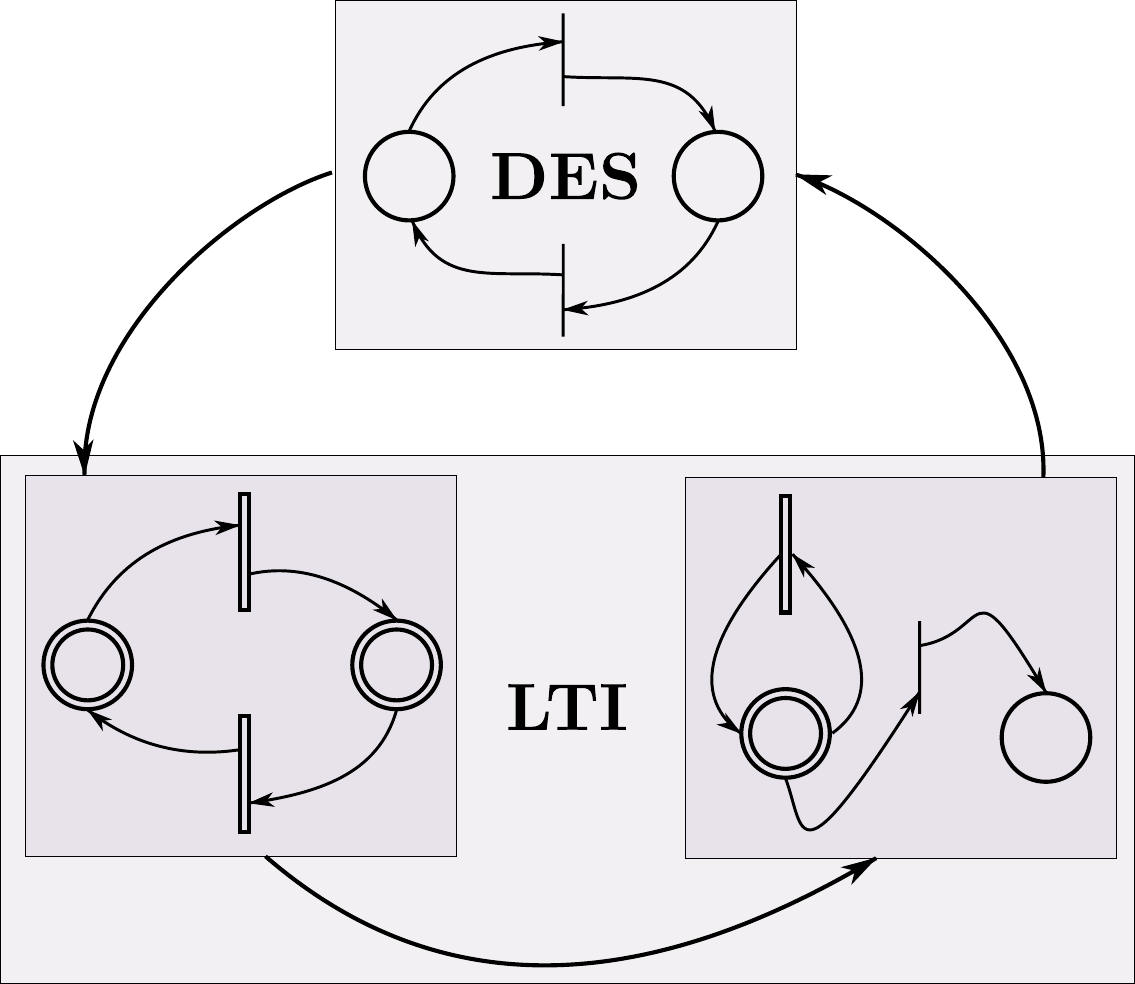}
  \caption{Structure of the SDH system. The system integrates Discrete-Event System (DES) dynamics with multiple LTI subsystems. Transitions between LTI subsystems are governed by the discrete-event controller.}
  \label{sdh}
\end{figure}

\noindent The primary challenge is to identify an appropriate Petri Net (PN) framework capable of accurately representing the hybrid system architecture while adhering to the schematic structure shown in Figure~\ref{sdh}. The proposed modelling framework comprises two distinct components:

\begin{itemize}
  \item A discrete-event PN formalism to capture the system's logical behaviour.
  \item A hybrid PN structure that integrates a continuous PN submodel for Linear Time-Invariant (LTI) system dynamics with a discrete PN submodel to govern mode switching conditions and transition logic.
\end{itemize}

\noindent
This dual-structure approach enables a comprehensive representation of both the continuous dynamics and discrete switching behaviour inherent in fault-prone hybrid systems. The continuous PN component models the LTI dynamics using appropriate place/transition structures, while the discrete PN component implements switching logic via guard conditions and transition firing rules.  LTI systems can be effectively described using Timed Continuous Petri Nets (TCPNs), which have demonstrated efficiency in modelling continuous processes~\cite{giua2018petri, Hamdi2024AFROS}. Building on this foundation, and in alignment with the Switched Dynamic Hybrid (SDH) structure illustrated in Figure~\ref{sdh}, we propose a combined modelling approach. Here, the continuous dynamics are captured using TCPN formalism, which allows for precise representation of continuous state evolution, temporal aspects of LTI behaviour, and flow dynamics through continuous places and transitions. Simultaneously, discrete-event dynamics are represented using Discrete Petri Nets (DPNs), which enable logical modelling of mode switching, description of discrete state transitions, and enforcement of guard conditions. This hybrid PN architecture offers a unified modelling framework for SDH systems, preserving the analytical strengths of TCPNs for continuous dynamics while leveraging DPNs for discrete-event representation and switching logic.

\subsection{Petri Net Definitions and Preliminaries}

To effectively integrate the two PN models and establish a robust framework for hybrid system modelling, a rigorous mathematical foundation is required. Modelling dynamical systems using PN methodologies necessitates a clear understanding of fundamental mathematical constructs and formal definitions. For further details, readers are referred to~\cite{giua2018petri, david2010discrete}.

\begin{definition}
An ordinary autonomous Discrete Petri Net (DPN) is formally defined as a five-tuple:

\begin{equation}\label{eq1pe}
N_D = (P^D, T^D, Pre^D, Post^D, M^D_{0}){\color{red}\,,}
\end{equation}

\noindent
where $P^D $ and $T^D $ are finite sets of discrete places and transitions, respectively (see Figure~\ref{pt}). The functions $Pre^D : T^D \times P^D \rightarrow \mathbb{N}$ and $Post^D : T^D \times P^D \rightarrow \mathbb{N}$ define the input and output incidence relations. Specifically, $Pre^D(P^D_{i}, T^D_{j})$ denotes the weight of the arc from place $P^D_{i}$ to transition $T^D_{j}$, and $Post^D(T^D_{j}, P^D_{i})$ denotes the weight of the arc from transition $T^D_{j}$ to place $P^D_{i}$. $M^D_0 \in \mathbb{N}$ is the initial discrete marking. The discrete marking dynamics is governed by:

\begin{equation}
M^D(k+1) = M^D(k) + W^D \sigma(k){\color{red}\,,}
\end{equation}

\noindent where $W^D=Post^D-Pre^D$ is the discrete incidence matrix and $\sigma(k)$ denote the discrete transition firing sequence vector.

\end{definition}

\begin{figure}[h]
  \centering
  \includegraphics[width=8.5cm]{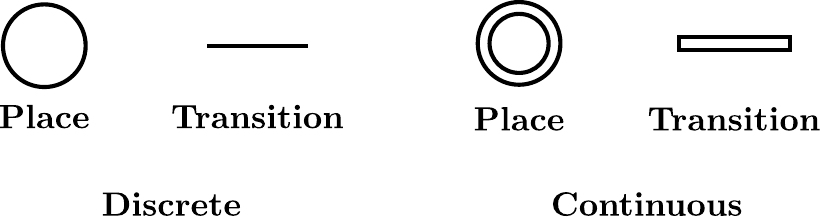}
  \caption{Graphical representation of a DPN showing discrete and continuous places and transitions}
  \label{pt}
\end{figure}

\begin{definition}
A Continuous Petri Net (CPN) is an autonomous PN characterised by real-valued place markings and continuous state evolution \cite{vardakis2025petri}. Graphically, continuous places and transitions are depicted using double lines (see Figure~\ref{pt}), distinguishing them from their discrete counterparts, which use single lines. Formally, a CPN is defined as a five-tuple:

\begin{equation}\label{eq1pe}
N_C = (P^C, T^C, Pre^C, Post^C, M^C_{0}){\color{red}\,,}
\end{equation}

\noindent
where $P^C$ and $T^C$ are finite sets of continuous places and transitions, respectively, with $P^C \cap T^C = \emptyset$. The matrices $Pre^C \in \mathbb{R}^{|T^C| \times |P^C|}$ and $Post^C \in \mathbb{R}^{|T^C| \times |P^C|}$ represent the input and output incidence relations, respectively. The vector $M^C_{0} \in \mathbb{R}^{|P^C|}$ defines the initial marking of the net.
\end{definition}

\begin{definition}
A Timed Continuous Petri Net (TCPN) extends the CPN model by incorporating temporal dynamics \cite{lefebvre2015gradient}. In TCPN, a time function is defined over the set of continuous transitions $T^C$, with $T^C \subseteq \mathbb{R} \cup \{\infty\}$. Each transition $T^C_j$ is associated with a corresponding maximum firing speed $V_j$. The actual firing speed $v_j(T^C_j)$ must satisfy:

\begin{equation}\label{fircond}
v_j(T^C_j) \leq V_j(T^C_j) {\color{red}.}
\end{equation}

\noindent
The marking evolution is governed by:

\begin{equation}
M^C(t + dt) = M^C(t) + W^C v(t) \, dt {\color{red}\,,}
\end{equation}

\noindent
where $W^C = Post^C - Pre^C$ is the continuous incidence matrix and $v(t)$ is the firing speed vector at time $t$. Over a finite time interval $[t_1, t_2]$, the marking dynamics are described by:

\begin{equation}\label{equmark}
M^C(t_2) = M^C(t_1) + W^C \int_{t_1}^{t_2} v(t) \, dt {\color{red}.}
\end{equation}

\end{definition}

\section{Proposed Extended Timed Continuous Petri Net (ETCPN)}\label{sec4}

To unify discrete and continuous dynamics within hybrid systems, an Extended Timed Continuous Petri Net (ETCPN) is proposed. This model integrates the key concepts from Discrete Petri Nets (DPN), Continuous Petri Nets (CPN), and Timed Continuous Petri Nets (TCPN), providing a coherent framework suited to the dynamic structure illustrated in Figure~\ref{sdh}. The ETCPN architecture combines two complementary formalisms: an ordinary discrete PN that captures discrete-event dynamics, such as mode transitions and fault occurrences, and a discrete-time TCPN that represents continuous state evolution governed by linear time-invariant (LTI) dynamics.

\begin{definition}\label{defETCPN}

The ETCPN is formally defined as a quadruple:
\begin{equation}
ETCPN = (N, f, M_0, \Gamma) {\color{red}.}
\end{equation}

\noindent Here $N$ denotes the net structure, $f$ is the function that assigns each node type, $M_0 \in \mathbb{R}$ the initial marking and $\Gamma$ is the timing map. Specifically, the net structure $N$ is defined by the tuple:

\begin{equation}\label{equP}
N = (P, T, Pre, Post) {\color{red}.}
\end{equation}

\noindent As defined in Equation~\eqref{equP}, $P$ represents the finite set of places, which is subdivided into continuous places $P^C$ and discrete places $P^D$. Likewise, $T$ denotes the finite set of transitions, partitioned into continuous transitions $T^C$ and discrete transitions $T^D$. These are formally expressed as:
\begin{eqnarray}
P &=& P^C \cup P^D{\color{red}\,,} \\
T &=& T^C \cup T^D {\color{red}.}
\end{eqnarray}

\noindent
The incidence matrices $Pre$ and $Post$, along with the initial marking $M_0$, follow the standard PN definition.

\noindent
The function $f$ assigns each node a type, indicating whether it is discrete (D) or continuous (C):

\begin{equation}
f: P \cup T \rightarrow \{D, C\} {\color{red}.}
\end{equation}

\noindent
The element $\Gamma$ defines the timing map, which assigns temporal attributes to transitions evolving over time, particularly those involved in discrete-time dynamics. The following assumptions are considered:

\begin{itemize}
  \item For discrete transitions, the timing map associates a fixed delay time $d_j$ to each transition $T_j$.
  
  \item For continuous transitions, a firing speed $v_j(k)$ is assigned. This speed reflects instantaneous values: it corresponds to the maximum firing speed when the transition is enabled, and is zero otherwise. In this case, the marking change occurs instantaneously at sampling instants. According to \cite{david2010discrete}, integrating the firing speed over $ [kT_s, (k+1)T_s] $ yields:
   \begin{equation}
     \int_{kT_s}^{(k+1)T_s} v_j(t) dt = v_j(k) T_s.
   \end{equation}
   
   Thus, the marking update Equation \eqref{equmark} becomes:
   \begin{equation}
   M^C((k+1)T_s) = M^C(kT_s) + W^C v_j(kT_s).
   \end{equation}
   For a switched system, the firing speed $v_j(k)$ depends on the active mode  $q$. Simplifying the notation by denoting $ M^C(k) \equiv M^C(kT_s) $, we obtain:
   \begin{equation}\label{eqdesc}
   M_q^C(k+1) = M_q^C(k) + W_q^C v_j(k){\color{red}.}
   \end{equation}  
   
   Since $ v_j(k) $ is proportional to the marking of the corresponding place (i.e., $ v_j(k) = M^C(k) $), Equation \eqref{eqdesc} is extended to:
   \begin{equation}\label{eqdescm20}
   M_q^C(k+1) = M_q^C(k) + W_q^C M_q^C(k) {\color{red}\,,}
   \end{equation} 
   where $ W_q^C $ is the mode-dependent incidence matrix.

\end{itemize}

\noindent To preserve the coherence of the Petri Net model and its formalism, we specify that:
\begin{itemize}
    \item $Pre$, $Post$, and $M$ take positive integer values when related to discrete PN,
    \item $Pre$, $Post$, and $M$ take real values when related to continuous PN.
\end{itemize}

\begin{remark}

Due to the special structure of the TCPN, if an arc connects a discrete place to a continuous transitions, there must be a same arc connecting this continuous transition to the same discrete place with identical weight.

\end{remark}

\end{definition}

\subsection{Marking Evolution of the ETCPN}

The marking dynamics of an Extended Timed Continuous Petri Net (ETCPN) evolve according to the firing sequence of transitions, which governs the system behaviour. To characterise the marking evolution, it is first necessary to compute the incidence matrix of the ETCPN, defined as:

\begin{equation}\label{mrk}
W_q = Post(P_i, T_j) - Pre(P_i, T_j) {\color{red}.}
\end{equation}

\noindent
According to Definition~\ref{defETCPN}, the incidence matrix $W_q$ can be partitioned as:

\begin{equation}\label{eqmrk}
W_q = 
\begin{bmatrix}
W_q^D & W_q^{CD} \\
W_q^{DC} & W_q^C
\end{bmatrix} {\color{red},}
\end{equation}

\noindent
where $W_q^D$ and $W_q^C$ represent the incidence matrices associated with the discrete-event and discrete-time dynamics, respectively. The matrices $W_q^{CD}$ and $W_q^{DC}$ describe the interactions between discrete and continuous dynamics and which are assumed to be zero for our case. Given the discrete transition firing sequence vector $\sigma(k)$, the overall marking dynamics are expressed as:

\begin{equation}\label{mark1}
 M_q(k+1) = \begin{bmatrix}
M_q^D(k+1) \\
M_q^C(k+1)
\end{bmatrix} = \begin{bmatrix}
M_q^D(k) + W_q^D \sigma(k)\\
M_q^C(k) + W_q^C M_q^C(k) )
\end{bmatrix} = M_q(k) + W_q \left( \sigma(k) + M_q^C(k) \right) {\color{red}.}
\end{equation}

\noindent
From Equation~\eqref{mark1}, the marking vector $M_q(k)$ naturally decomposes into two components: the discrete marking $M_q^D(k) \in \mathbb{N}^{|P^D|}$ and the continuous marking $M_q^C(k) \in \mathbb{R}^{|P^C|}$. Thus, the overall marking can be written as:

\begin{equation}\label{markm}
M_q(k) =
\begin{bmatrix}
M_q^D(k) \\
M_q^C(k)
\end{bmatrix}{\color{red}.}
\end{equation}

\begin{definition}[Enabling Conditions in ETCPN]
\label{def:enabling}

The ETCPN contains two types of transitions, discrete and continuous. Their enabling conditions are defined as follows:

\begin{itemize}
\item \textbf{Discrete Transition:} A discrete transition $T^D_j$ is \textbf{enabled} if and only if:
\begin{enumerate}
\item  For every input place $P_i^D$, the discrete marking meets or exceeds the arc weight: $M_q^D(P_i^D) \geq Pre(P_i^D, T^D_j)$.
\item Any guard condition associated with $T^D_j$ (state variable dependent condition) is satisfied.
\end{enumerate}

When enabled, $T^D_j$ may fire, instantaneously changing the discrete marking according to the incidence relation: $ W_q^D = Post(T^D_j, P^D_i) - Pre(P^D_i, T^D_j)$.

\item \textbf{Continuous Transition:} A continuous transition $T^C_j$ is \textbf{enabled} (or \textbf{active}) if and only if:
\begin{enumerate}
\item The discrete marking $M_q^D$ satisfies the condition for the operational mode $q$ in which $T^C_j$ is active. This is equivalent to the enabling of the discrete transitions that govern the mode switch.
\item The transition is not inhibited by its continuous inputs; that is, for all continuous input places $P_i^C$, the continuous marking must be sufficient to allow the flow (typically, $M_q^C(P_i^C) > 0$ for standard semantics).
\end{enumerate}

When enabled, $T^C_j$ fires continuously and must respect the condition of {\color{red}Equation~\eqref{fircond}}. The continuous marking evolves according to the marking-dependent flow function \label{markdep}. 
\end{itemize}

\end{definition}

\subsection{Hybrid Discrete-Time ETCPN Model representation}

To effectively represent discrete-time hybrid dynamic systems (HDS), the ETCPN formalism is employed to model the various system components, including discrete-time subsystems, discrete events, and their interactions. Figure~\ref{fig5} illustrates the general evolution of the discrete-event dynamics within an HDS. The discrete places $P^D$ define the $D$ modes of the HDS, which switch from one mode to another through discrete transitions $T^D$. Mode switching is governed by the following conditions:

\begin{itemize}
  \item Activation of an upstream place, indicated by the presence of a token.
  \item Satisfaction of a switching condition, evaluated through a test function.
\end{itemize}

The linear time-invariant (LTI) subsystems of the HDS can be mapped into the ETCPN framework under the following assumptions:

\begin{itemize}
  \item Each signal variable in the state-space representation is associated with a continuous place $P^C$.
  \item The number of continuous places is equal to the number of continuous transitions.
  \item Each continuous transition has a single input place, with the input incidence matrix defined as:
  
  \begin{equation}\label{eqind}
  Pre(P^C_{i}, T^C_{j}) =
  \begin{cases}
  1 & \text{if } i = j {\color{red}\,,} \\
  0 & \text{if } i \neq j {\color{red}.}
  \end{cases}
  \end{equation}
  
  \item The initial marking of the discrete places $P^D$ defines the initial operational mode of the system and must be specified accordingly. For the continuous places $P^C$, the initial marking is typically set to represent the initial state of the LTI subsystem, which can be zero, except for designated input places which may receive nonzero initial markings to represent external inputs.
  \item The incidence matrix integrates the state-space model matrices, as formalised in Theorem~\ref{theo1}.
\end{itemize}

\begin{theoreme}\label{theo1}
Given the ETCPN definitions and the assumptions presented above, the incidence matrix $W_q^C$ can be formulated as follows:

\begin{itemize}
   \item \textbf{Without output:}

  \begin{equation}\label{eqaut}
  W_q^C =
  \begin{bmatrix}
  I_{p \times p} & 0_{p \times n} \\
  B_q & A_q 
  \end{bmatrix} - I_{f \times f}{\color{red}.}
  \end{equation}

  \item \textbf{With output:}

  \begin{equation}\label{eqnaut}
  W_q^C =
  \begin{bmatrix}
  I_{p \times p} & 0_{p \times n} & 0_{p \times r} \\
  B_q & A_q& 0_{n \times r}\\
  C_q B_q & C_q A_q & 0_{r \times r}
  \end{bmatrix} - I_{f \times f}{\color{red}.}
  \end{equation}

\end{itemize}
\end{theoreme}

\noindent The detailed proof of Theorem~\ref{theo1} is provided in \hyperref[proofA]{Appendix A}.

\medskip 
The evolution of the ETCPN model replicates the behaviour of the discrete-time system defined in Equations~\eqref{sds1} and \eqref{sds2}. 
Importantly, the discrete state-space parameters are directly embedded within the $Post$ matrix.

\noindent Figure~\ref{fig5} illustrates the TCPN representation for the case without output. The continuous places $P^C_i$, for $i = 1, \ldots, p$, represent the input signals $u(k)$, while places $P^C_i$, for $i = p+1, \ldots, p+n$, correspond to the discrete-time states $x_i$, where $i = 1, \ldots, n$. The arc weights from continuous transitions to continuous places reflect the elements of the matrices $A_q \in \mathbb{R}^{n \times n}$ and $B_q \in \mathbb{R}^{n \times p}$.

\begin{figure}[H]
  \centering
  \includegraphics[scale=0.6]{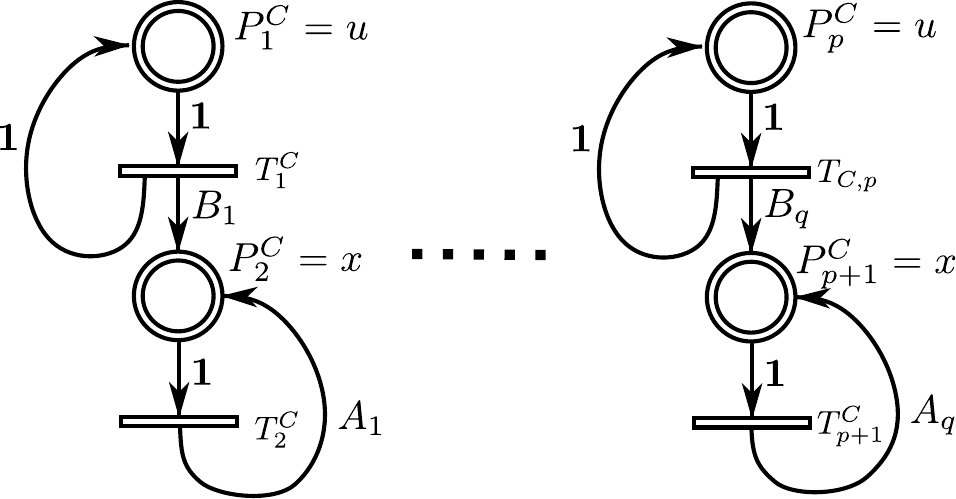}
  \caption{TCPN representation of ETCPN submodel (without output case)}\label{fig5}
\end{figure}

\noindent Figure~\ref{fig6} presents the TCPN structure for the case with output, where the graphical dynamics are extended to incorporate output variables.

\begin{figure}[H]
  \centering
  \includegraphics[scale=0.6]{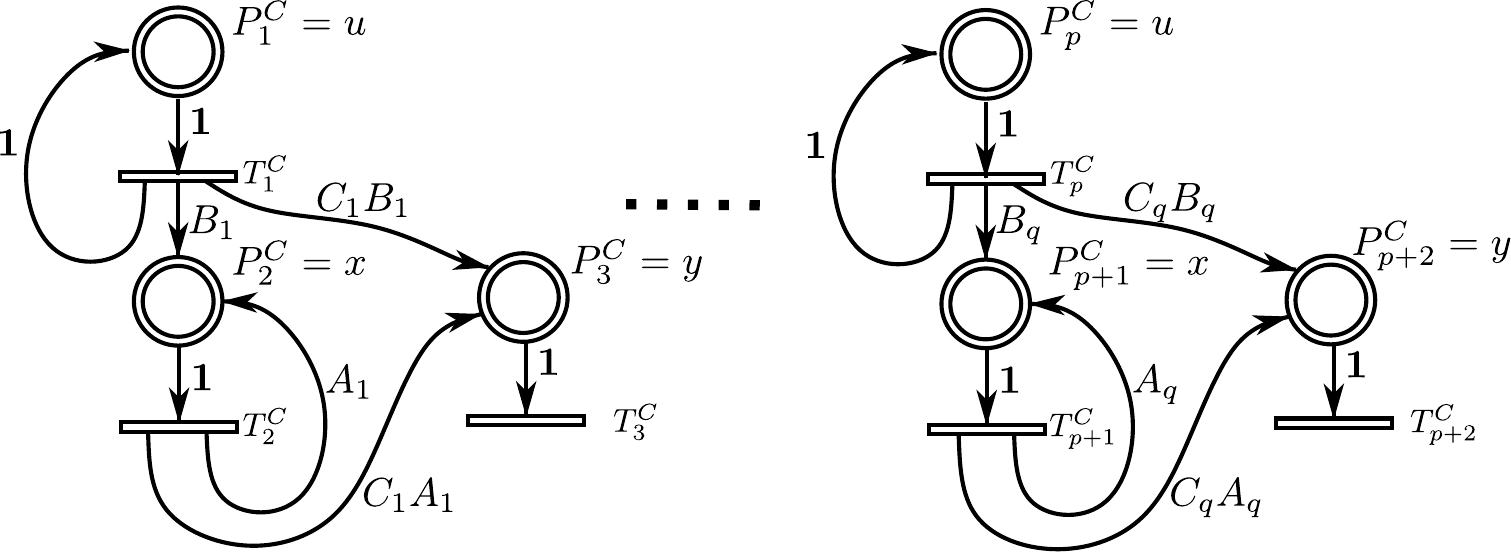}
  \caption{TCPN representation of ETCPN submodel (with output case)}\label{fig6}
\end{figure}

\noindent The interaction between the discrete and continuous dynamics of the HDS, which governs the switching conditions, is captured using upward and downward level crossings, as shown in Figure~\ref{fig7}. The LTI models are first constructed, and each switching condition is subsequently represented through the TCPN structure.

\begin{figure}[H]
  \centering
  \includegraphics[scale=0.6]{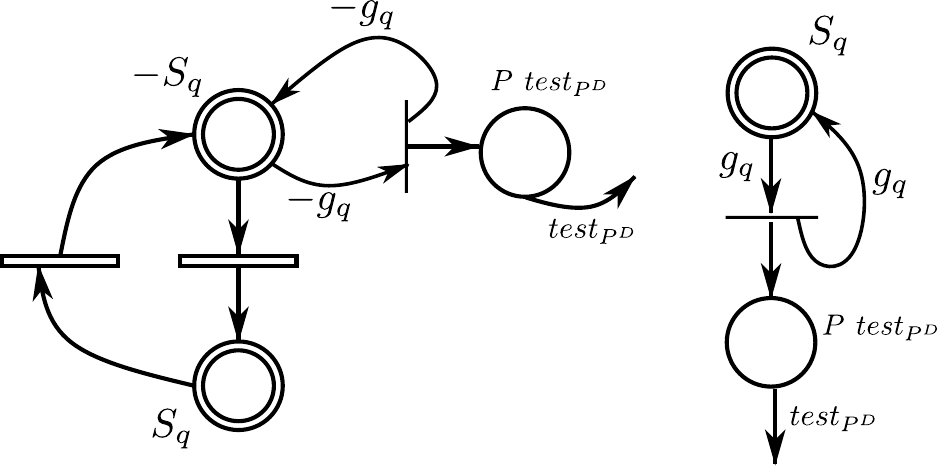}
  \caption{Graphical representation of upward and downward level crossings for switching conditions}\label{fig7}
\end{figure}

The arcs labeled $test_{P^D}$ in Figure~\ref{fig7} represent test arcs, a fundamental component of the ETCPN switching logic. A test arc from a discrete place $P^D$ to a transition does not consume tokens; its sole function is to check the marking of $P^D$. This mechanism formally encodes the enabling of transitions based on the discrete state of the system. Specifically, a continuous transition is active if and only if the discrete place associated with its operational mode is marked. This creates a direct, formal link between the discrete mode (defined by the marking of $P^D$ places) and the activation of the corresponding continuous dynamics (subsystems $A_q, B_q, C_q$ from Figure~\ref{fig5}), thereby implementing the core switching behavior of the hybrid system.

\subsection{Hybrid Evolution of the ETCPN Dynamics Model}

To effectively model hybrid dynamic systems (HDS) using the ETCPN framework, it is essential to establish the interaction and consistency between the continuous and discrete dynamics and their respective markings.

\noindent Assuming reciprocal arcs such that $W_q^{DC} = 0$ and $W_q^{CD} = 0$, Equation~\eqref{eqmrk} simplifies to:

\begin{equation}\label{eqmrk1}
W_q =
\begin{bmatrix}
W_q^D & 0 \\
0 & W_q^C
\end{bmatrix}{\color{red}.}
\end{equation}

\noindent This structure implies that no state jumps occur during the activation of each mode. Nonetheless, it remains necessary to define interaction equations that link the discrete and continuous markings. For a hybrid system, the incidence matrix $ W_q^C $ depends on the active discrete mode $ q $, which is encoded in the discrete marking $ M_q^D $. Following the approach in~\cite{HAMDI2009310}, we introduce the coupling through matrix $Z_q$, which defined as:

\begin{equation}\label{zeq}
Z_q = M_q^D (k) \otimes I_{f \times f}{\color{red}\,,}
\end{equation}

\noindent where $ \otimes $ denotes the Kronecker product and $ I_{f \times f} $ is an identity matrix of appropriate dimensions. To couple the discrete and continuous dynamics, we project the continuous incidence matrix through $Z_q$, which serves as a mode selector~\cite{davrazos2007modeling}. The term $ W_q^C Z_q $ effectively "selects" the active subsystem's dynamics based on $ M_q^D $. Thus, Equation \eqref{eqdescm20} is extended to:

\begin{equation}\label{mrkz}
M_q^C(k+1) = M_q^C(k) + W_q^C Z_q M_q^C(k){\color{red}.}
\end{equation}

\noindent This can be rewritten as:

\begin{eqnarray}
M_q^C(k+1) &=& \left( I + W_q^C Z_q \right) M_q^C(k) \nonumber\\
    &=& \left( I + W_q^C (M_q^D(k) \otimes I_{f \times f}) \right) M_q^C(k) {\color{red}.}
\end{eqnarray}

\noindent To explicitly show the dependence on discrete transitions. Substituting yields:

\begin{equation}\label{eqn33}
M_q^C(k+1) = \underbrace{\left(I + W_q^C \left(M_0^D(k) + W^D \sigma(k)\right) \otimes I_{f \times f} \right)}_{\bar{A}_q} M_q^C(k){\color{red}.}
\end{equation}

\noindent Based on Equations~\eqref{eqau2},~\eqref{sds}, and~\eqref{eqn33} the discrete-time representation of the HDS can be expressed as an augmented system:

\begin{eqnarray}
X(k+1) &=& \bar{A}_q X(k) {\color{red}\,,}\\
Y(k) &=&\begin{bmatrix}0 & C_q\end{bmatrix} X(k) {\color{red}\,,}\\
X(k) &=& \begin{bmatrix} u\\
 x\end{bmatrix}{\color{red}\,,}
\end{eqnarray}

\noindent where $X(k) \in \mathbb{R}^{P^C}$ is the augmented discrete-time state vector, $Y(k)$ denotes the output vector, and $\bar{A}_q \in \mathbb{R}^{P^C \times P^C}$ represents the mode-dependent system matrix.

\section{ETCPN-Based Residual Generator Framework}\label{sec5}

Building upon the modelling approach presented above, this section introduces the residual generation framework, illustrated in Figure~\ref{figge}. The proposed scheme is grounded in the ETCPN hybrid observer, which synthesizes a switched observer directly from the ETCPN model. This observer structure enables the systematic generation of residuals essential for fault detection.

\noindent For the estimation process, it is assumed that the active mode is unknown. Accordingly, both the ETCPN switched model and the proposed observer operate in parallel to provide accurate state and output estimations.

\noindent The output estimation task is divided into two complementary parts:

\begin{itemize}
  \item \textbf{Discrete-event observer:} responsible for monitoring and estimating the evolution of discrete markings and transition firings.
  \item \textbf{Discrete LTI dynamic observer:} responsible for estimating the continuous state evolution governed by LTI dynamics.
\end{itemize}

\noindent Figure~\ref{figge} depicts the interaction between these two components. The marking estimation is determined based on the observed system markings, estimated firings, and the switching logic. This estimated marking, in conjunction with the system input and output, drives the continuous state estimation process. Together, these estimations form the basis for residual generation, which is subsequently used for fault detection.

\begin{figure}[H]
  \centering
  \includegraphics[scale=0.7]{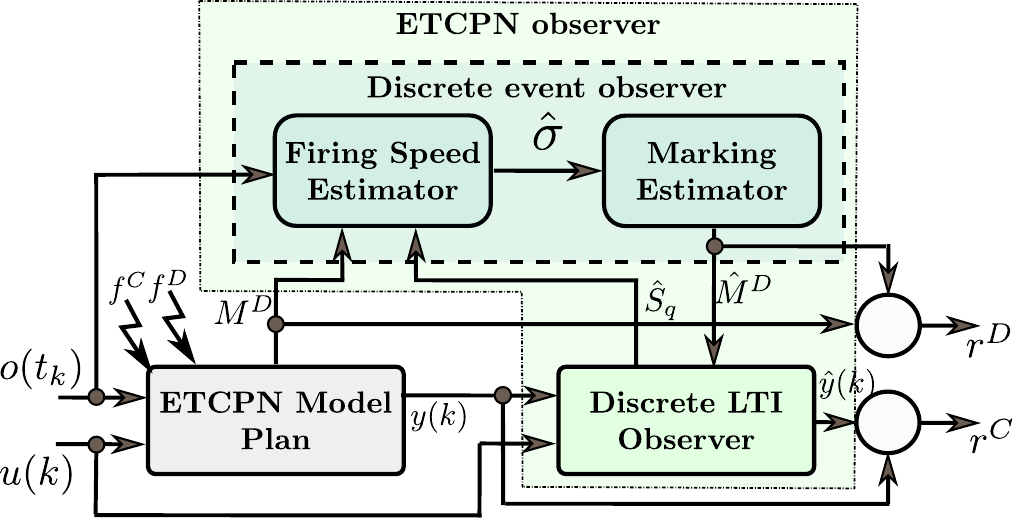}
  \caption{Architecture of the ETCPN-based hybrid observer framework}\label{figge}
\end{figure}

\subsection{Discrete Event Observer Design}

The design of the ETCPN discrete event observer is based on the classical reduced-order Luenberger observer approach~\cite{HAMDI2009310}, which enables the estimation of discrete markings. In this formulation, both the marking and firing vectors are partitioned into measurable and non-measurable components, defined as follows:

\begin{equation}\label{mark}
M^D(k) = 
\begin{bmatrix}
M^{D_o}(k) \\ M^{D_{\overline{o}}}(k)
\end{bmatrix}, \quad
\sigma(k) =
\begin{bmatrix}
\sigma^{o}(k) \\ \sigma^{\overline{o}}(k)
\end{bmatrix}{\color{red}\,,}
\end{equation}

\noindent where $M^{D_o}(k) \in \mathbb{N}^{P^{D_o}}$ and $M^{D_{\overline{o}}}(k) \in \mathbb{N}^{P^{D_{\overline{o}}}}$ represent the markings of measured and unmeasured places, respectively. Similarly, $\sigma^{o}(k) \in \mathbb{N}^{m_o}$ and $\sigma^{\overline{o}}(k) \in \mathbb{N}^{m_{\overline{o}}}$ correspond to the firing vectors of measured and unmeasured transitions.

\noindent The general state-space representation of the ETCPN system is given by:

\begin{eqnarray}\label{outy}
E \mathcal{Y}^D(k+1) &=& A \mathcal{Y}^D(k) {\color{red}\,,}\\
\overline{\psi}(k) &=& H \mathcal{Y}^D(k){\color{red}.}
\end{eqnarray}

\noindent Where:

\begin{equation*}
\mathcal{Y}^D(k) =
\begin{bmatrix}
M^{D}(k) \\ \sigma(k)
\end{bmatrix} \in \mathbb{N}^{P^D + m}, \quad
\overline{\psi}(k) =
\begin{bmatrix}
M^{D_o}(k) \\ \sigma^o(k)
\end{bmatrix} \in \mathbb{N}^{P^{D_o} + m_o}{\color{red}\,,}
\end{equation*}

\noindent and:

\begin{equation*}
E = 
\begin{bmatrix}
I_{P^D} & W^D
\end{bmatrix}, \quad
A = 
\begin{bmatrix}
I_{P^D} & 0_{P^D \times m}
\end{bmatrix}{\color{red}\,,}
\end{equation*}

\begin{equation*}
H =
\begin{bmatrix}
I_{P^{D_o}} & 0_{P^{D_o} \times P^{D_{\overline{o}}}} & 0_{P^{D_o} \times m_o} & 0_{P^{D_o} \times m_{\overline{o}}} \\
0_{m \times P^{D_o}} & 0_{m \times P^{D_{\overline{o}}}} & I_{m_o} & 0_{m_o \times m_{\overline{o}}}
\end{bmatrix}{\color{red}.}
\end{equation*}

\noindent The global marking estimation based on Equation~\eqref{outy} is achieved using the following observer design:

\begin{eqnarray}
\hat{M}^D(k+1) &=& F \hat{M}^D(k) + G \psi(k) {\color{red}\,,}\\
\hat{\mathcal{Y}}^D(k) &=& R \hat{M}^D(k) + N \psi(k){\color{red}.}
\end{eqnarray}

\noindent Where $F$, $G$, $R$, and $N$ are observer gain matrices, selected to ensure accurate estimation of the system outputs. Further details on their determination can be found in~\cite{HAMDI2009310}.

\subsection{Discrete-Time LTI Observer Design}

This subsection introduces the switched Luenberger observer, designed to estimate the discrete-time dynamics of the ETCPN model. The observer formulation is expressed as:

\begin{eqnarray} \label{eq43}
\hat{x}(k+1) &=& A_q \hat{x}(k) + B_q u(k) + L_q (y(k) - \hat{y}(k)) {\color{red}\,,}\\
\hat{y}(k) &=& C_q \hat{x}(k){\color{red}.}
\end{eqnarray}

\noindent Where $L_q$ denotes the observer gain associated with subsystem $q$. This formulation leads to a reformulated state-space representation:

\begin{equation}\label{osv45}
\hat{x}(k+1) = \underbrace{\left(A_q - L_q C_q\right)}_{\underline{A}_q} \hat{x}(k) + \overline{B}_q \, \overline{u}(k){\color{red}\,,}
\end{equation}

\noindent where $\overline{B}_q = \begin{bmatrix} B_q & L_q \end{bmatrix}$ and $\overline{u}(k) = \begin{bmatrix} u(k) & y(k) \end{bmatrix}$.

\noindent To maintain consistency with the ETCPN structure, this observer can be integrated into the ETCPN framework through the following incidence matrix:

\begin{equation}\label{wmat}
W^C_{o_q} =
\begin{bmatrix}
I_{p \times p} & 0_{p \times n} \\
\overline{B}_q & \underline{A}_q- I_{n \times n} %
\end{bmatrix}
- I_{f \times f}{\color{red}.}
\end{equation}

\noindent To compute the observer gains $L_q$ and guarantee observer convergence, a design procedure based on Linear Matrix Inequality (LMI) constraints is applied, as formalized in the following theorem.

\begin{theoreme}\label{theo2}

The estimation error $e(k) = x(k) - \hat{x}(k)$ of the ETCPN-based hybrid observer converges asymptotically to zero for any switching sequence between discrete modes $q \in Q$, if there exist symmetric positive definite matrices $T_q = T_q^T > 0$ for each mode $q$, and matrices $G_q$ and $F_q$ of appropriate dimensions, such that the following LMI is satisfied for every possible transition from mode $q$ to mode $\acute{q}$:

\begin{equation}\label{theob}
\begin{bmatrix}
T_q - G_q - G_q^T & G_q^T A_q^T - C_q^T F_q^T \\
\ast & -T_{\acute{q}}
\end{bmatrix} \prec 0{\color{red}.}
\end{equation}

\noindent Where $T_q = P_q^{-1}$ and $T_{\acute{q}} = P_{\acute{q}}^{-1}$. The observer gain $L_q$ for each discrete mode $q$ is then given by $L_q=F_q C_q (C_q G_q^{T})^{-1}$.

\end{theoreme}

The detailed proofs of Theorem~\ref{theo2} are provided in \hyperref[proofB]{Appendix B}. To illustrate the practical application of these LMI conditions, consider a simple system with two discrete modes. The stability of the observer must be guaranteed for all possible switching scenarios, which in this case are: 1) switching from Mode 1 to Mode 2, and 2) switching from Mode 2 to Mode 1. This requires solving a set of LMIs for each switching case.

\begin{itemize}
  \item \textbf{Case 1: Mode 1 switches to Mode 2 ($q=1, \acute{q}=2$)}
  For this switching event, the LMI \eqref{theob} is instantiated with the matrices corresponding to the initial mode ($q=1$) and the Lyapunov matrix for the target mode ($\acute{q}=2$). The resulting LMI is given by:
  
  \begin{equation}\label{LMI_1to2}
  \begin{bmatrix}
  T_1 - G_1 - G_1^T & G_1^T A_1^T - C_1^T F_1^T \\
  \ast & -T_2
  \end{bmatrix} \prec 0 {\color{red}.}
  \end{equation}
  
  \item \textbf{Case 2: Mode 2 switches to Mode 1 ($q=2, \acute{q}=1$)}
  Similarly, for the reverse switch, the LMI is instantiated with the matrices of mode $q=2$ and the Lyapunov matrix for mode $\acute{q}=1$:
  
  \begin{equation}\label{LMI_2to1}
  \begin{bmatrix}
  T_2 - G_2 - G_2^T & G_2^T A_2^T - C_2^T F_2^T \\
  \ast & -T_1
  \end{bmatrix} \prec 0 {\color{red}.}
  \end{equation}
\end{itemize}

To design the observer, the LMIs \eqref{LMI_1to2} and \eqref{LMI_2to1} must be solved simultaneously for the decision variables $T_1 \succ 0$, $T_2 \succ 0$, $G_1$, $G_2$, $F_1$, and $F_2$.

The choice of the LMI formulation in Theorem 2 is fundamentally dictated by the switched-system nature of the ETCPN model. The discrete Petri Net governs the transitions between continuous operational modes, each with distinct dynamics $(A_q, B_q, C_q)$. This switching behavior requires a mode-dependent Lyapunov stability condition that must hold for all possible transitions $q \rightarrow \acute{q}$ in the net's reachability graph. Consequently, the LMI structure explicitly encodes these transitions as individual constraints. Moreover, the matrices $A_q$ and $C_q$ within each LMI are not arbitrary but are structurally derived from the ETCPN's continuous subnet and sensor configuration, respectively. As such, the LMI constraints are not merely a mathematical tool but a direct formalization of the stability conditions imposed by the interconnected discrete and continuous structure of the ETCPN.

Once the ETCPN-based residual generator is designed, the generated residuals can be used for fault detection. Figure~\ref{fig_res} presents the overall detection scheme, highlighting the interaction between the nominal hybrid and faulty systems and the corresponding detection processes.

\begin{figure}[H]
  \centering
  \includegraphics[scale=0.7]{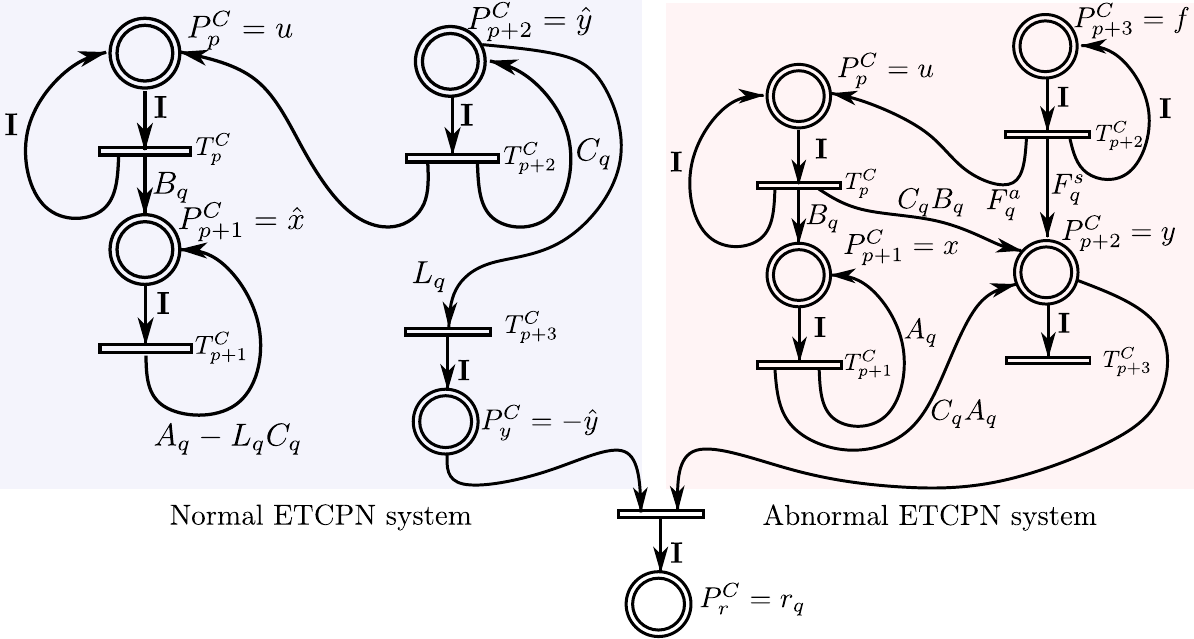}
  \caption{Fault detection process using the ETCPN-based residual generator}\label{fig_res}
\end{figure}

\subsection{Residual Generation and Fault Detection Methodology}
In the proposed ETCPN framework, residual signals serve as crucial indicators for fault detection by quantifying the deviation between actual system measurements and the estimates provided by the hybrid observer. These residuals are designed to capture inconsistencies caused by potential faults while remaining insensitive to normal process variations and measurement noise. In the proposed scheme, residuals are defined based on the system state vector rather than the outputs.

In the proposed scheme, residuals are defined based on the system state vector rather than the outputs. This formulation enables the detection of subtle deviations in system dynamics by comparing the actual state vector $x(k)$ with its estimated counterpart $\hat{x}(k)$, thus generating the state residual $r_x(k)$
    \begin{equation}
        r_x(k) = x(k) - \hat{x}(k) {\color{red}.}
    \end{equation}
  
Under fault-free conditions, these residuals remain close to zero, capturing only normal process variations and measurement noise. However, when faults occur, the residuals deviate significantly from zero, indicating abnormal system behaviour. State residuals are particularly effective for detecting faults that affect internal system dynamics, such as actuator malfunctions, parameter drifts, or component degradation. For example, partial actuator failures or gain changes in control loops directly influence the state evolution and are well captured by $r_x(k)$.

\medskip

\noindent In addition to state residuals, the proposed ETCPN framework supports the generation of discrete-event and output residuals, which offer complementary diagnostic capabilities. Discrete-event residuals $r_\psi(k)$, defined as the difference between observed and estimated discrete events, are useful for identifying logical faults, such as erroneous event triggering, missing events, or incorrect mode transitions. Similarly, output residuals $r_y(k)$, computed from the discrepancy between the measured output and its estimated value, are advantageous for detecting faults that manifest directly at the system outputs. Typical examples include sensor malfunctions, measurement bias, or signal transmission errors.

\medskip

\noindent Depending on the application requirements, residuals can be selectively used or combined. State residuals offer deep insight into dynamic faults, discrete-event residuals target logical and sequencing anomalies, while output residuals capture faults visible at the system boundaries. This flexibility enhances the robustness and diagnostic coverage of the proposed ETCPN-based fault detection methodology, allowing it to address a wide range of fault types in hybrid discrete-continuous systems.

\medskip

\noindent In this study, residual-based anomaly detection is applied primarily to the state residuals $r_x(k)$. However, the same detection methodology can be extended to discrete-event residuals $r_\psi(k)$ and output residuals $r_y(k)$ depending on the fault scenario and system characteristics.

\subsubsection{Anomaly Detection Techniques}

The generated residuals are analysed using three semi-supervised anomaly detection techniques, trained exclusively on fault-free data to avoid the need for labelled fault samples:

\begin{itemize}

    \item \textbf{One-Class Support Vector Machine (OC-SVM)}~\cite{scholkopf2001estimating}: 
    This kernel-based method learns a decision boundary that encloses the majority of the normal residual distribution in a high-dimensional feature space. The boundary is constructed to exclude a small fraction of potential anomalies, controlled by the hyperparameter $\nu$. During monitoring, residuals falling outside this learned boundary are flagged as anomalous. OC-SVM is particularly effective for capturing global deviations and is robust to small variations in the normal data.

\item \textbf{Support Vector Data Description (SVDD)}~\cite{tax1999support}: 
SVDD is a kernel-based anomaly detection technique that aims to encapsulate fault-free residual patterns within a hypersphere in the transformed feature space. The model learns to tightly enclose normal data points, while allowing for some flexibility through a user-defined parameter to accommodate natural variations and noise. Residuals falling outside this boundary are flagged as anomalies, making SVDD particularly effective in capturing abrupt or significant deviations that may indicate faults in the system dynamics. Furthermore, the kernel function provides the capability to model non-linear residual distributions, enhancing its fault detection versatility.

    \item \textbf{Elliptic Envelope (EE)}~\cite{rousseeuw1999fast}: 
    EE fits a multivariate Gaussian distribution to the residual data generated under fault-free conditions. An ellipsoidal decision boundary is derived to encompass the majority of the normal observations. Residuals that fall significantly outside this ellipsoid are classified as anomalies. This approach is well suited when residuals approximately follow a Gaussian distribution and offers a simple yet effective solution for multivariate outlier detection.

\end{itemize}

\noindent These techniques operate directly on residual features and offer complementary detection capabilities, enhancing the robustness of the proposed fault detection framework.

\medskip

\noindent The schematic in Figure~\ref{fig_detection_schematic} summarizes the residual generation process and the integration of semi-supervised anomaly detection methods for fault decision-making.

\begin{figure}[h!]
  \centering
  \includegraphics[width=4.5cm]{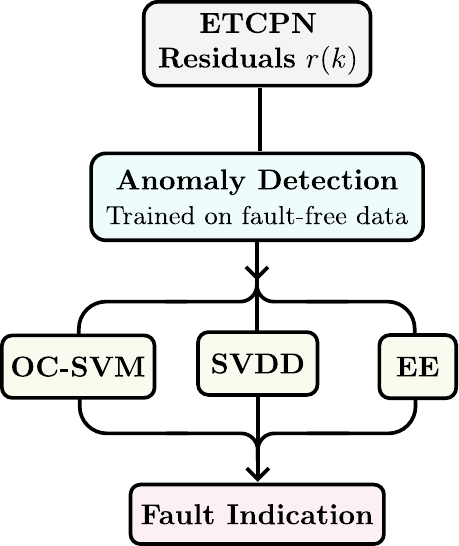}
  \caption{Residual-based fault detection scheme using ETCPN observer and semi-supervised anomaly detection methods}\label{fig_detection_schematic}
\end{figure}

\subsubsection{Detection performance evaluation metrics}

The performance of the proposed fault detection approach is quantitatively assessed using the following widely adopted evaluation metrics:

\begin{itemize}
  \item \textbf{Accuracy (ACC):}
  \begin{equation}
  Accuracy = \frac{TP + TN}{TP + TN + FP + FN} {\color{red}.}
  \end{equation}
  Represents the overall proportion of correctly classified instances, both fault-free and faulty.

  \item \textbf{True Positive Rate (Recall):}
  \begin{equation}
  TPR = \frac{TP}{TP + FN} {\color{red}.}
  \end{equation}
  Measures the ability to correctly identify actual faults, also referred to as detection rate or sensitivity.

  \item \textbf{False Positive Rate (FPR):}
  \begin{equation}
  FPR = \frac{FP}{FP + TN} {\color{red}.}
  \end{equation}
  Reflects the proportion of fault-free samples incorrectly classified as faulty, directly linked to false alarm occurrences.

  \item \textbf{F1 Score:}
  \begin{equation}
  F1 = 2 \times \frac{Precision \times Recall}{Precision + Recall}{\color{red}\,,} \quad \text{where} \quad Precision = \frac{TP}{TP + FP} {\color{red}.}
  \end{equation}
  The harmonic mean of precision and recall, providing a balanced assessment that considers both detection capability and false alarms. Where, $T$ and $F$ are respectively true and false, while $P$ and $N$ denote positive and negative successively.
\end{itemize}

\noindent These metrics offer a comprehensive evaluation framework to validate the reliability, robustness, and detection accuracy of the proposed ETCPN-based semi-supervised fault detection methodology.

\section{Results and Discussion}\label{sec6}

To validate the theoretical developments related to the ETCPN modelling and fault detection framework, a hybrid discrete-time system composed of two modes ($q = \{1, 2\}$) is considered. The dynamics of each mode are defined by the following state-space matrices:

\begin{equation*}
A_1 = 
\begin{bmatrix}
\cos \left(\frac{\pi}{3}\right) & \sin \left(\frac{\pi}{3}\right) \\
- \sin \left(\frac{\pi}{3}\right) & \cos \left(\frac{\pi}{3}\right)
\end{bmatrix},
\quad
B_1 = 
\begin{bmatrix}
1 \\
0
\end{bmatrix}{\color{red}.}
\end{equation*}

\begin{equation*}
A_2 = 
\begin{bmatrix}
\cos \left(\frac{2\pi}{3}\right) & \sin \left(\frac{2\pi}{3}\right) \\
- \sin \left(\frac{2\pi}{3}\right) & \cos \left(\frac{2\pi}{3}\right)
\end{bmatrix},
\quad
B_2 = 
\begin{bmatrix}
1 \\
0
\end{bmatrix}{\color{red}.}
\end{equation*}

\noindent To align with the ETCPN formalism, the discrete-event component is modelled using two discrete places and two discrete transitions, which enable switching between the LTI subsystems according to the conditions $x_1(k)>0$, $x_1(k) \leq 0$ and $M_q^D(k)$. The mode selection mechanism adheres to the following rules:  

\begin{itemize}
    \item When the switching law $S_1$ is satisfied ($x_1(k) > 0$ \text{and} $M_1^D(k)= \begin{bmatrix} 1 & 0 \end{bmatrix}^T$), the second subsystem is activated.
    \item If $S_2$ is satisfied $(x_1(k) \leq 0$ \text{and} $M_2^D(k)= \begin{bmatrix} 0 & 1 \end{bmatrix}^T )$, the first subsystem is activated.
\end{itemize}

\noindent Based on Theorem~\ref{theo1} and the switching logic described above, the continuous marking vector $M_q^C(k)$ can be expressed as:

\begin{equation*}
M_q^C(k) = 
\begin{bmatrix}
u(k) & x_1(k) & x_2(k)
\end{bmatrix}^T{\color{red}.}
\end{equation*}

\noindent The corresponding $Post$ and $Pre$ matrices for the two modes are defined as follows:

\begin{equation*}
Post^C_1 = 
\begin{bmatrix}
1 & 0 & 0 \\
1 & \cos \left(\frac{\pi}{3}\right) & \sin \left(\frac{\pi}{3}\right) \\
0 & - \sin \left(\frac{\pi}{3}\right) & \cos \left(\frac{\pi}{3}\right)
\end{bmatrix},
\quad
Post^C_2 = 
\begin{bmatrix}
1 & 0 & 0 \\
1 & \cos \left(\frac{2\pi}{3}\right)& \sin \left(\frac{2\pi}{3}\right) \\
0 & - \sin \left(\frac{2\pi}{3}\right) & \cos \left(\frac{2\pi}{3}\right)
\end{bmatrix}{\color{red}.}
\end{equation*}

\begin{equation*}
Pre^C_1 = Pre^C_2 = 
\begin{bmatrix}
1 & 0 & 0 \\
0 & 1 & 0 \\
0 & 0 & 1
\end{bmatrix}{\color{red}.}
\end{equation*}

\noindent Finally, the resulting incidence matrices $W^C_1$ and $W^C_2$ for each mode are:

\begin{equation*}
W^C_1 =
\begin{bmatrix}
0 & 0 & 0 \\
1 & \cos \left(\frac{\pi}{3}\right)-1 & \sin \left(\frac{\pi}{3}\right) \\
0 & - \sin \left(\frac{\pi}{3}\right) & \cos \left(\frac{\pi}{3}\right)-1
\end{bmatrix},
\quad
W^C_2 =
\begin{bmatrix}
0 & 0 & 0 \\
1 & \cos \left(\frac{2\pi}{3}\right) -1& \sin \left(\frac{2\pi}{3}\right) \\
0 & - \sin \left(\frac{2\pi}{3}\right) & \cos \left(\frac{2\pi}{3}\right)-1
\end{bmatrix} {\color{red}.}
\end{equation*}

\noindent This numerical configuration serves as a benchmark scenario to assess the modelling capacity of the proposed ETCPN approach and to validate the performance of the residual-based fault detection scheme under hybrid dynamic conditions.

\medskip
Figure~\ref{fig10} presents the ETCPN schematic representation of the considered switched hybrid system comprising two operational modes. The model integrates three key components: discrete-event dynamics, continuous LTI subsystems, and switching laws, all structured within a unified HDS framework. Discrete events and switching logic are represented through discrete places and transitions, whereas continuous dynamics are described using continuous places corresponding to input and state variables. The interconnections between continuous transitions and places reflect the system matrices, effectively capturing the evolution of continuous state variables under each mode. Additionally, the switching conditions are modeled via test transitions, which govern the activation of the corresponding LTI subsystems based on system states.

\begin{figure}[h!]
  \centering
  \includegraphics[scale=0.75]{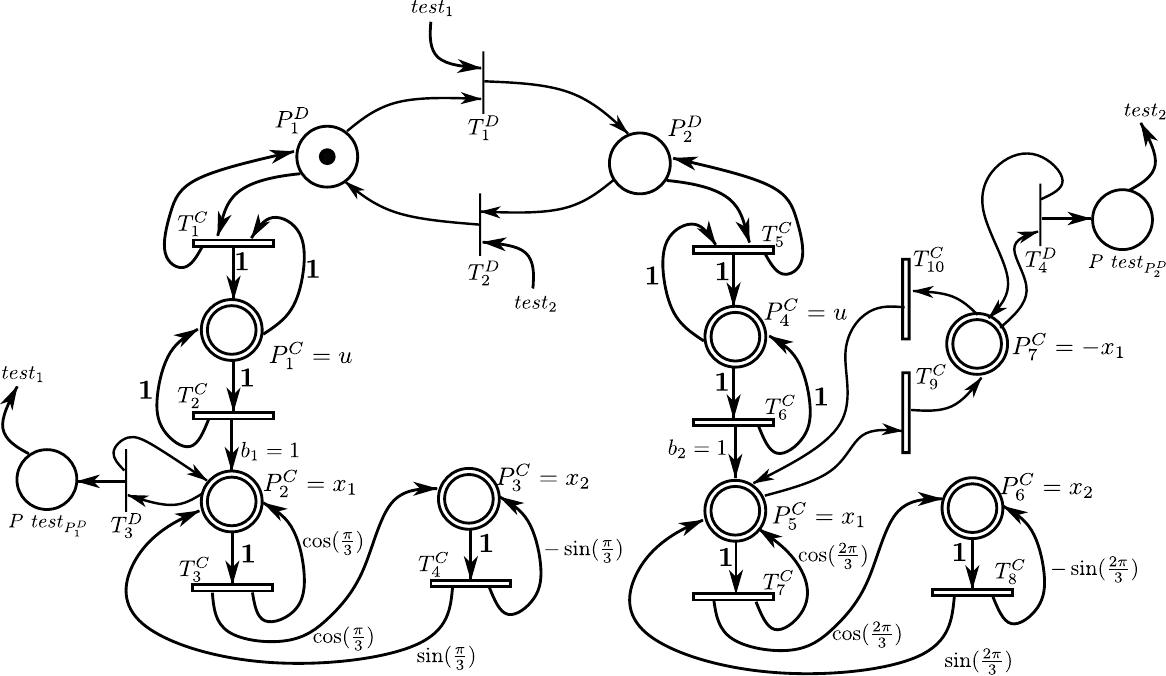}
  \caption{ETCPN representation of the hybrid system with two modes, illustrating discrete-event dynamics, continuous LTI subsystems, and switching logic within a unified modelling framework.}\label{fig10}
\end{figure}

\medskip

Based on this graphical representation and the proposed modelling architecture, Figure~\ref{fig:evolution} presents the simulation results that validate the ETCPN approach. The results capture the dynamic behaviour of the switched hybrid system, illustrating the evolution of the system input, continuous states, and discrete switching modes.

\begin{figure}[h!]
    \centering
    \subfloat[System input $U(k)$]{\includegraphics[scale=0.35]{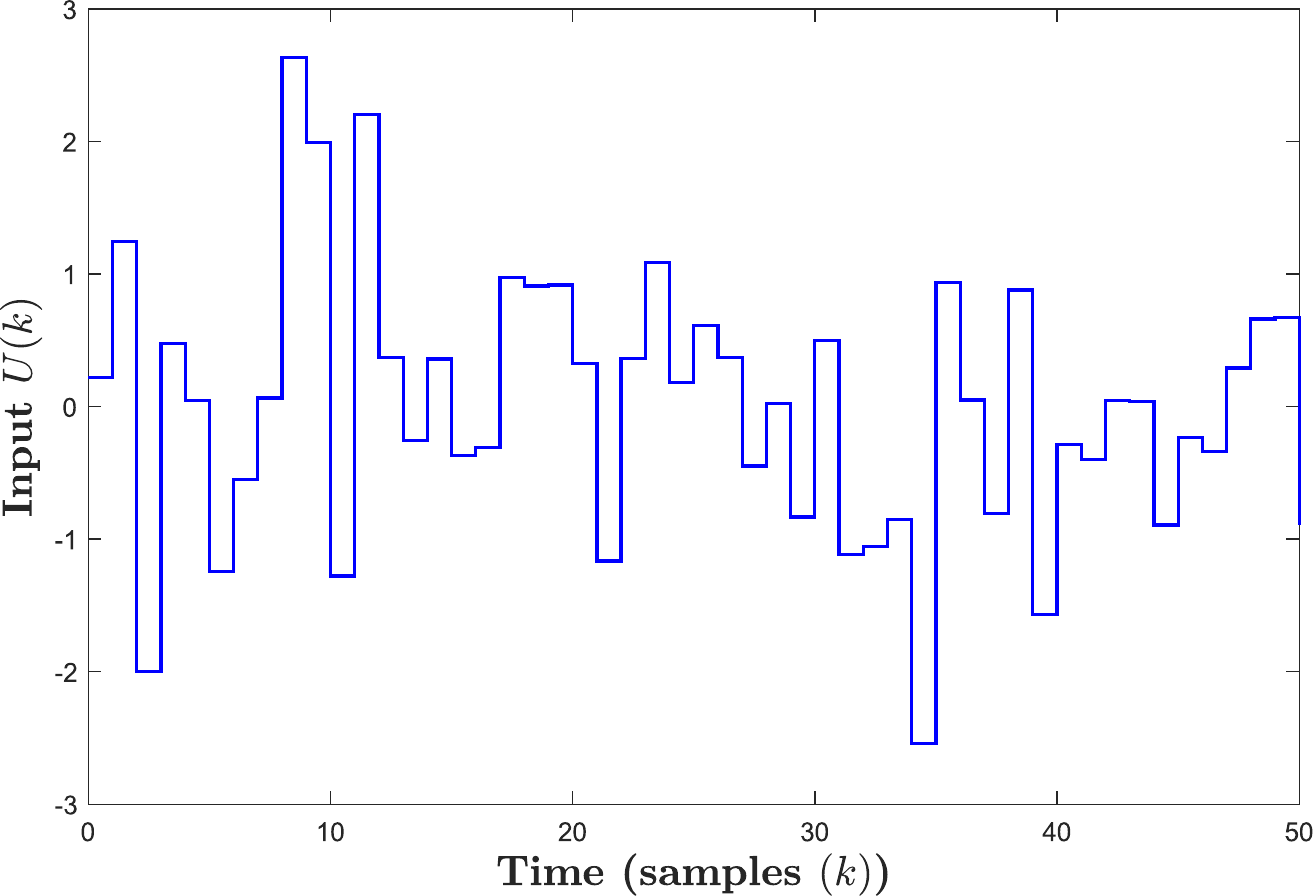}} \\
    \subfloat[ETCPN model states $x_1(k)$ and $x_2(k)$]{\includegraphics[scale=0.35]{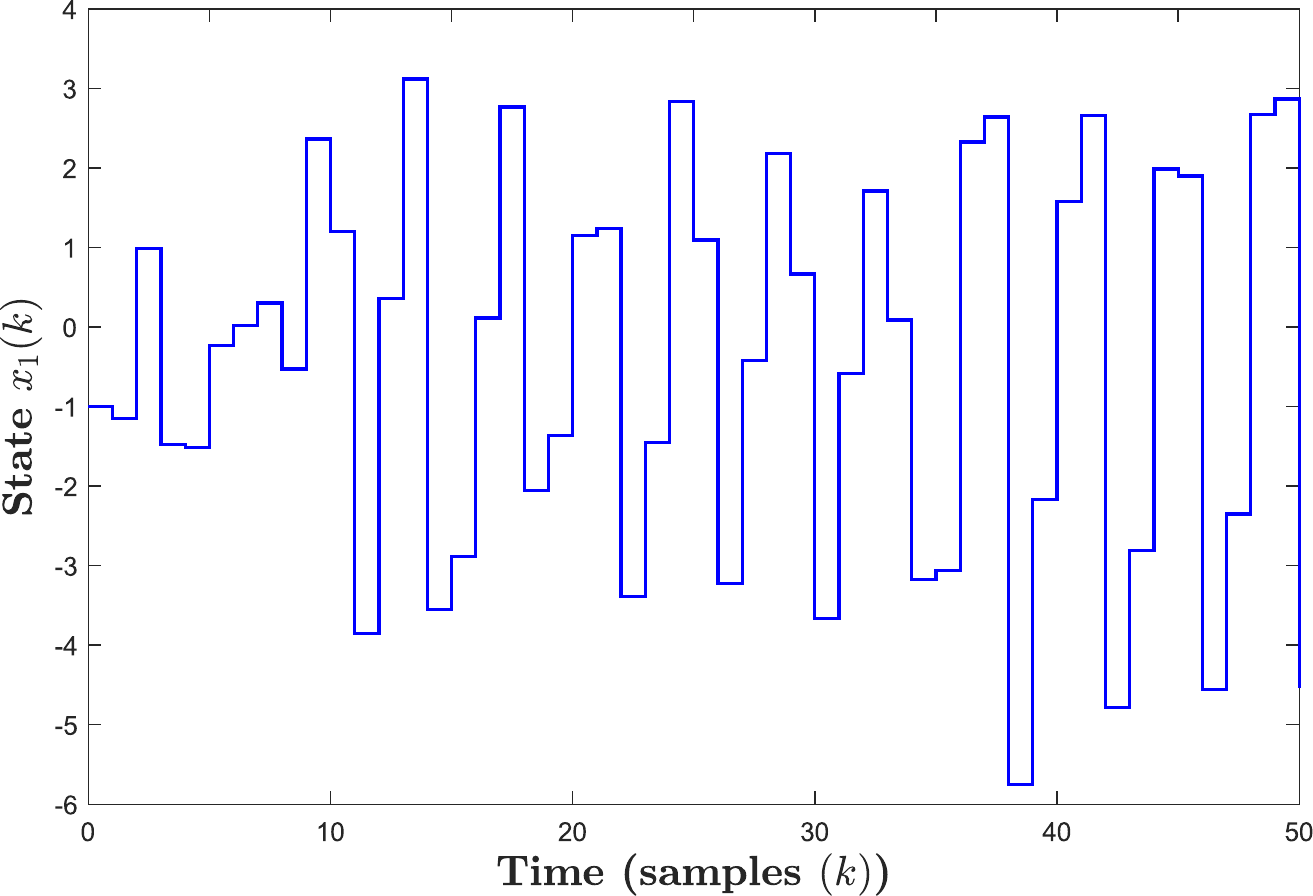}
    \includegraphics[scale=0.35]{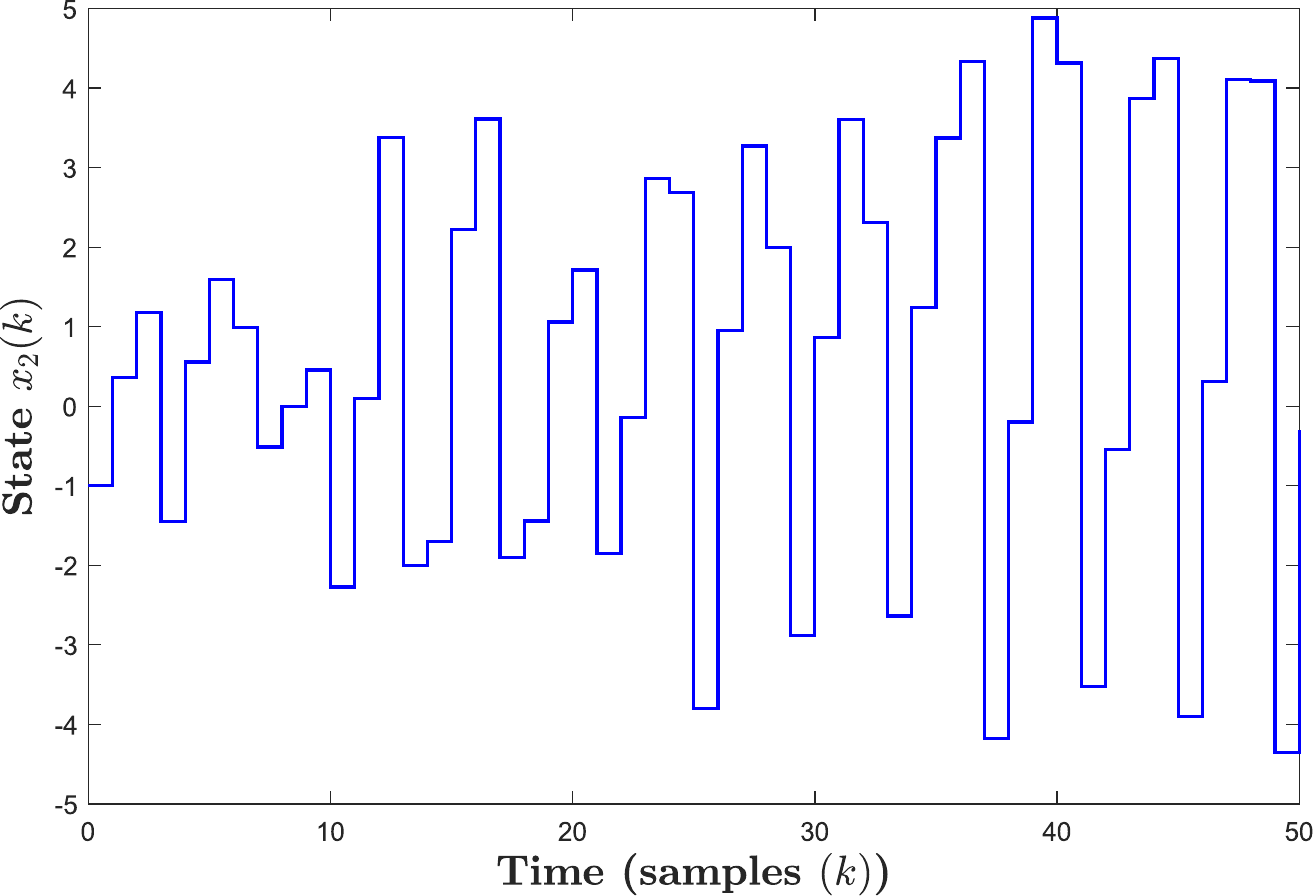}} \\
    \subfloat[ETCPN discrete modes $q(k)$]{\includegraphics[scale=0.35]{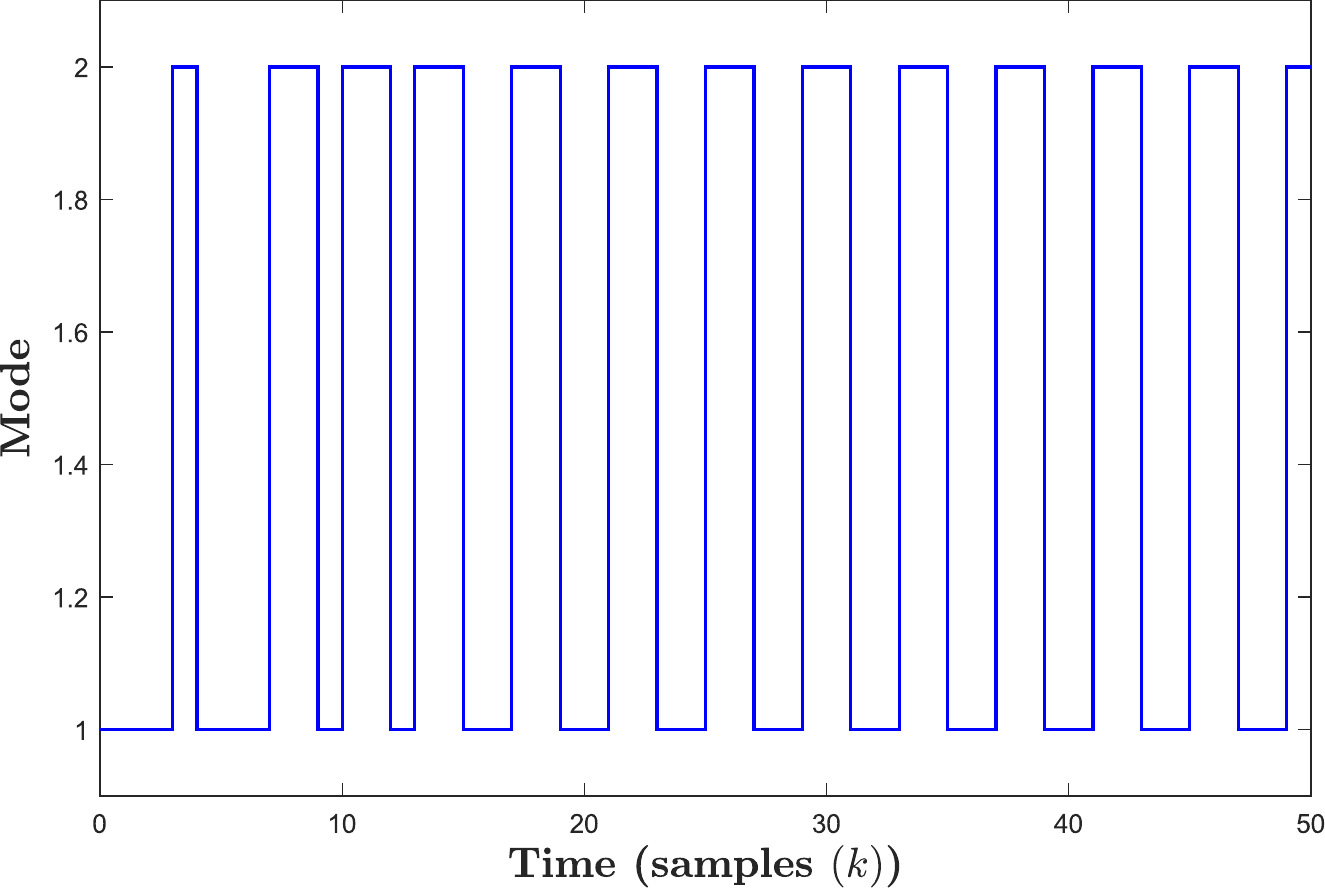}}
    \caption{Simulation results of the ETCPN model illustrating (a) evolution of the system input $U(k)$, (b) continuous model states $x_1(k)$ and $x_2(k)$, and (c) switching behaviour through discrete modes $q(k)$}
    \label{fig:evolution}
\end{figure}

\medskip
 
\noindent
To assess the effectiveness of the proposed ETCPN-based observer and to demonstrate its ability to accurately estimate system states under hybrid dynamic conditions, the observer gains were computed by solving the formulated LMI problem. The resulting observer gains corresponding to the two operating modes are:

\begin{equation*}
L_1 = 
\begin{bmatrix}
0.866 \\
0.5
\end{bmatrix}{\color{red}\,,}
\quad \text{and} \quad
L_2 = 
\begin{bmatrix}
0.866 \\
-0.5
\end{bmatrix}{\color{red}.}
\end{equation*}

\noindent
These gain matrices are applied in the observer structure to ensure stable and accurate estimation performance across the two switching modes of the hybrid discrete-time system.

\medskip

Figures~\ref{fig15} and~\ref{fig16} present the evolution of the proposed ETCPN switched observer, highlighting its performance in estimating both continuous and discrete states of the hybrid system. Figure~\ref{fig15} shows the accurate tracking of the continuous-time states by the observer, while Figure~\ref{fig16} illustrates the correct identification and estimation of the discrete system modes. The results confirm that the proposed observer ensures rapid convergence and high fidelity in following the actual system dynamics, demonstrating robustness across hybrid dynamics.

\begin{figure}[H]
     \centering
         \includegraphics[scale=0.35]{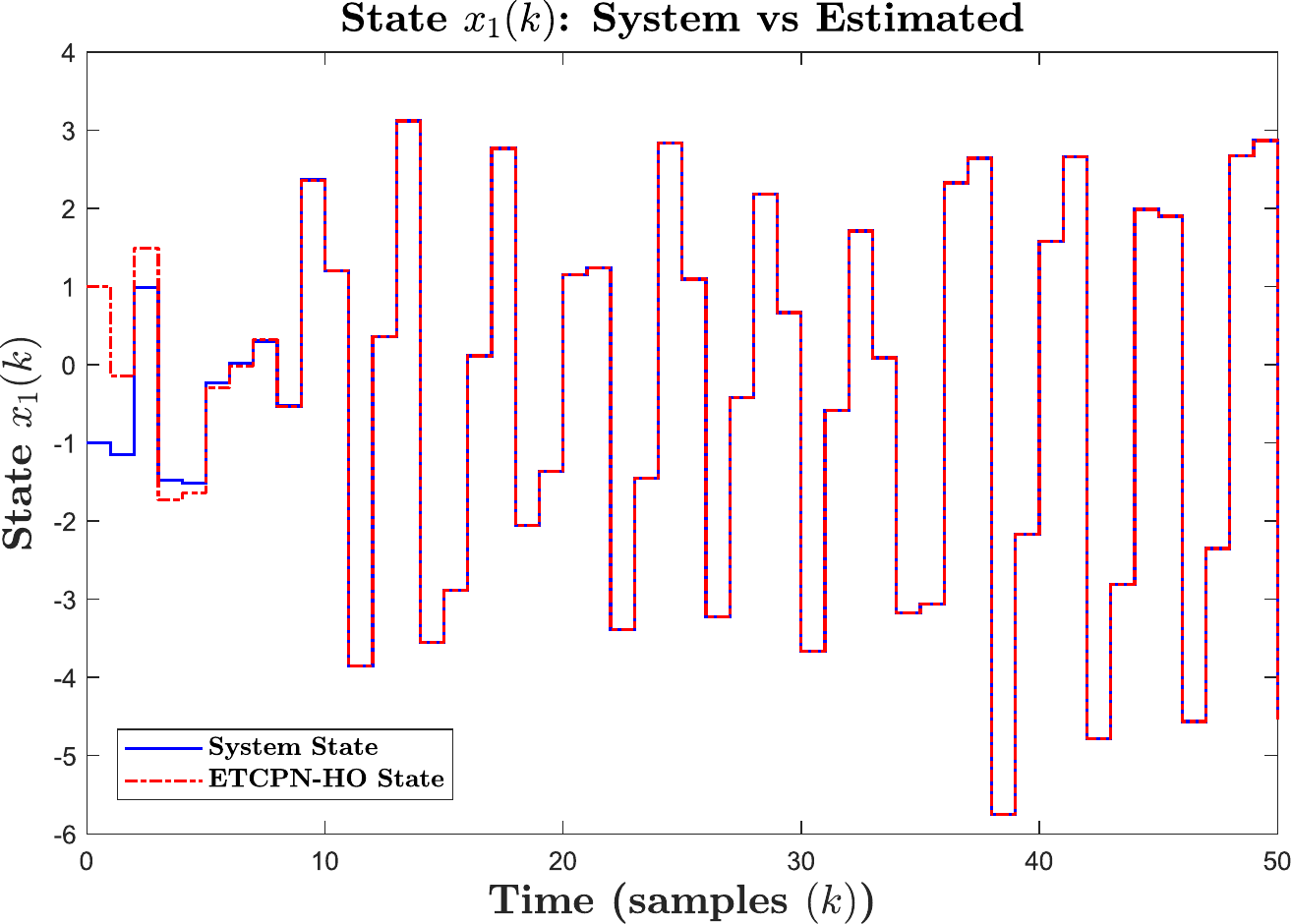}
         \includegraphics[scale=0.35]{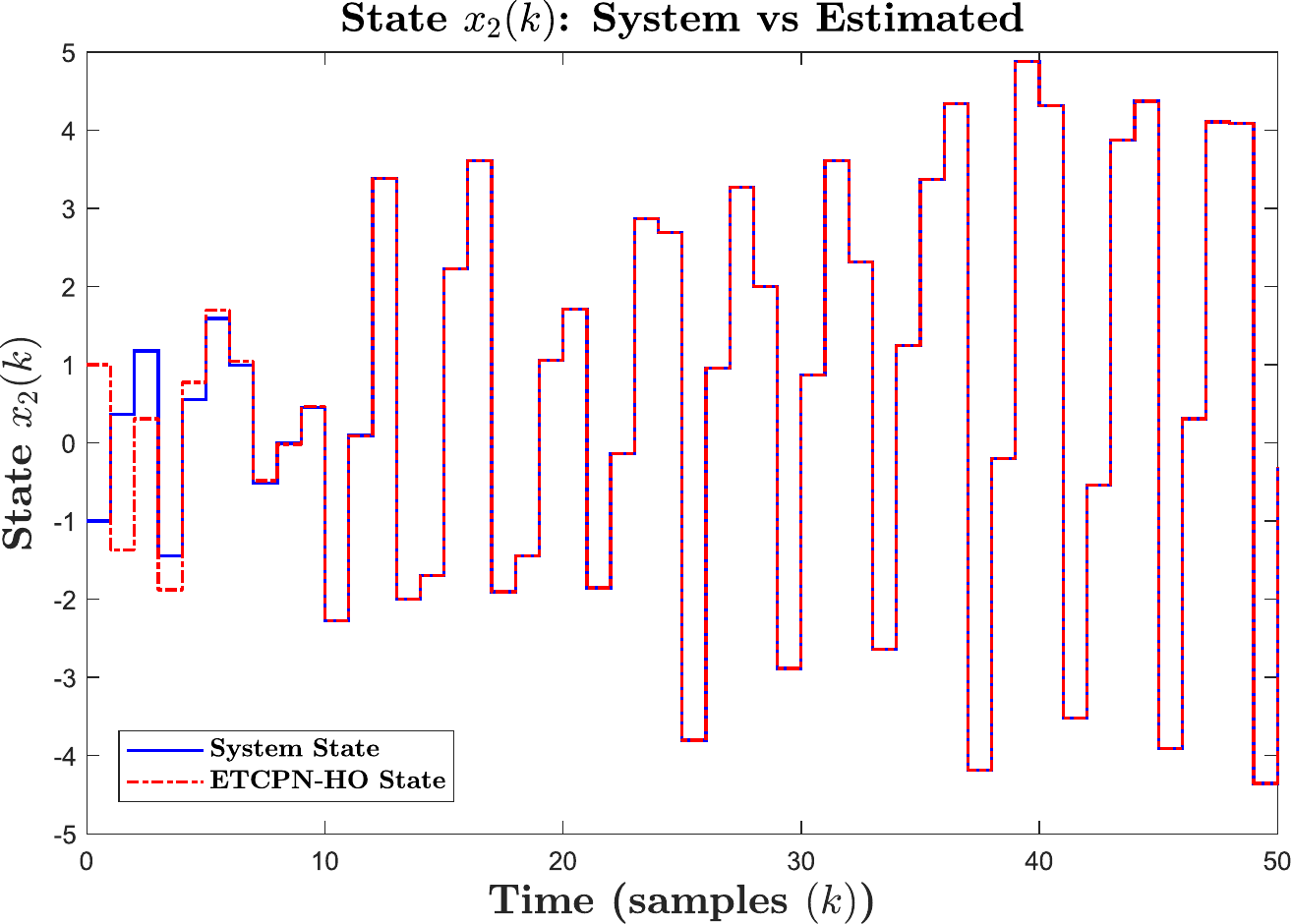}
     \caption{Evolution of the ETCPN switched observer states: estimated trajectories of $x_1(k)$ and $x_2(k)$ tracking the actual system states.}\label{fig15}
\end{figure}

\begin{figure}[H]
  \centering
  \includegraphics[scale=0.35]{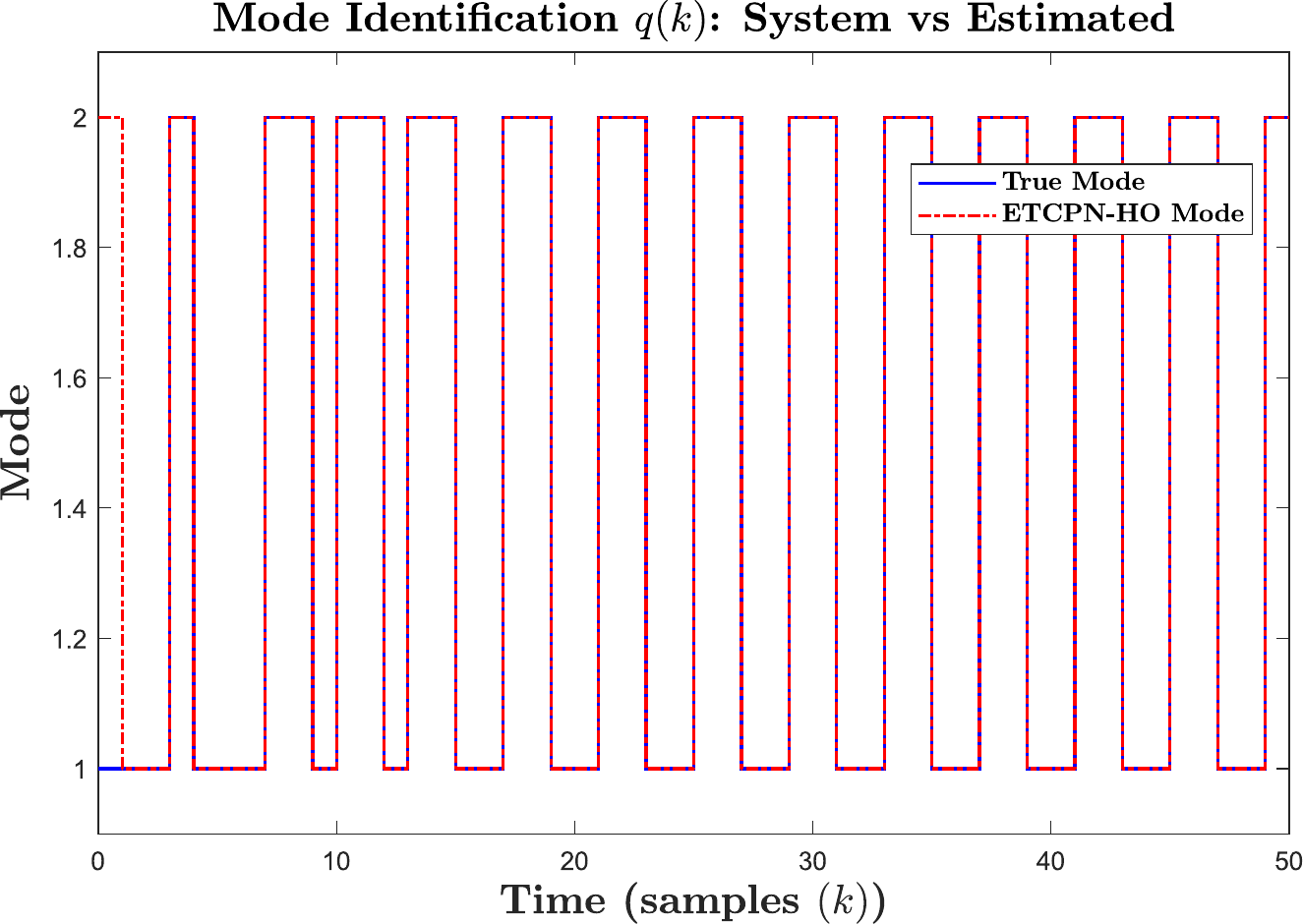}
  \caption{Evolution of the ETCPN switched observer modes $q(k)$: estimated discrete modes accurately capturing the system’s switching behavior.}\label{fig16}
\end{figure}

\medskip 
The effectiveness of the proposed observer design is further validated through the analysis of the estimation error dynamics. As shown in Figure~\ref{fig19}, the estimation errors exhibit rapid convergence to zero, which clearly demonstrates both the accuracy and stability of the observer. This fast convergence highlights the robustness of the ETCPN-based observer in faithfully reconstructing the system states and modes across hybrid dynamics.

\begin{figure}[H]
     \centering
         \includegraphics[scale=0.35]{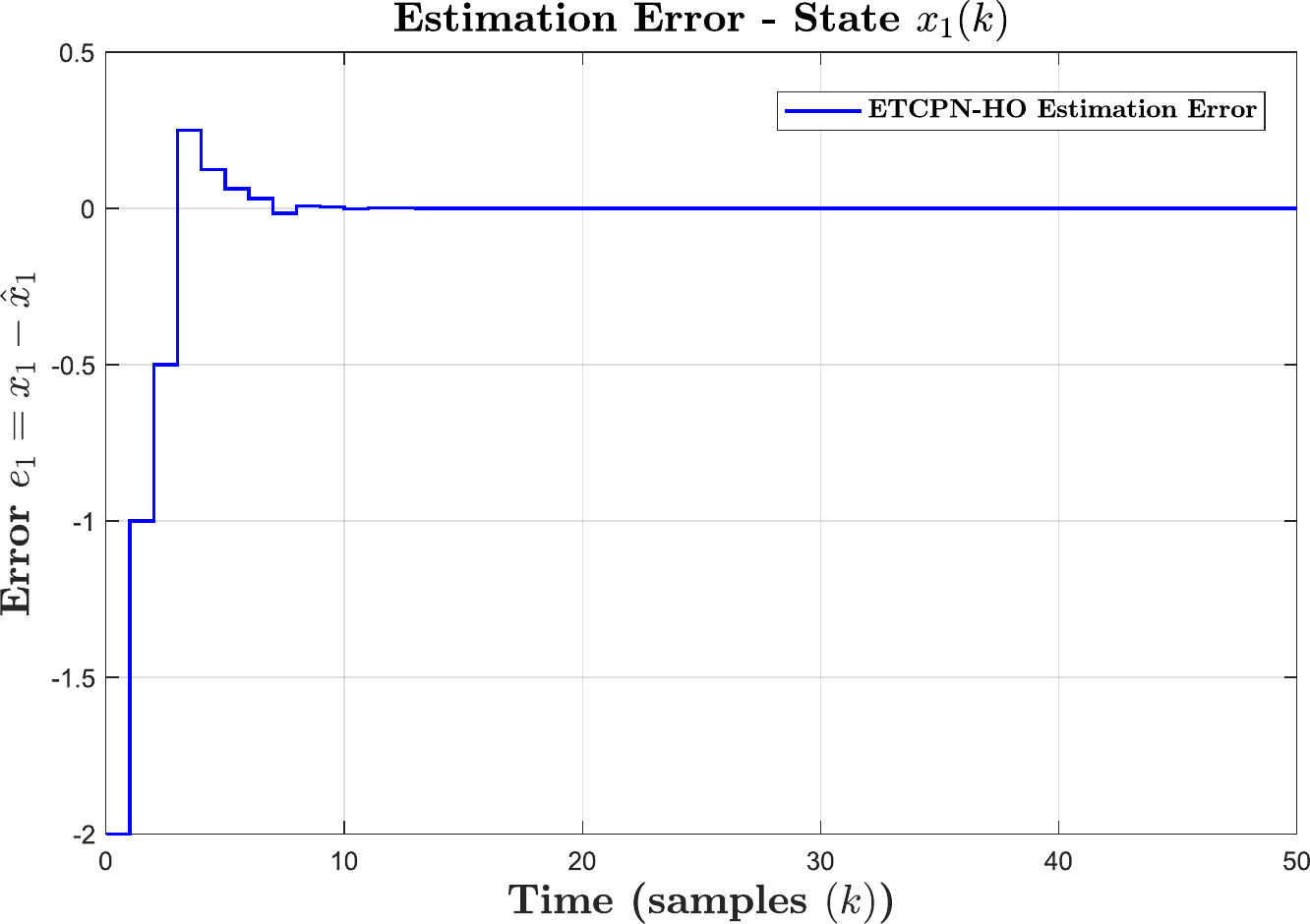}
         \includegraphics[scale=0.35]{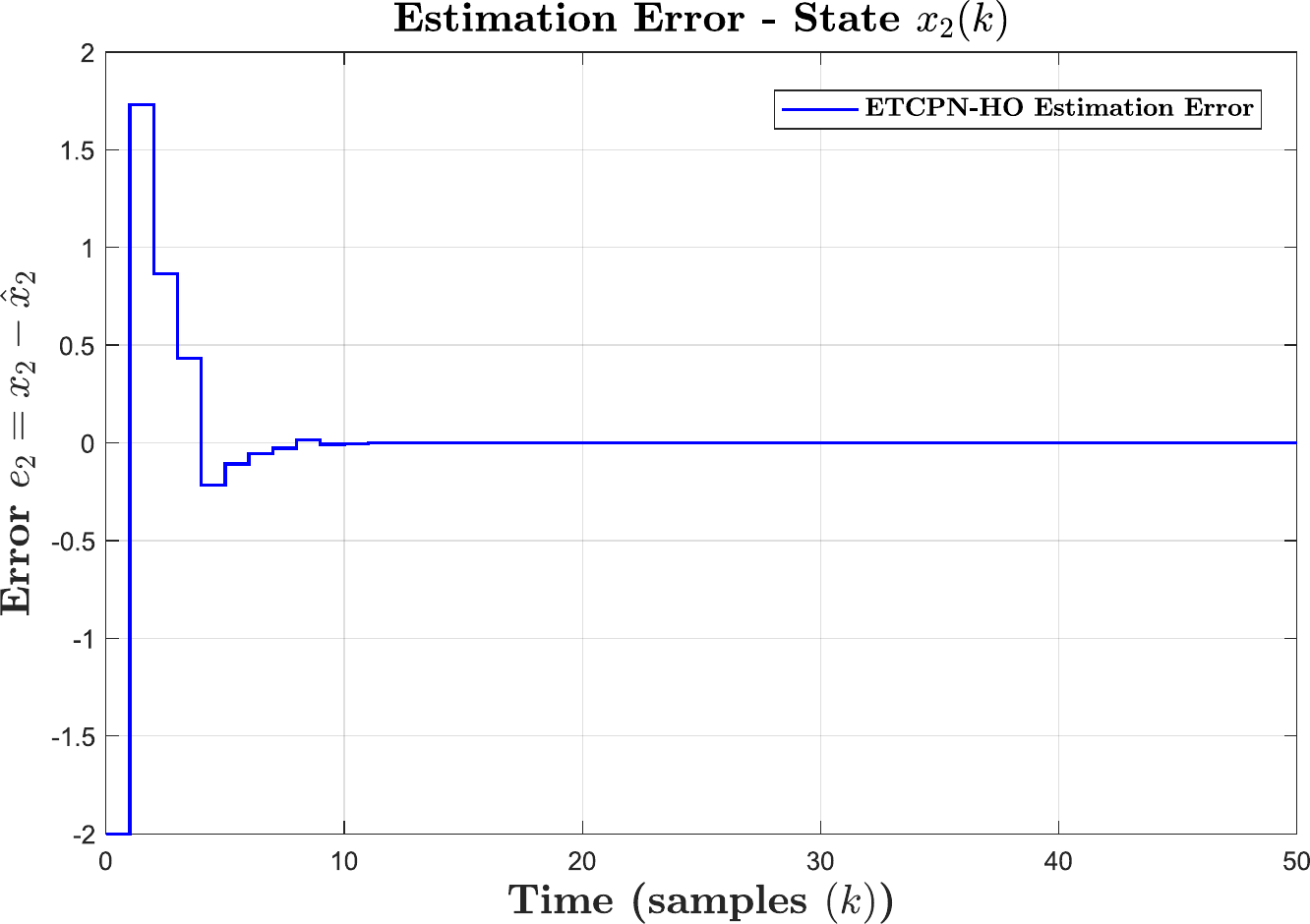}
     \caption{Evolution of estimation error dynamics: observer errors rapidly converging to zero, indicating accurate and stable hybrid state estimation.}\label{fig19}
\end{figure}

\medskip

Finally, to further assess the fault detection capabilities of the proposed ETCPN modeling and observer framework, three representative fault scenarios are investigated. The first involves faults affecting the discrete-event dynamics, the second focuses on faults impacting the continuous-time subsystem, and the third scenario integrates concurrent faults that affect both domains simultaneously. The performance of the proposed ETCPN-HO is benchmarked against a state of the art Hybrid Automaton-based Observer (HA-OB) within these scenarios. To implement the semi-supervised ETCPN based framework, the training data for anomaly detection algorithms (OC-SVM, SVDD, and Elliptic Envelope) was generated by simulating the nominal (fault-free) ETCPN model. The residuals $ r(k) $, calculated as the discrepancy between the measured outputs and the ETCPN-HO's estimates under these nominal conditions, served as the feature vectors for training. This dataset exclusively represents normal system behavior, ensuring the detectors are trained to recognize the system's baseline operational state without any exposure to fault data. 

The following sections detail the fault injection process, residual generation, and detection results, thereby illustrating the practical diagnostic effectiveness of the proposed scheme.

\paragraph{\textbf{Case 1 - Faults Affecting the Discrete-Event Part:}}

In this scenario, faults are deliberately introduced into the discrete-event component of the system, specifically through mode blocking events that interfere with normal switching logic. Such faults prevent mode transitions even when switching conditions are satisfied, thereby disrupting the expected hybrid system behavior. Two representative fault periods are considered to illustrate this effect: the first occurs during the interval $k\in[13,~17]$, where the system is forcibly held in mode~1, and the second during $k\in[30,~34]$, where it remains trapped in mode~2 (Figure~\ref{figf2}).

These injected faults enable a detailed investigation of how discrete-event anomalies propagate through the hybrid dynamics and challenge the observer's estimation capabilities. The results in Figure~\ref{figf1} demonstrate a key advantage of the ETCPN-HO, its rapid convergence following a fault event. Once the blocking fault ends and normal switching resumes, the ETCPN-HO quickly re-synchronizes with the true system state. In contrast, the HA-OB exhibits slower recovery and prolonged estimation error, which can lead to false alarms. This superior performance is attributed to the ETCPN's intrinsic, native representation of the system's concurrent switching logic, which enables more robust state estimation during these critical transition periods. This robustness is further emphasized in Figure~\ref{figf3}, where the ETCPN-HO estimation errors show sharp deviations during fault periods followed by a swift return to zero, underscoring the framework's responsiveness.

\begin{figure}[H]
     \centering
         \includegraphics[scale=0.35]{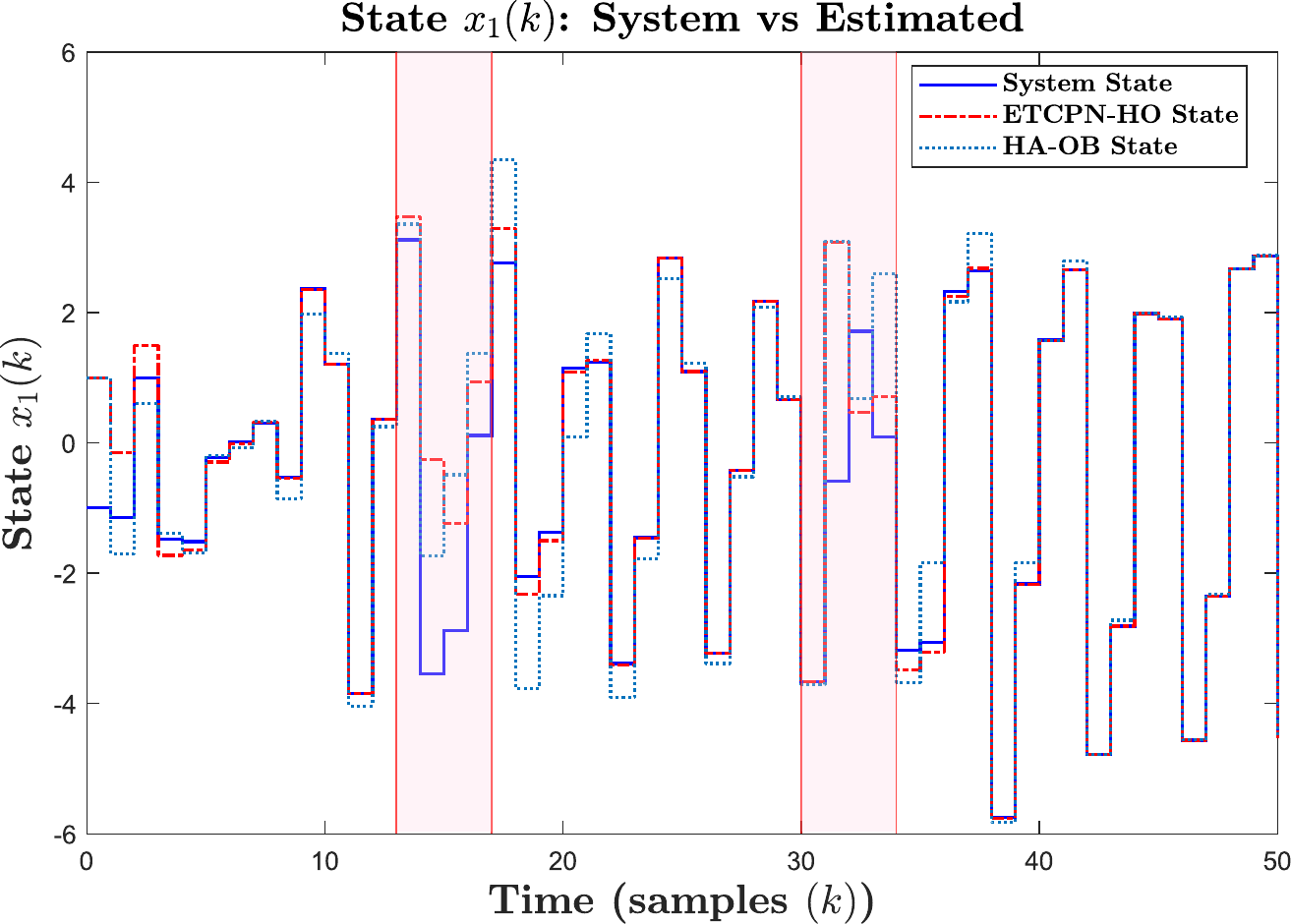}
         \includegraphics[scale=0.35]{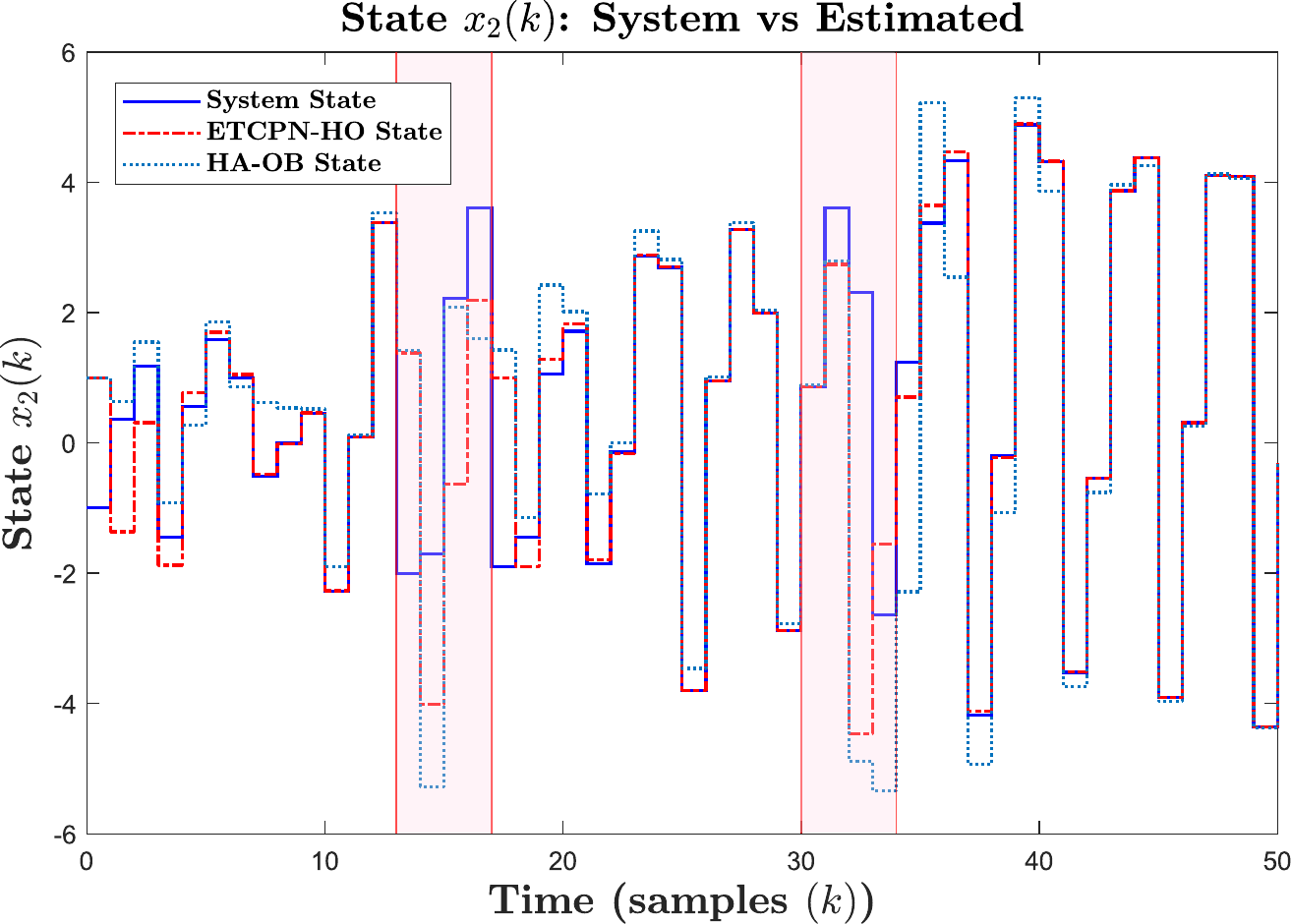}
     \caption{Evolution of ETCPN switched observer states during discrete-event faults.}\label{figf1}
\end{figure}

\begin{figure}[H]
  \centering
  \includegraphics[scale=0.35]{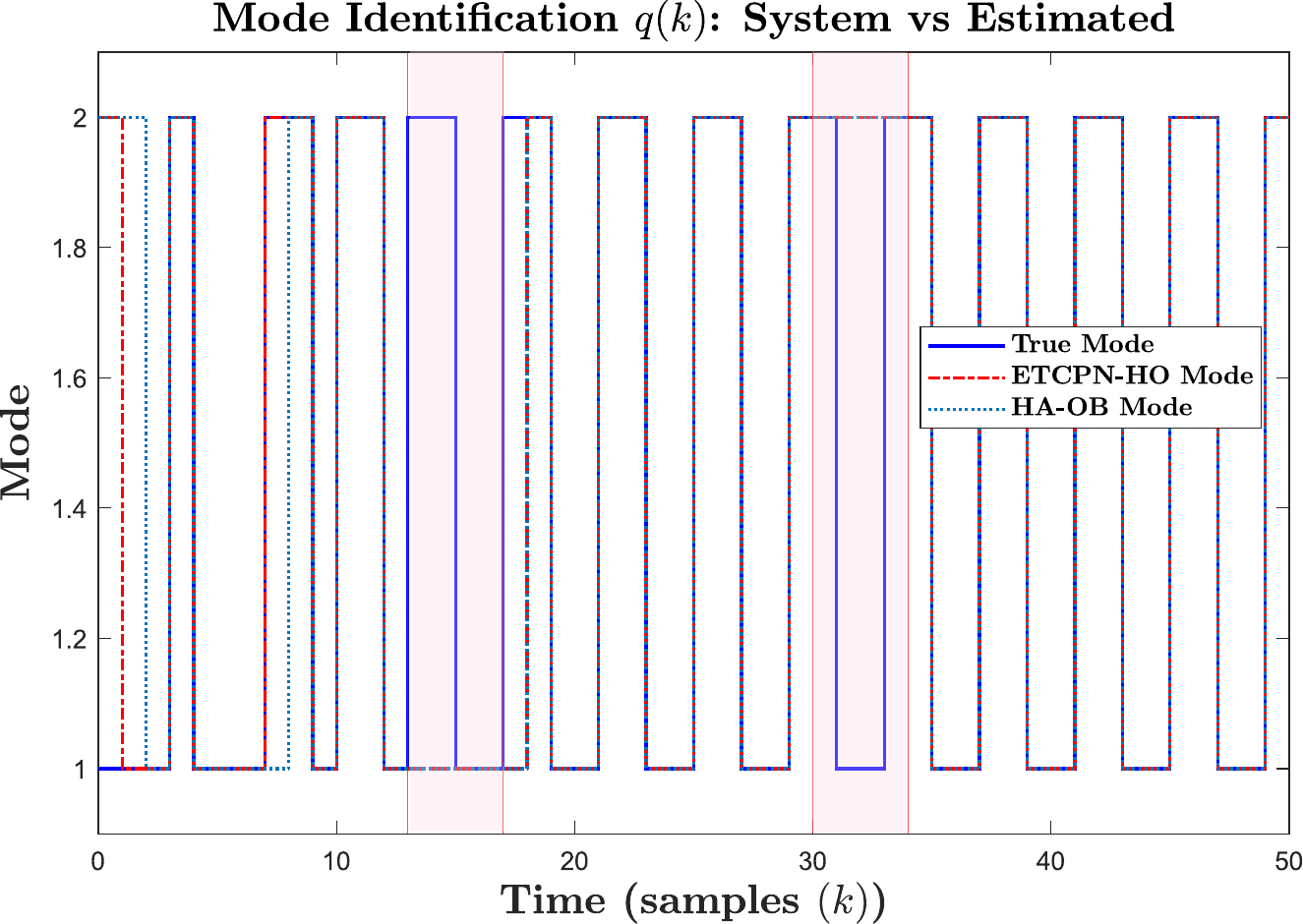}
  \caption{Evolution of ETCPN switched observer modes $q(k)$ illustrating mode blocking faults.}\label{figf2}
\end{figure}

\begin{figure}[H]
     \centering
         \includegraphics[scale=0.35]{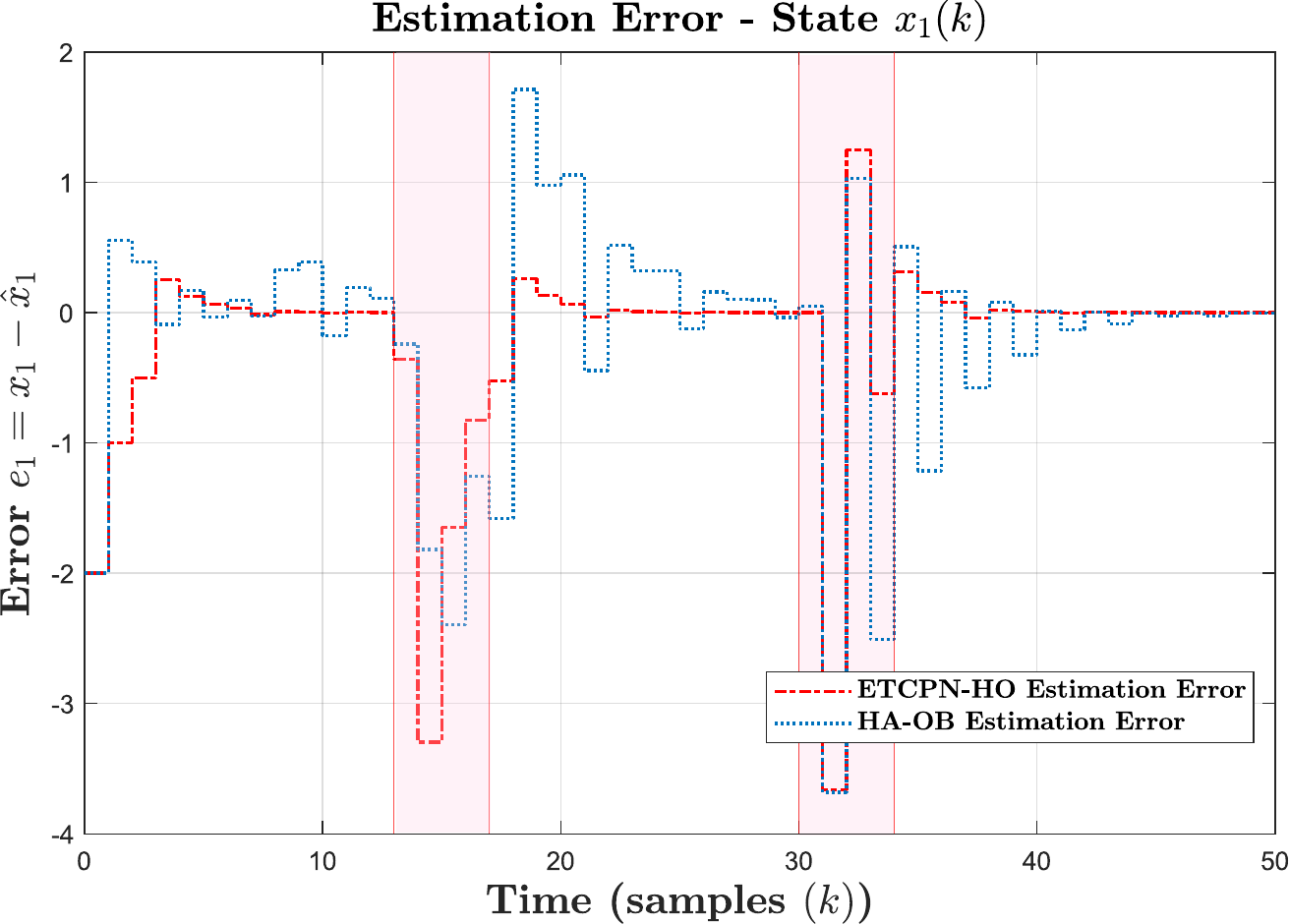}
         \includegraphics[scale=0.35]{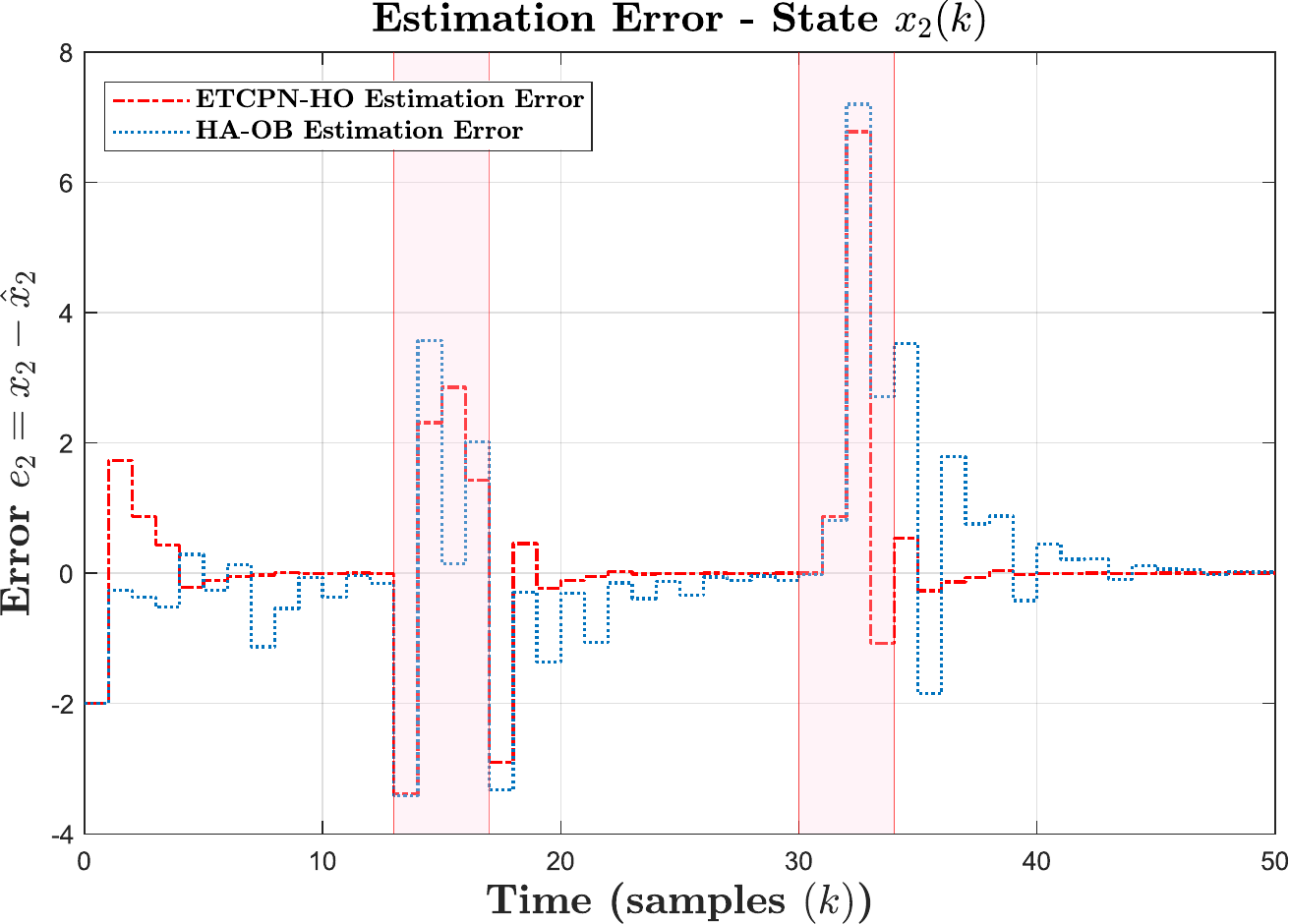}
     \caption{Estimation error dynamics highlighting fault detection and post-fault recovery.}\label{figf3}
\end{figure}

\medskip 

Following the fault injection described above, residuals were generated using the ETCPN hybrid observer and analyzed through three semi-supervised anomaly detection techniques, namely OC-SVM~\cite{scholkopf2001estimating}, SVDD~\cite{tax2004support}, and Elliptic Envelope (EE)~\cite{rousseeuw1999fast}. These models were trained exclusively on fault-free data to establish the nominal residual distributions, thereby eliminating the requirement for labelled fault samples. During the training phase, the detection models were configured as follows. For OC-SVM, the radial basis function (RBF) kernel was employed with a parameter $\nu=0.12$, while the kernel coefficient $\gamma$ was left at its default setting. In the case of SVDD, the RBF kernel was also used, with identical $\nu=0.12$ and a kernel coefficient $\gamma=0.1$. Finally, for the Elliptic Envelope model, the contamination parameter was set to $0.05$ under the assumption of Gaussian-distributed fault-free residuals.

\medskip 

The detection capability of each method was assessed based on the computed residuals using four evaluation metrics: Recall, False Positive Rate (FPR), F1 Score, and overall accuracy. Table~\ref{tab:case1_detection} presents the quantitative results obtained for discrete-event faults (Case 1).

\begin{table}[H]
\centering
\caption{Fault detection results using residual-based methods for discrete-event faults (Case 1).}
\label{tab:case1_detection}
\begin{tabular}{lcccc}
\hline
\textbf{Method} & \textbf{Accuracy} & \textbf{Recall} & \textbf{FPR} & \textbf{F1 Score} \\
\hline
OC-SVM & 0.911 & 0.800 & 0.057 & 0.80 \\
EE     & 0.844 & 0.700 & 0.114 & 0.67 \\
SVDD   & 0.889 & 0.900 & 0.114 & 0.78 \\
\hline
\end{tabular}
\end{table}

\medskip 

The results clearly demonstrate that all three approaches successfully detected the injected discrete-event faults. One-Class SVM and SVDD, in particular, achieved the highest recall rates (0.90), reflecting strong sensitivity to mode-blocking anomalies. Their F1 Scores were also relatively high (0.80 and 0.78, respectively), suggesting a balanced detection performance with limited false alarms. Although Elliptic Envelope performed reasonably well, it exhibited lower recall and F1 Score, indicating reduced effectiveness in capturing the discrete-event fault patterns, likely due to its reliance on global Gaussian assumptions which may not fully capture the local dynamics induced by mode blocking.  Overall, One-Class SVM and SVDD proved to be highly reliable detectors for discrete-event faults in the proposed ETCPN framework, offering an effective trade-off between sensitivity and false alarm rates in the absence of labelled fault data.

\paragraph{\textbf{ Case 2 -  Faults affecting the continuous-time dynamics:}}

To further validate the proposed framework, this case investigates the detection of faults that affect the continuous discrete-time dynamics of the hybrid system. Specifically, intermittent output sensor errors were introduced, characterized by signal distortions occurring during the time intervals $k \in [5,10]$ and $k \in [25,30]$. The injected fault pattern is illustrated in Figure~\ref{figu}.

\begin{figure}[h!]
  \centering
  \includegraphics[scale=0.35]{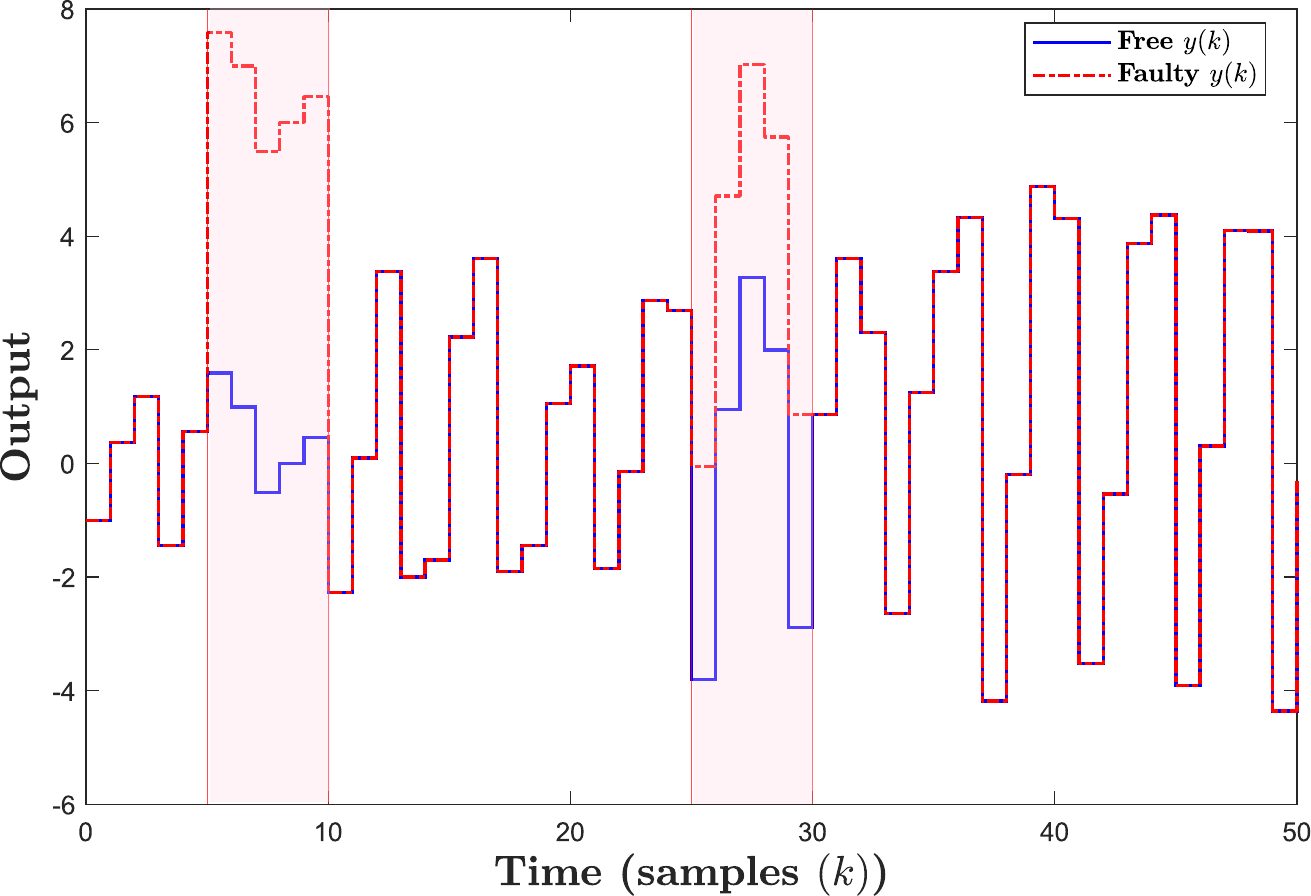}
  \caption{Injected intermittent output faults affecting continuous dynamics.}
  \label{figu}
\end{figure}

The results confirm the superior capability of the proposed ETCPN-HO in detecting and tracking faults within the continuous discrete time state dynamics. As shown in Figure~\ref{figxf1}, the ETCPN-HO demonstrates faster and more accurate convergence following fault injection compared to the HA-OB, which exhibits slower recovery and larger, more persistent estimation errors. This effective response is attributed to the tighter intrinsic coupling between discrete and continuous dynamics in the ETCPN architecture, enabling more robust state correction when confronted with faulty sensor outputs. Furthermore, Figure~\ref{figmf} highlights the resilience of the ETCPN-HO's discrete mode observer, which maintains accurate mode estimation despite the continuous sensor faults. In contrast, the HA-OB struggles to recover and converge to the actual system mode. The estimation error dynamics in Figure~\ref{error3} provide further confirmation of the ETCPN-HO's efficiency, with its residuals quickly decaying to near zero after each fault. The HA-OB, however, shows significant residual error throughout and after the fault periods.

Collectively, these results illustrate the effectiveness of the ETCPN-HO for resilient monitoring of hybrid dynamic systems affected by intermittent sensing faults. Overall, for continuous-domain faults, the ETCPN-HO framework demonstrates clear advantages over the HA-OB approach, providing more accurate state estimates and higher-quality residuals that enable more precise and reliable fault detection.

\begin{figure}[H]
  \centering
  \includegraphics[scale=0.35]{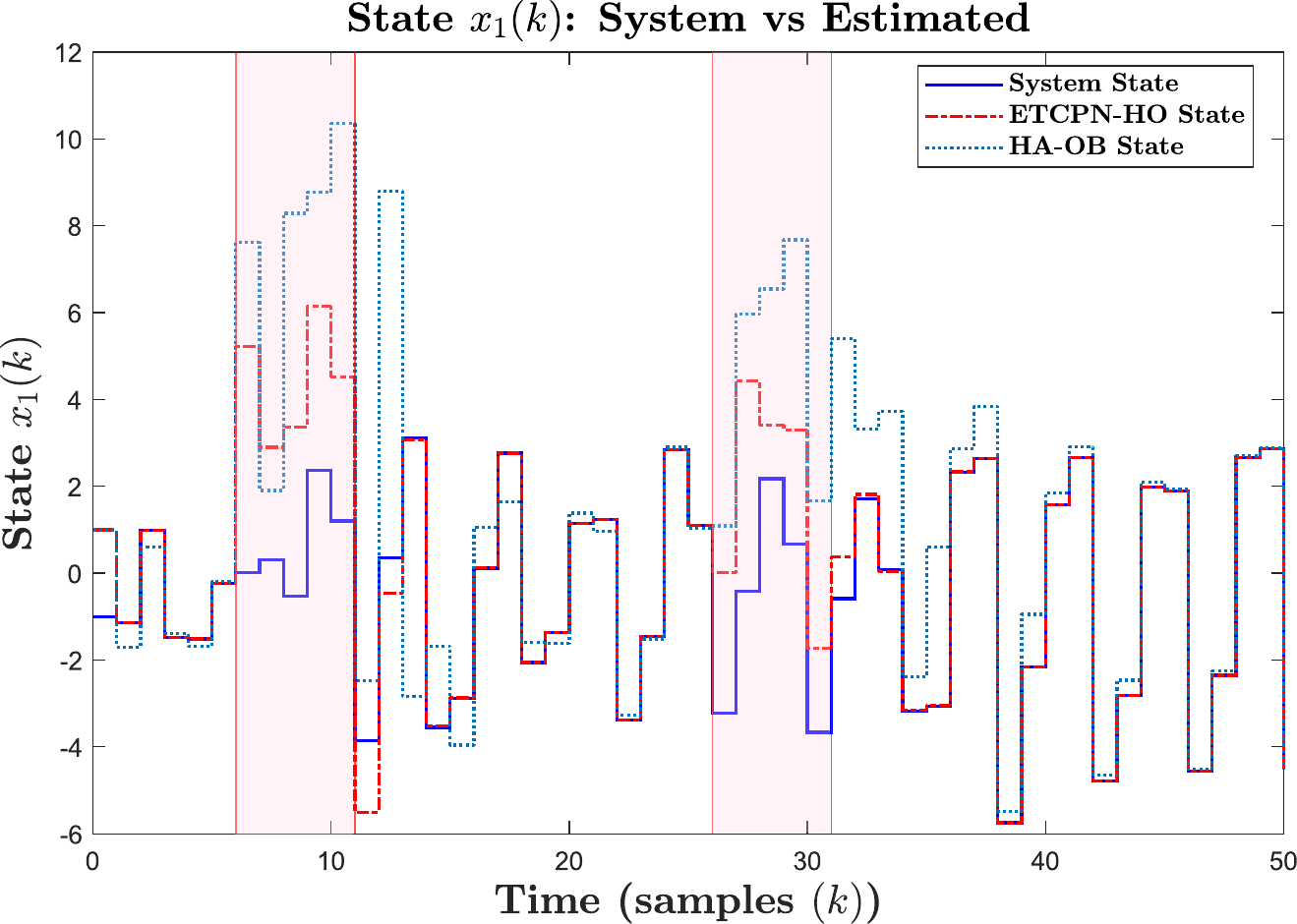}
  \includegraphics[scale=0.35]{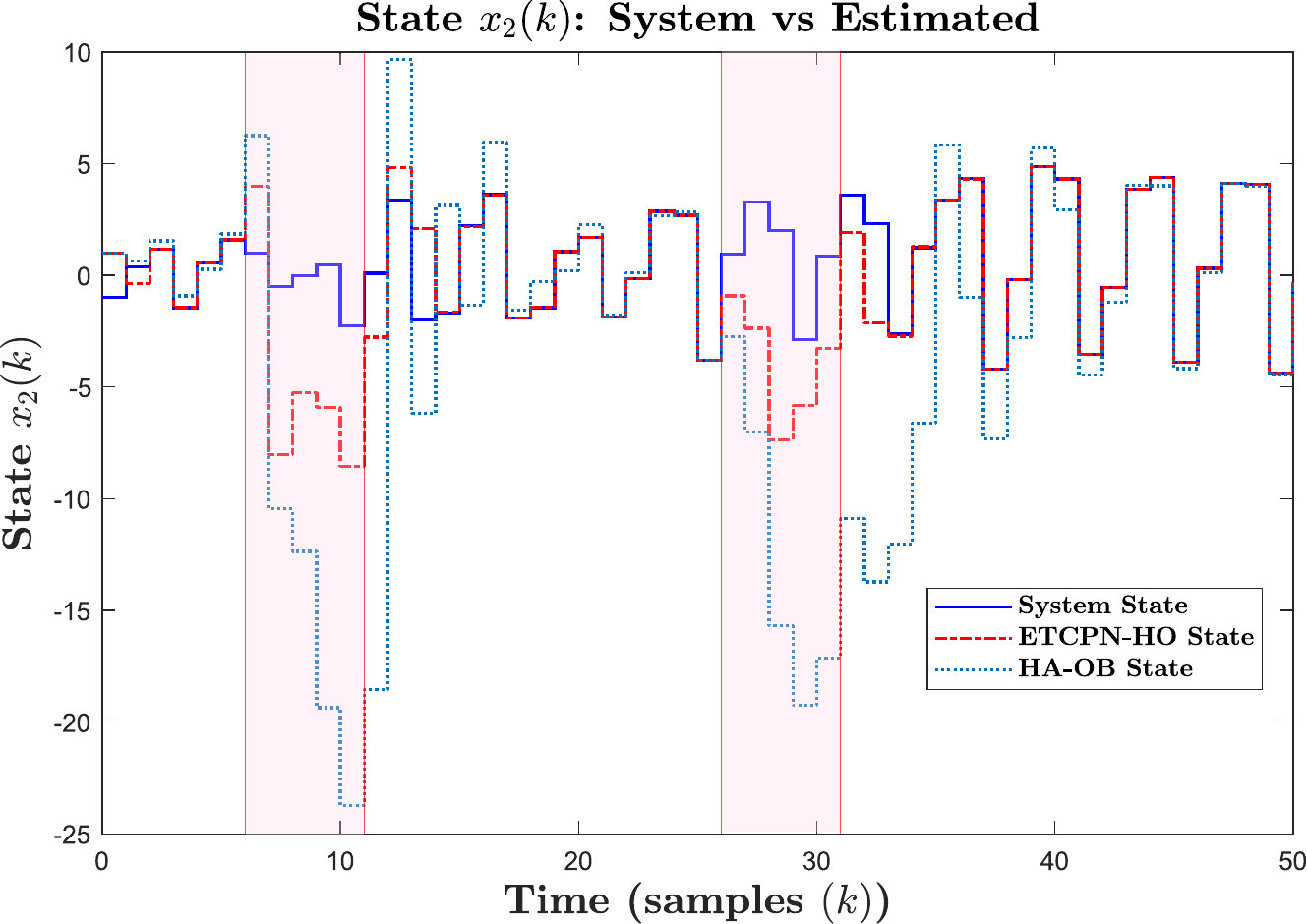}
  \caption{Evolution of ETCPN switched observer states under intermittent output faults.}
  \label{figxf1}
\end{figure}

\begin{figure}[H]
  \centering
  \includegraphics[scale=0.35]{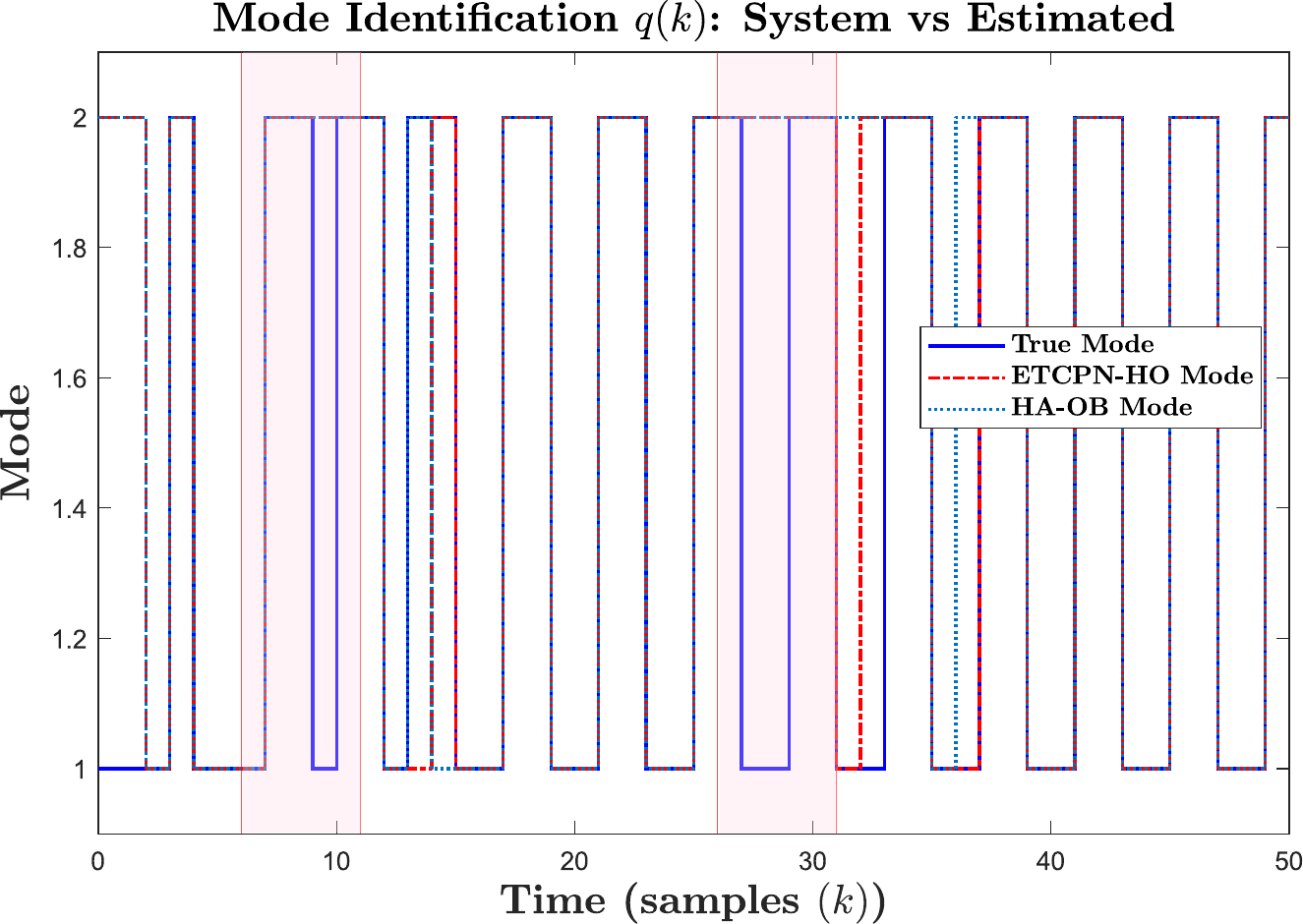}
  \caption{Evolution of ETCPN switched observer modes $q(k)$ during fault occurrences.}
  \label{figmf}
\end{figure}

\begin{figure}[H]
  \centering
  \includegraphics[scale=0.35]{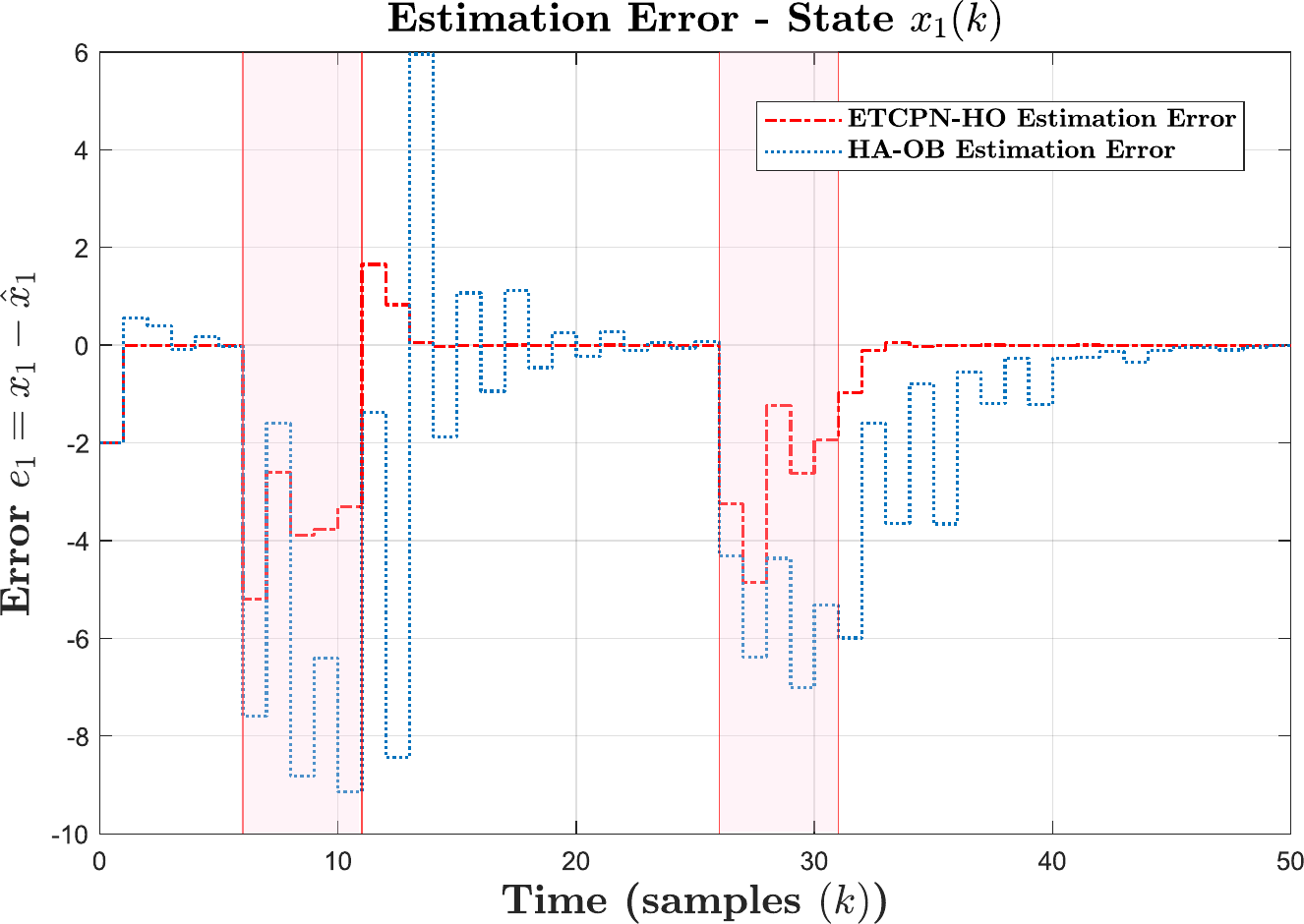}
  \includegraphics[scale=0.35]{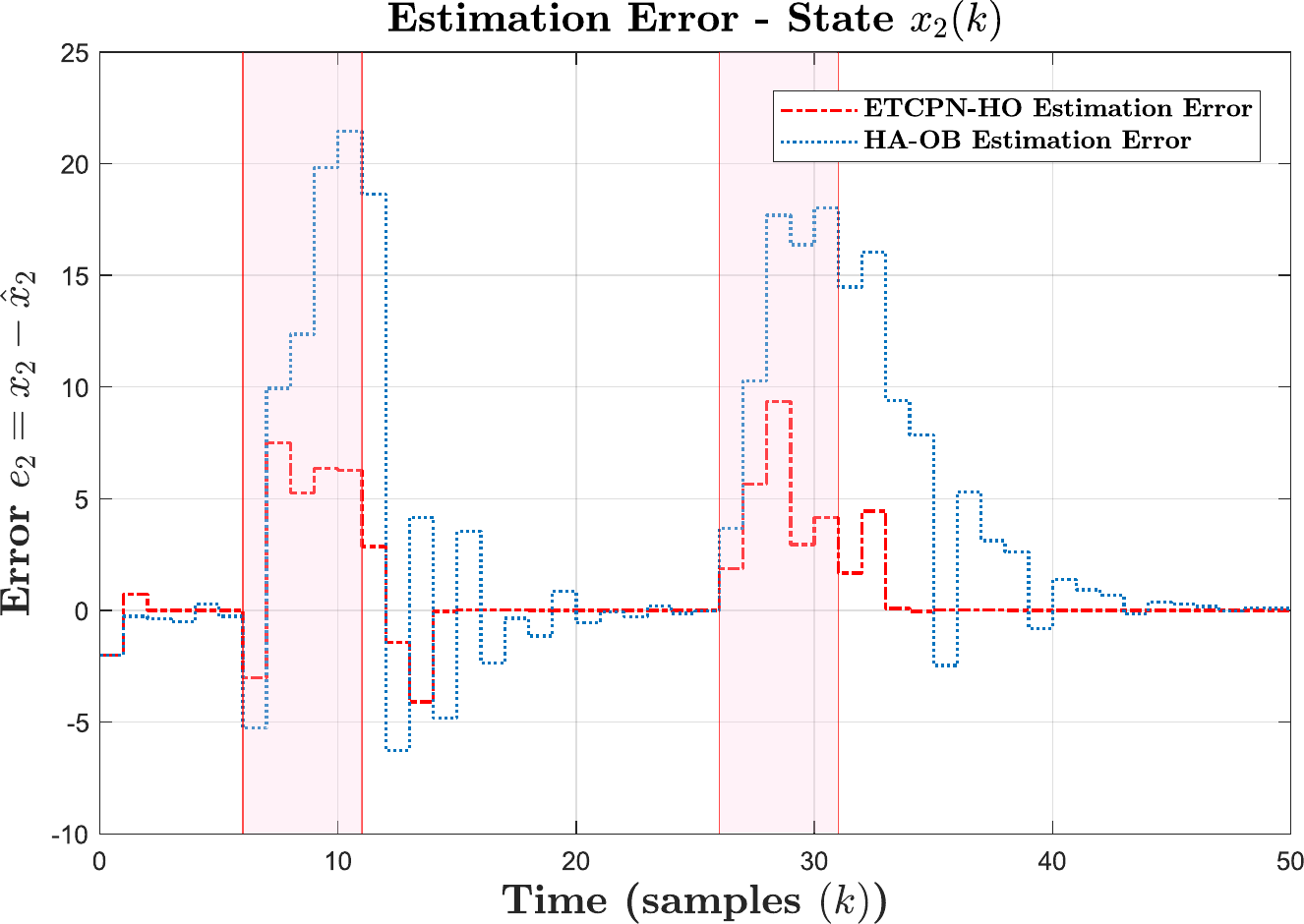}
  \caption{Evolution of estimation error dynamics during continuous faults.}
  \label{error3}
\end{figure}

In this scenario, faults affecting the continuous dynamics were considered, specifically output sensor faults introduced intermittently as described previously. The generated residuals using ETCPN-HO were analyzed through the three anomaly detection methods: One-Class SVM (OC-SVM), Elliptic Envelope (EE), and Support Vector Data Description (SVDD). The performance of each method was quantitatively evaluated using four key metrics: Recall, False Positive Rate (FPR), F1 Score, and overall accuracy. The detection results for Case 2 are summarized in Table~\ref{tab:case2_detection}.

\begin{table}[H]
\centering
\caption{Fault detection results for Case 2 (continuous-time faults) using residual-based methods.}
\label{tab:case2_detection}
\begin{tabular}{lcccc}
\hline
\textbf{Method} & \textbf{Accuracy} & \textbf{Recall} & \textbf{FPR} & \textbf{F1 Score} \\
\hline
One-Class SVM & 0.943 & 1.000 & 0.057 & 0.83 \\
EE            & 0.911 & 1.000 & 0.121 & 0.86 \\
SVDD          & 0.889 & 1.000 & 0.152 & 0.83 \\
\hline
\end{tabular}
\end{table}

\noindent The results clearly demonstrate the effectiveness of all three detection approaches in identifying continuous-time faults. Notably, all methods achieved perfect fault detection rates, as evidenced by Recall values of $1.00$, confirming their ability to successfully detect all fault occurrences in Case 2. This is attributed to the pronounced residual deviations induced by the intermittent sensor faults, which provided clear fault signatures easily distinguishable from the nominal behavior.

However, differences were observed in terms of false alarm rates (FPR). The One-Class SVM exhibited the lowest FPR (0.057), indicating a better trade-off between detection sensitivity and false positives. SVDD and Elliptic Envelope produced slightly higher FPR values of 0.152 and 0.121, respectively, which may be explained by their stricter decision boundaries or distributional assumptions, especially in the presence of noisy or borderline samples. In terms of overall detection quality (F1 Score), Elliptic Envelope achieved the highest F1 score (0.86), balancing perfect recall with an acceptable false alarm level. OC-SVM and SVDD followed closely, with F1 scores of 0.83, reflecting robust and consistent performance.

\paragraph{\textbf{ Case 3 -  Faults affecting both discrete event and continuous-time dynamics:}}

To comprehensively evaluate the effectiveness of the proposed ETCPN-based detection framework, we investigate a complex scenario involving concurrent faults affecting both the continuous dynamics and the discrete-event dynamics. Specifically, we introduce intermittent sensor faults in the continuous output and mode-blocking faults that interfere with the normal discrete switching logic. The fault injection protocol is designed to analyze both separate and simultaneous occurrences:
\begin{itemize}
  \item Separate Faults: A continuous sensor fault is active in the time interval $k \in [5, 10]$, followed by a discrete mode-blocking fault that forces the system to remain in Mode 1 during $k \in [20, 24]$.
  \item Simultaneous Faults: Both fault types are activated concurrently within the same time interval $k \in [37, 41]$, , where the discrete fault traps the system in Mode 2 while the continuous sensor fault is active.
\end{itemize}

The injected fault pattern is illustrated in Figure~\ref{figud}.

\begin{figure}[h!]
\centering
\includegraphics[scale=0.35]{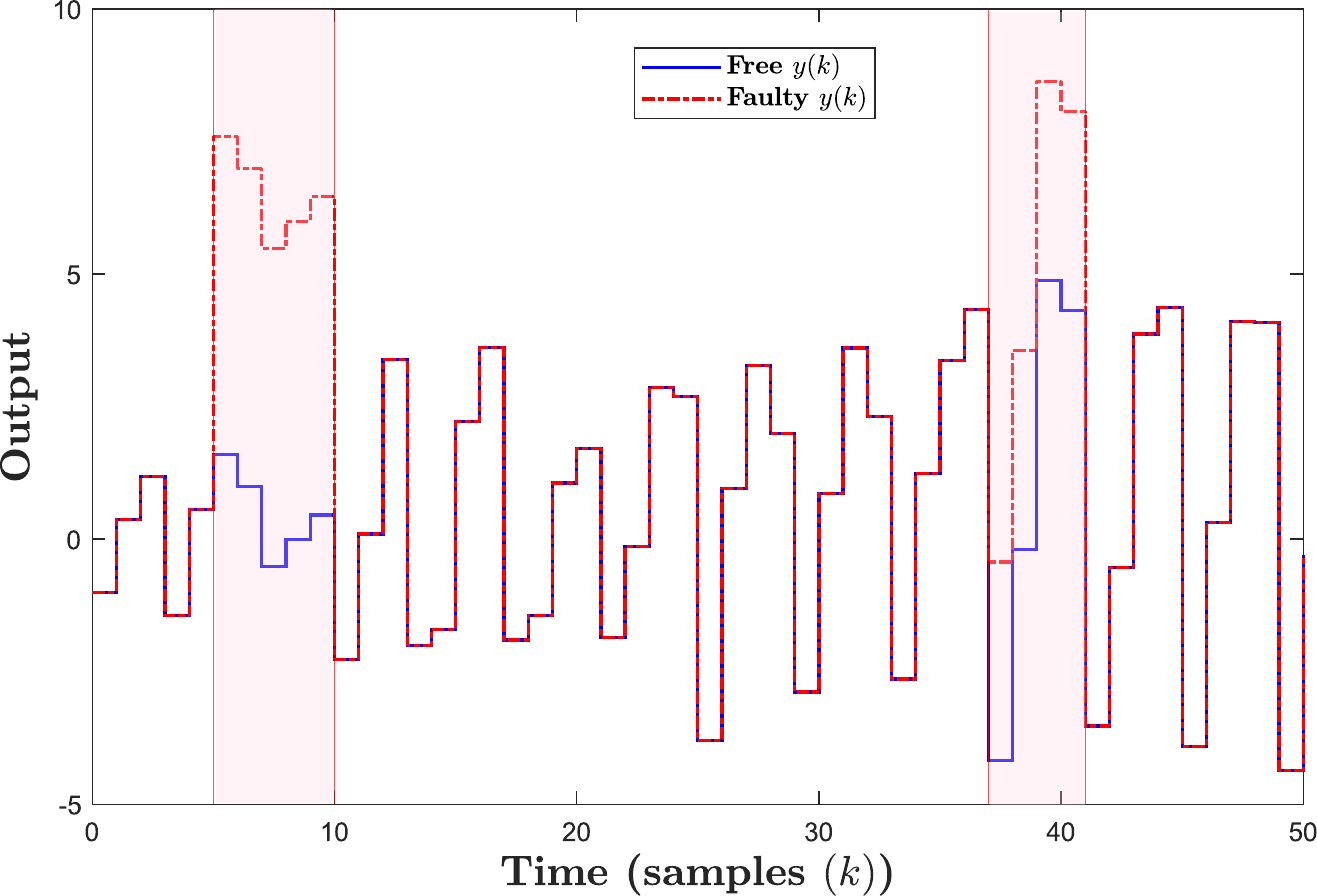}
\includegraphics[scale=0.35]{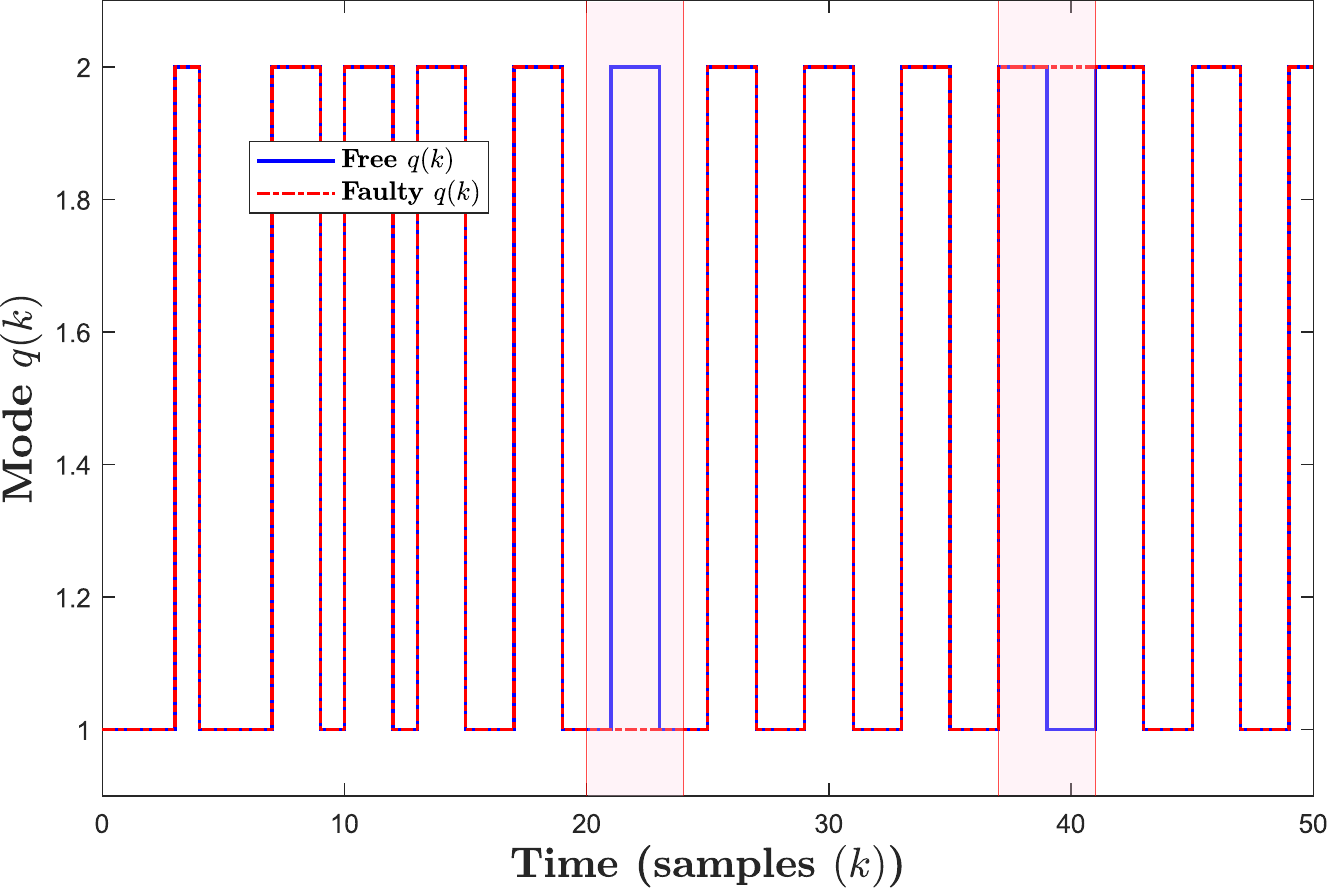}
\caption{Pattern of injected intermittent output sensor faults and discrete mode-blocking events.}
\label{figud}
\end{figure}

The results in Figure~\ref{figxf1d} demonstrate the decisive advantage of the ETCPN-HO. It exhibits rapid convergence and accurate state estimation following each fault injection (both separate and simultaneous), swiftly re-synchronizing with the true system state once normal operation resumes. In stark contrast, the HA-OB shows significantly degraded performance, struggling with slow recovery, larger estimation errors, and, in the case of simultaneous faults, a failure to converge properly within the observed timeframe. This performance gap underscores the critical importance of the ETCPN's intrinsically coupled modeling of hybrid dynamics for maintaining stability under complex fault conditions.

The evolution of the discrete mode estimate in Figure~\ref{figmff} further confirms the robustness of the ETCPN-HO, which successfully identifies and recovers from all mode-blocking events to provide accurate mode information during both separate and simultaneous fault occurrences. The estimation error dynamics in Figure~\ref{error33} provide further empirical evidence of stable performance, showing the ETCPN-HO's residuals converging rapidly to near zero after each fault. This behavior provides strong empirical evidence for the asymptotic stability of the ETCPN-HO, as guaranteed theoretically by the LMI conditions. In conclusion, for the most challenging scenario of concurrent faults, multiple and simultaneous faults, the proposed ETCPN-HO framework demonstrates a clear and decisive advantage over the HA-OB benchmark, providing the resilient and accurate monitoring necessary for dependable fault detection in safety-critical hybrid systems.

\begin{figure}[H]
  \centering
  \includegraphics[scale=0.35]{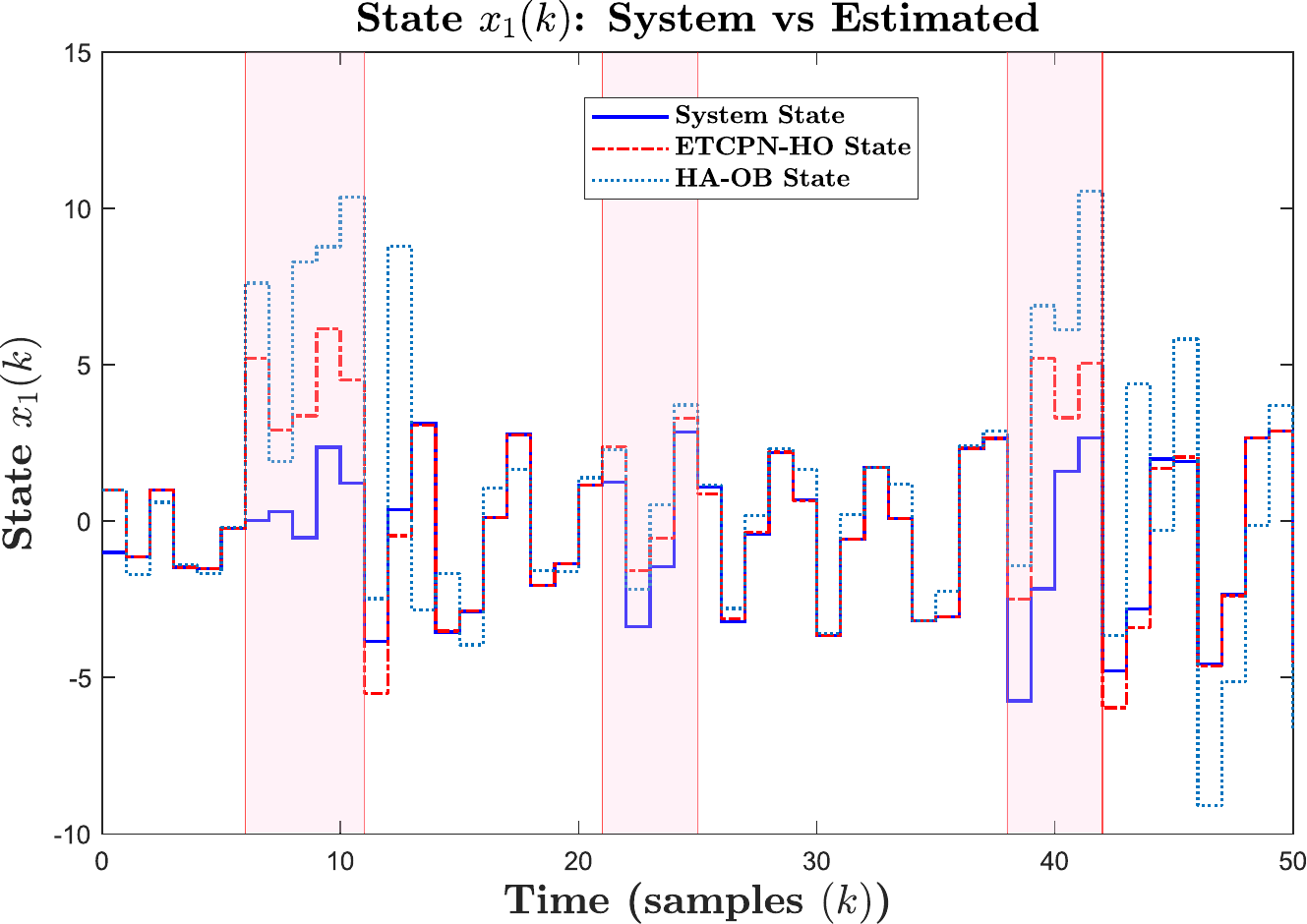}
  \includegraphics[scale=0.35]{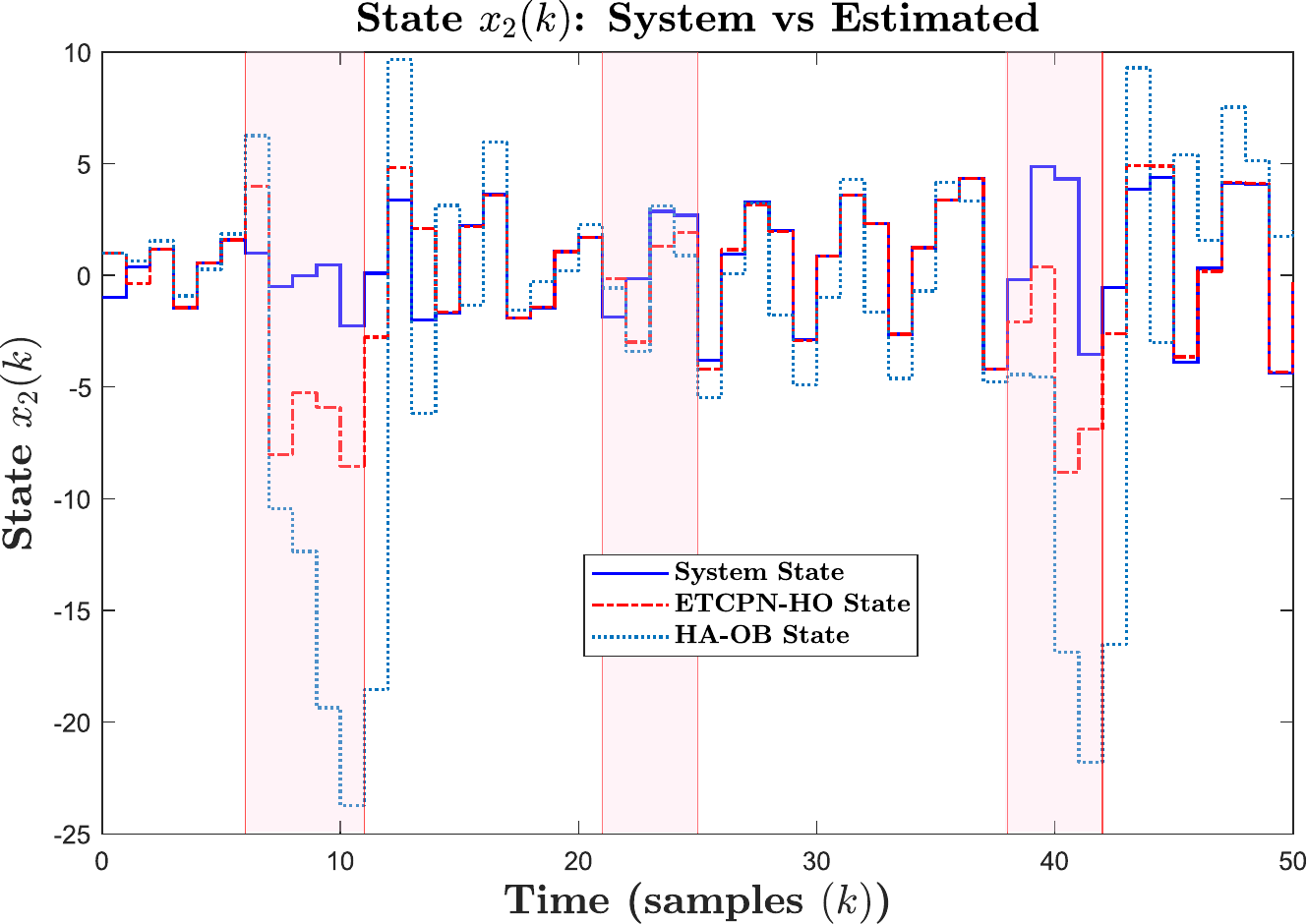}
  \caption{Evolution of ETCPN switched observer states under intermittent output faults.}
  \label{figxf1d}
\end{figure}

\begin{figure}[H]
  \centering
  \includegraphics[scale=0.35]{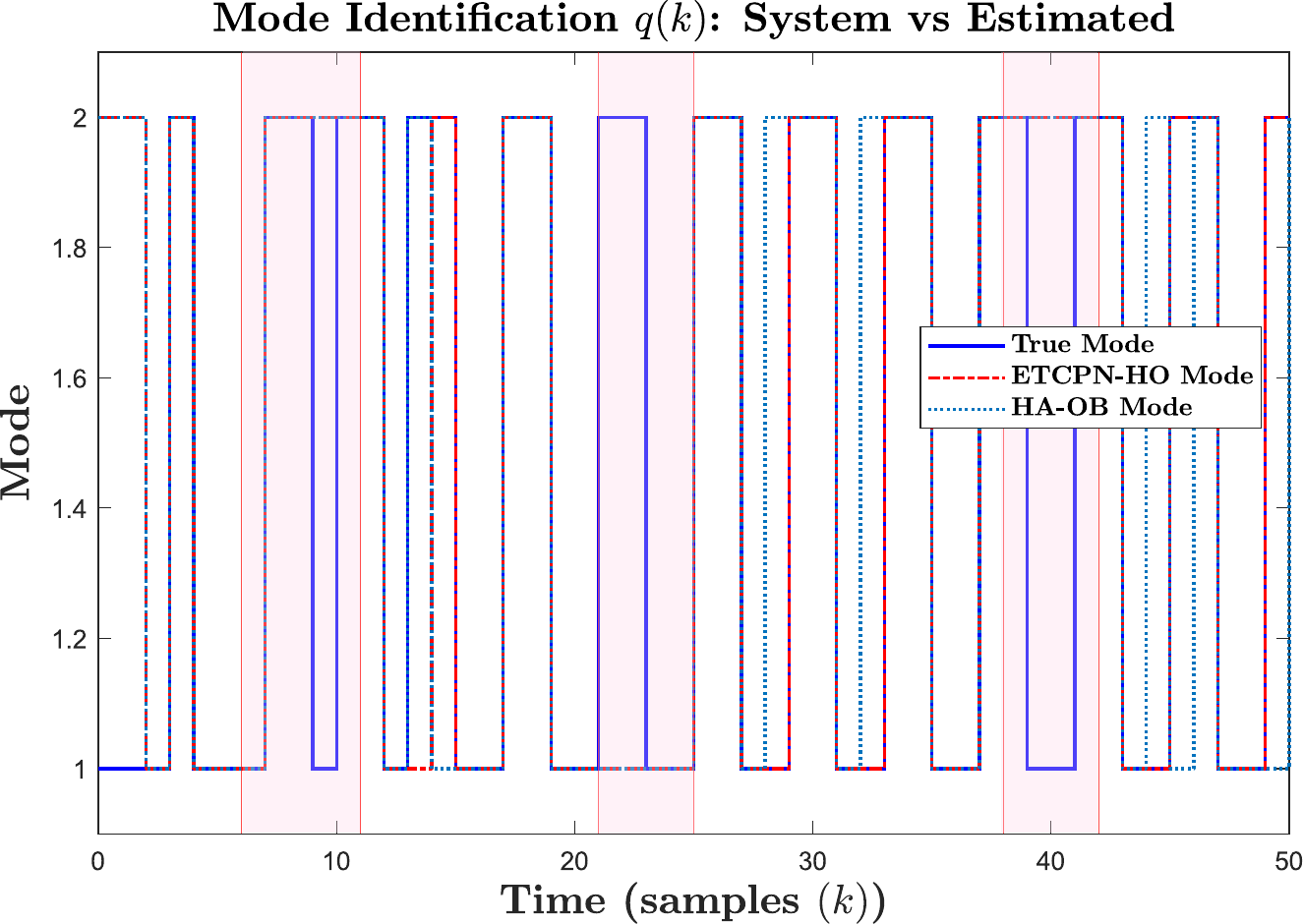}
  \caption{Evolution of ETCPN switched observer modes $q(k)$ during fault occurrences.}
  \label{figmff}
\end{figure}

\begin{figure}[H]
  \centering
  \includegraphics[scale=0.35]{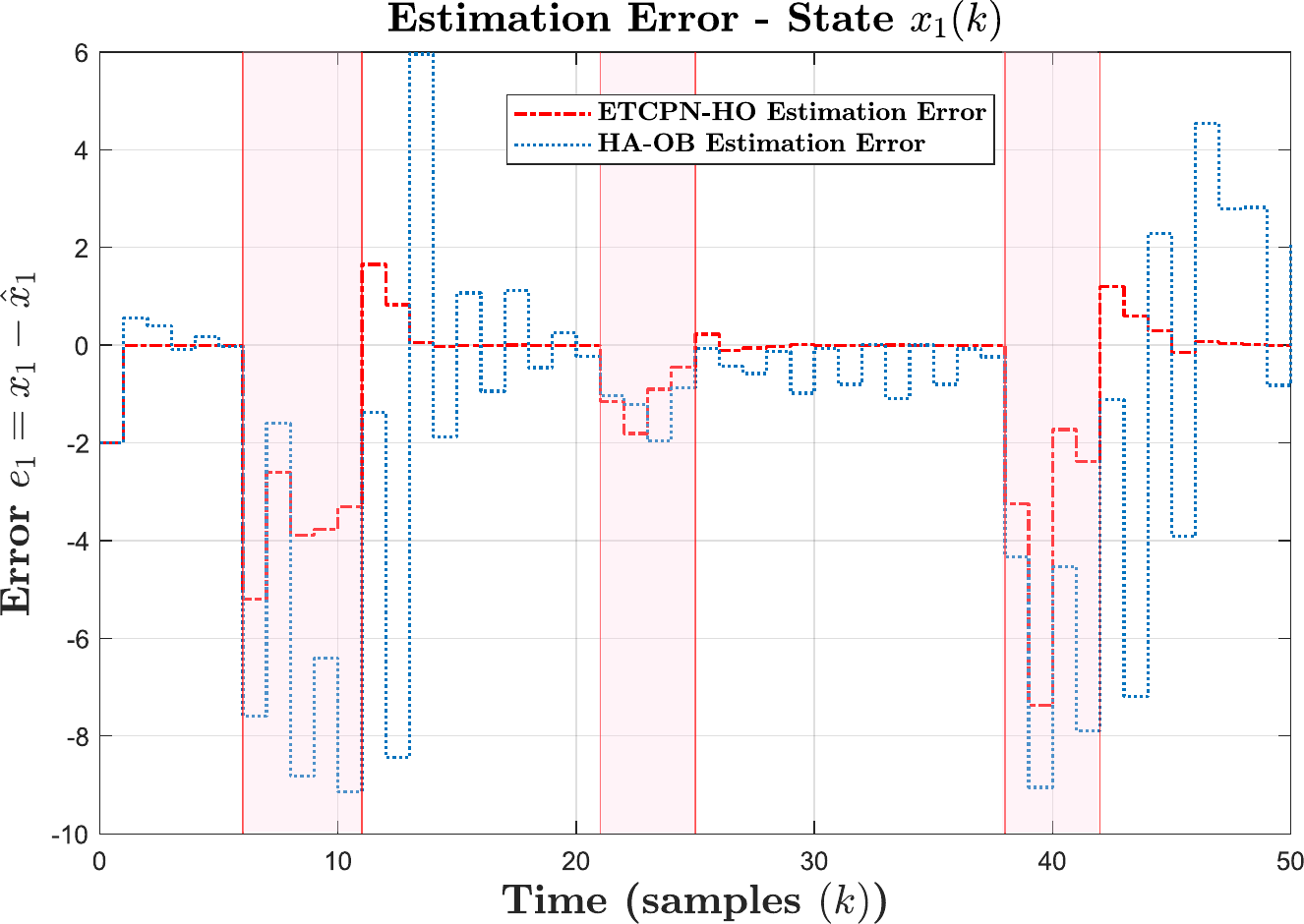}
  \includegraphics[scale=0.35]{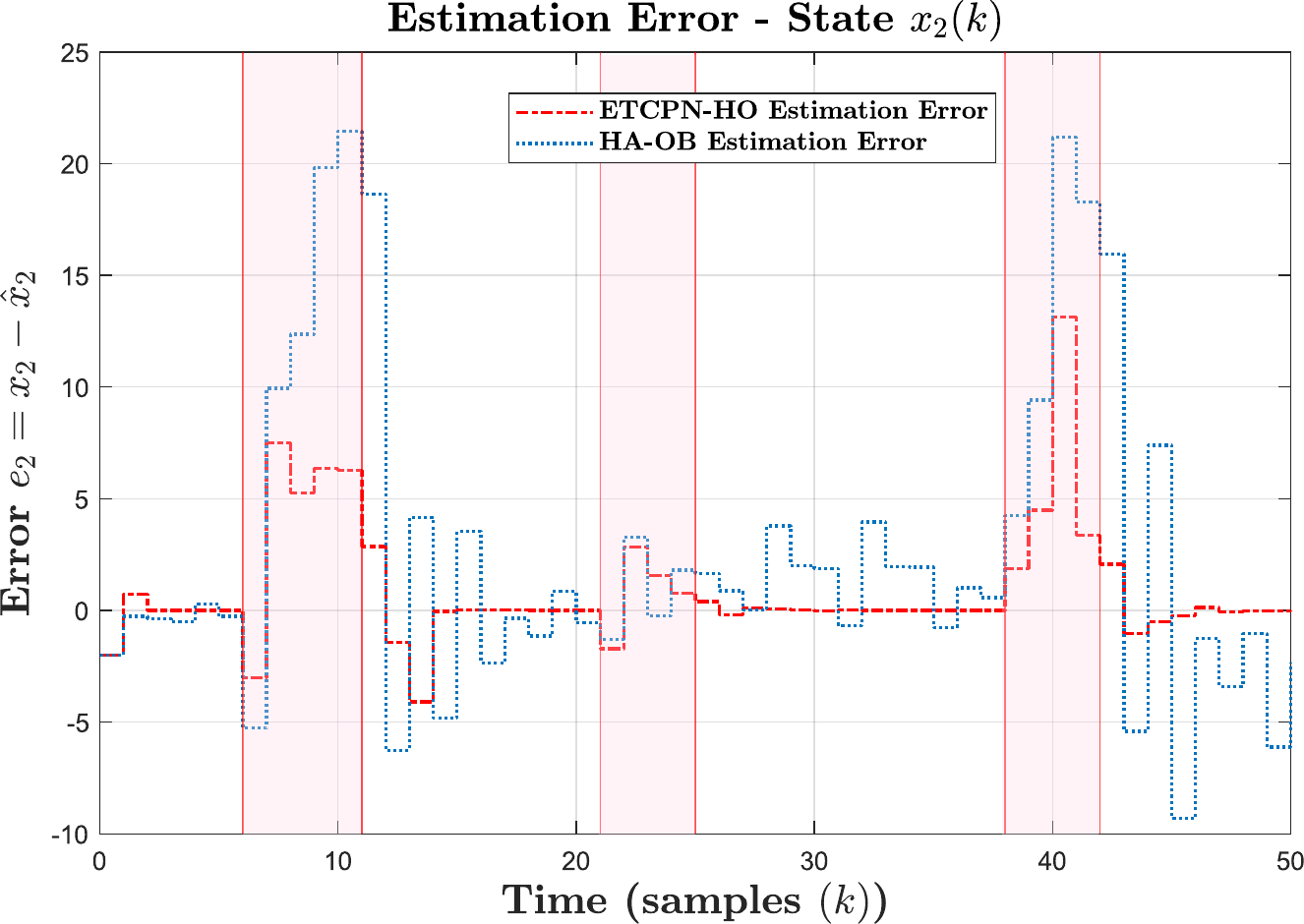}
  \caption{Evolution of estimation error dynamics during continuous faults.}
  \label{error33}
\end{figure}

Table~\ref{tab:case3} reports the quantitative detection metrics of the three considered detection methods for Case 3. The results indicate that all methods achieve relatively high overall accuracy, with OCSVM obtaining the best performance (90\%) and the highest F1 score (0.828), reflecting a balanced tradeoff between precision and recall. SVDD achieves competitive results with an F1 score of 0.786 and a recall of 0.688, though slightly lower than OCSVM. In contrast, the EE achieves perfect specificity with zero false positives, but this comes at the cost of reduced sensitivity (recall of 0.562). Results highlight the inherent trade-off between maximizing fault detection coverage and minimizing false alarms. OCSVM provides the most robust performance in handling the hybrid faults of Case 3, while SVDD offers a reasonable balance.

\begin{table}[h!]
\centering
\caption{Detection performance across the three methods on Case 3.}
\label{tab:case3}
\begin{tabular}{lcccc}
\hline
\textbf{Method} & \textbf{Accuracy} & \textbf{Recall} & \textbf{FPR} & \textbf{F1 Score} \\
\hline
OCSVM & 0.900 & 0.750 & 0.029 & 0.828 \\
EE & 0.860 & 0.562 & 0.000 & 0.720 \\
SVDD & 0.880 & 0.688 & 0.029 & 0.786 \\
\hline
\end{tabular}
\end{table}

\medskip

\noindent Overall, the experimental investigations, spanning both discrete-event and continuous fault scenarios, clearly demonstrate the effectiveness of the proposed ETCPN-based residual generation and hybrid observer framework. The developed residual signals, when processed by semi-supervised anomaly detection algorithms, namely OC-SVM, SVDD, and Elliptic Envelope, enabled accurate and timely fault detection without the need for labeled fault data. While discrete-event faults were characterized by abrupt mode blocking, continuous faults induced subtle deviations in the system state trajectories. Nevertheless, the three fault types were successfully captured through the combined estimation and detection architecture. The comparative performance of these semi-supervised algorithms aligns with their theoretical foundations. The robustness of OC-SVM and SVDD to complex, non-Gaussian residual distributions makes them generally well-suited for fault detection in hybrid dynamic systems. Conversely, the Elliptic Envelope's performance is highly dependent on the Gaussianity of the nominal data, a condition often violated during dynamic mode transitions. While the specific performance metrics are case-dependent, these relative trends are expected to generalize across similar applications of the proposed ETCPN-HO framework. These results underscore the flexibility and robustness of the proposed method, establishing its suitability for real-world hybrid systems where faults may occur across both discrete and continuous dynamics.

\section{Conclusion}\label{sec7}

This paper introduced an innovative fault detection framework for hybrid dynamical systems by integrating Extended Timed Continuous Petri Net (ETCPN) modeling with semi-supervised anomaly detection. Overall, the results confirm the potential of the proposed approach as a flexible and efficient solution for fault detection in hybrid systems. The framework's strength lies in its unified modeling formalism, which is purpose-built for systematic observer design (whose stability is guaranteed via Linear Matrix Inequalities (LMIs)) and its use of state residuals for high-sensitivity detection of incipient faults without requiring pre-labeled fault data. Validation across multiple fault scenarios, including discrete mode-blocking, continuous sensor faults, and combined hybrid faults occurring separately or simultaneously, demonstrated the framework's capability for accurate state estimation and rapid fault detection. Among the tested semi-supervised methods, One-Class SVM and SVDD achieved the best trade-off between detection rate and false alarms.

A distinct practical advantage is the framework's real-time feasibility. The computationally intensive LMI-based observer synthesis is an offline process. At the same time, online execution relies only on efficient matrix-vector operations and simple residual checks, making it suitable for deployment on modern industrial controllers. The memory footprint, which scales with the number of modes and the state dimension, remains manageable for typical industrial applications. The primary challenges for industrial adoption lie in the initial modeling effort and the long-term maintenance required to address potential model mismatch as physical systems evolve. 

Future work will explore adaptive thresholding and uncertainty modeling to enhance detection robustness under noisy and highly variable operating conditions. In addition, explainable and interpretable fault detection will be investigated through advanced explainable AI (XAI) techniques, enabling operators to gain actionable insights into the nature and causes of detected anomalies.

\appendix

\section*{Appendix A. Proof of Theorem~\ref{theo1}}\label{proofA}

\begin{proof}

The demonstration proceeds in two parts:

\begin{enumerate}

\item \textbf{Case without output.}

Consider the system without an output. The ETCPN places are defined as:

\begin{equation}\label{eqau1}
P_q^C =
\begin{bmatrix}
u_{p \times 1} & x_{n \times 1}
\end{bmatrix}^T {\color{red}\,,}
\end{equation}

\noindent where $u$ and $x$ denote the input and state vectors, respectively. The marking vector $M^C(k)$ is thus expressed as:

\begin{equation}\label{eqau2}
M_q^C(k) =
\begin{bmatrix}
u(k) & x(k)
\end{bmatrix}^T{\color{red}.}
\end{equation}

\noindent At time $k+1$, assuming maximum firing speed, the marking evolves according to:

\begin{equation}\label{eq23}
M_q^C(k+1) = \left[ I_{f \times f} + (Post - Pre) \right] M_q^C(k){\color{red}.}
\end{equation}

\noindent From the system model, the evolution can also be expressed as:

\begin{equation}\label{eq24}
\begin{bmatrix}
u(k+1) \\
x(k+1)
\end{bmatrix}
=
\begin{bmatrix}
u(k) \\
A_q x(k) + B_q u(k)
\end{bmatrix}{\color{red}\,,}
\end{equation}

\noindent which is equivalent to:

\begin{equation}\label{eq25}
\begin{bmatrix}
u(k+1) \\
x(k+1)
\end{bmatrix}
=
\begin{bmatrix}
I_{p \times p} & 0_{p \times n} \\
B_q & A_q
\end{bmatrix}
\begin{bmatrix}
u(k) \\
x(k)
\end{bmatrix}{\color{red}.}
\end{equation}

\noindent By comparison with Equation~\eqref{eq23}, and noting the structure of $(Post - Pre)$, the incidence matrix is deduced:

\begin{equation}\label{eq26}
(Post - Pre) =
\begin{bmatrix}
0_{p \times p} & 0_{p \times n} \\
B_q & A_q - I_{n \times n}
\end{bmatrix}{\color{red}.}
\end{equation}

\noindent Where:

\begin{equation}\label{eaq27}
Pre = I_{f \times f}{\color{red}\,,}
\end{equation}

\begin{equation}\label{eq27}
Post =
\begin{bmatrix}
I_{p \times p} & 0_{p \times n} \\
B_q & A_q 
\end{bmatrix}{\color{red}.}
\end{equation}

\noindent Therefore, the incidence matrix $W_q^C$ is given by:

\begin{equation*}
W_q^C = Post - Pre =
\begin{bmatrix}
I_{p \times p} & 0_{p \times n} \\
B_q & A_q 
\end{bmatrix}
- I_{f \times f}{\color{red}.}
\end{equation*}

\item \textbf{Case with output.}

When outputs are considered, the continuous places are extended to:

\begin{equation}\label{eq28}
P_q^C =
\begin{bmatrix}
u_{p \times 1} & x_{n \times 1} & y_{r \times 1}
\end{bmatrix}^T{\color{red}\,,}
\end{equation}

\noindent and the marking vector becomes:

\begin{equation}\label{eq29}
M_q^C(k) =
\begin{bmatrix}
u(k) & x(k) & y(k)
\end{bmatrix}^T{\color{red}.}
\end{equation}

\noindent Following a similar procedure as in the previous case, the incidence matrix formulation leads to the expression presented in Equation~\eqref{eqnaut} of Theorem~\ref{theo1}.
\end{enumerate}
\end{proof}

\section*{Appendix B. Proof of Theorem~\ref{theo2}} \label{proofB}

\begin{proof}

Consider the discrete time continuous state dynamics of the ETCPN in mode $q$ given by Equation~\eqref{eqn33}, the estimation error $e(k) = x(k) - \hat{x}(k)$ therefore has the dynamics:
\begin{equation}\label{error_dynamics}
e(k+1) = (A_q - L_q C_q) e(k) = \underline{A}_q e(k){\color{red}.}
\end{equation} 

This represents a switched linear system, where the switching signal is the active discrete mode $q$ of the Petri net. To guarantee stability under arbitrary switching sequences generated by the discrete Petri net, we employ a mode-dependent Lyapunov function:

\begin{equation}
V_q(e(k)) = e(k)^T P_q e(k), \quad P_q = P_q^T > 0{\color{red}.}
\end{equation}

For asymptotic stability, the difference of this function must be negative definite for any switch from mode $q$ to mode $\acute{q}$:

\begin{align}
\Delta V &= V_{\acute{q}}(e(k+1)) - V_q(e(k)) \prec 0 \nonumber \\
&=e(k)^T (A_q - L_q C_q)^T P_{\acute{q}} (A_q - L_q C_q) e(k) -  e(k)^T P_q e(k)\\
&=e(k)^T \left[(A_q - L_q C_q)^T P_{\acute{q}} (A_q - L_q C_q) -  P_q \right] e(k)\\
&= e(k)^T \left[ \underline{A}_q^T P_{\acute{q}} \underline{A}_q - P_q \right] e(k) \prec 0 \label{lyap_inequality}
\end{align}

Where $P_q$ and $P_{\acute{q}}$ are successfully the Lyapunov matrix associated with the mode $q$ and their successor mode $\acute{q}$. This inequality is not linear in the variables $P_q, P_{\acute{q}}, L_q$. To convexify it, we perform the following steps:

\begin{itemize}
  \item Apply the Schur complement to \eqref{lyap_inequality}. This is equivalent to:
\begin{equation}\label{schur1}
\begin{bmatrix}
-P_q & \underline{A}_q^T P_{\acute{q}}\\

* & -P_{\acute{q}}^{-1}
\end{bmatrix} \prec 0{\color{red}.}
\end{equation}
  \item Define a change of variables: $T_q = P_q^{-1}$ and $T_{\acute{q}} = P_{\acute{q}}^{-1}$. Introduce a variable $G_q$, non-singular matrix, pre-multiplying \eqref{schur1} by $\text{diag}(G_q^T, T_{\acute{q}})$ and post-multiplying it by $\text{diag}(G_q, T_{\acute{q}})$ yields:
\begin{equation}\label{schur2}
\begin{bmatrix}
-G_q^T T_q^{-1}G_q & G_q^T \underline{A}_q^T \\

* & -T_{\acute{q}}
\end{bmatrix} \prec 0{\color{red}.}
\end{equation}
  \item The inequality $-G_q^T T_q^{-1}G_q \prec -G_q - G_q^T + T_q$ always holds~\cite{zeroual2019road}. Using this and defining a new variable $F_q$, to avoid the non-linearity $L_q C_q G_q$, where $F_qC_q = L_q C_q G_q $, we can write the inequality \eqref{schur2} in the form of the LMI \eqref{theob}:

\begin{equation}
\begin{bmatrix}
T_q - G_q - G_q^T & G_q^T A_q^T - C_q^T F_q^T \\
\ast & -T_{\acute{q}}
\end{bmatrix} \prec 0{\color{red}.}
\end{equation}
\end{itemize}

\end{proof}

\bibliographystyle{IEEEtran}
\bibliography{bib}

@Article{taghavian2024constrained,
  author  = {Ali Taghavian and Ali Safi and Esmaeel Khanmirza},
  journal = {Journal of the Franklin Institute},
  title   = {Constrained computational hybrid controller for Input Affine Hybrid Dynamical Systems},
  year    = {2024},
  issn    = {0016-0032},
  number  = {17},
  pages   = {107142},
  volume  = {361},
}

@Article{khorasgani2017structural,
  author  = {Khorasgani, Hamed and Biswas, Gautam},
  journal = {IEEE Transactions on Automation Science and Engineering},
  title   = {Structural Fault Detection and Isolation in Hybrid Systems},
  year    = {2018},
  number  = {4},
  pages   = {1585-1599},
  volume  = {15},
}

@Article{wang2012hybrid,
  author  = {Zhenheng Wang and Jinsong Zhao and Helen Shang},
  journal = {Journal of Process Control},
  title   = {A hybrid fault diagnosis strategy for chemical process startups},
  year    = {2012},
  issn    = {0959-1524},
  number  = {7},
  pages   = {1287-1297},
  volume  = {22},
}

@Article{shi2024novel,
  author  = {Yilin Shi and Ning Zhang and Xiaolu Song and Hongguang Li and Qunxiong Zhu},
  journal = {Journal of Process Control},
  title   = {Novel approach for industrial process anomaly detection based on process mining},
  year    = {2024},
  issn    = {0959-1524},
  pages   = {103165},
  volume  = {136},
}

@Article{yang2010observer,
  author  = {Yang, Hao and Jiang, Bin and Cocquempot, Vincent},
  journal = {International Journal of Robust and Nonlinear Control},
  title   = {Observer-based fault-tolerant control for a class of hybrid impulsive systems},
  year    = {2010},
  number  = {4},
  pages   = {448-459},
  volume  = {20},
}

@Article{diedrich2019model,
  author  = {Diedrich, Alexander and Maier, Alexander and Niggemann, Oliver},
  journal = {Proceedings of the AAAI Conference on Artificial Intelligence},
  title   = {Model-Based Diagnosis of Hybrid Systems Using Satisfiability Modulo Theory},
  year    = {2019},
  month   = {Jul.},
  number  = {01},
  pages   = {1452-1459},
  volume  = {33},
}

@InProceedings{hofbaur2002mode,
  author    = {Hofbaur, Michael W. and Williams, Brian C.},
  booktitle = {Hybrid Systems: Computation and Control},
  title     = {Mode Estimation of Probabilistic Hybrid Systems},
  year      = {2002},
  address   = {Berlin, Heidelberg},
  editor    = {Tomlin, Claire J. and Greenstreet, Mark R.},
  pages     = {253--266},
  publisher = {Springer Berlin Heidelberg},
}

@Article{liu2018new,
  author  = {Hai Liu and Maiying Zhong and Yang Liu},
  journal = {IFAC-PapersOnLine},
  title   = {A New Strategy of Adaptive Observer Based Fault Isolation},
  year    = {2018},
  issn    = {2405-8963},
  note    = {10th IFAC Symposium on Fault Detection, Supervision and Safety for Technical Processes SAFEPROCESS 2018},
  number  = {24},
  pages   = {1373-1378},
  volume  = {51},
}

@Article{lien2020adaptive,
  author  = {Lien, Yu-Hsuan and Peng, Chao-Chung and Chen, Yi-Hsuan},
  journal = {Applied Sciences},
  title   = {Adaptive Observer-Based Fault Detection and Fault-Tolerant Control of Quadrotors under Rotor Failure Conditions},
  year    = {2020},
  issn    = {2076-3417},
  number  = {10},
  volume  = {10},
}

@Article{yan2024fault,
  author  = {Jing-Jing Yan and Chao Deng and Wei-Wei Che and Xiao-Xu Liu},
  journal = {Journal of the Franklin Institute},
  title   = {Fault estimation for cyber–physical systems with intermittent measurement transmissions via a hybrid observer approach},
  year    = {2024},
  issn    = {0016-0032},
  number  = {3},
  pages   = {1497-1509},
  volume  = {361},
}

@Article{novelli2025identification,
  author  = {Novelli, Nico and Belardinelli, Pierpaolo and Lenci, Stefano},
  journal = {Nonlinear Dynamics},
  title   = {Identification of hybrid dynamic systems via a sparse regression algorithm},
  year    = {2025},
  issn    = {1573-269X},
  month   = aug,
  number  = {16},
  pages   = {20565--20588},
  volume  = {113},
}

@Article{yin2024dynamic,
  author  = {Shicai Yin and Tao Peng and Chao Yang and Chunhua Yang and Zhiwen Chen and Weihua Gui},
  journal = {Control Engineering Practice},
  title   = {Dynamic hybrid observer-based early slipping fault detection for high-speed train wheelsets},
  year    = {2024},
  issn    = {0967-0661},
  pages   = {105736},
  volume  = {142},
}

@Article{zhu2025stability,
  author    = {Zhu, Yanzheng and Che, Junxing and Wu, Fen and Chen, Xinkai and Zheng, Weixing and Zhou, Donghua},
  journal   = {IEEE/CAA Journal of Automatica Sinica},
  title     = {Stability, Control and Fault Diagnosis of Switched Linear Parameter Varying Systems: A Survey},
  year      = {2025},
  number    = {9},
  pages     = {1745--1761},
  volume    = {12},
  publisher = {IEEE/CAA Journal of Automatica Sinica},
}

@InProceedings{harrou2018traffic,
  author    = {Harrou, Fouzi and Zeroual, Abdelhafid and Sun, Ying},
  booktitle = {2018 Annual American Control Conference (ACC)},
  title     = {Traffic congestion detection based on hybrid observer and GLR test},
  year      = {2018},
  pages     = {604-609},
}

@Article{chen2021hybrid,
  author  = {Yu Chen and Yue-E Wang and Di Wu},
  journal = {Journal of the Franklin Institute},
  title   = {Hybrid state observer‐based event‐triggered control for switched linear systems with quantized input},
  year    = {2021},
  issn    = {0016-0032},
  number  = {17},
  pages   = {9086-9109},
  volume  = {358},
}

@Article{lefebvre2017discussion,
  author  = {Dimitri Lefebvre and Enrique Aguayo-Lara},
  journal = {IFAC-PapersOnLine},
  title   = {A Discussion on Fault detection for a class of Hybrid Petri Nets},
  year    = {2017},
  issn    = {2405-8963},
  note    = {20th IFAC World Congress},
  number  = {1},
  pages   = {6837-6842},
  volume  = {50},
}

@Article{davrazos2007modeling,
  author  = {Gregory Davrazos and Nick T. Koussoulas},
  journal = {Simulation Modelling Practice and Theory},
  title   = {Modeling and stability analysis of state-switched hybrid systems via Differential Petri Nets},
  year    = {2007},
  issn    = {1569-190X},
  number  = {8},
  pages   = {879-893},
  volume  = {15},
}

@Article{chen2014novel,
  author  = {Fuyang Chen and Li Wang and Bin Jiang and Changyun Wen},
  journal = {Journal of the Franklin Institute},
  title   = {A novel hybrid petri net model for urban intersection and its application in signal control strategy},
  year    = {2014},
  issn    = {0016-0032},
  number  = {8},
  pages   = {4357-4380},
  volume  = {351},
}

@Article{renganathan2011observer,
  author  = {K. Renganathan and VidhyaCharan Bhaskar},
  journal = {ISA Transactions},
  title   = {An observer based approach for achieving fault diagnosis and fault tolerant control of systems modeled as hybrid Petri nets},
  year    = {2011},
  issn    = {0019-0578},
  number  = {3},
  pages   = {443-453},
  volume  = {50},
}

@Article{hatte2024transforming,
  author  = {Léonie Hatte and Pauline Ribot and Elodie Chanthery},
  journal = {IFAC-PapersOnLine},
  title   = {Transforming Time Petri Nets into Heterogeneous Petri Nets for Hybrid System Monitoring},
  year    = {2024},
  issn    = {2405-8963},
  note    = {12th IFAC Symposium on Fault Detection, Supervision and Safety for Technical Processes SAFEPROCESS 2024},
  number  = {4},
  pages   = {646-651},
  volume  = {58},
}

@PhdThesis{petri1962kommunikation,
  author     = {Petri, Carl},
  school     = {TU Darmstadt},
  title      = {Kommunikation mit Automaten},
  year       = {1962},
  dnbtitleid = {481889272},
}

@Article{grobelna2021challenges,
  author  = {Grobelna, Iwona and Karatkevich, Andrei},
  journal = {Electronics},
  title   = {Challenges in Application of Petri Nets in Manufacturing Systems},
  year    = {2021},
  issn    = {2079-9292},
  number  = {18},
  volume  = {10},
}

@InBook{kahraman2010manufacturing,
  author    = {Kahraman, Cengiz and T{\"u}ys{\"u}z, Fatih},
  editor    = {Kahraman, Cengiz and Yavuz, Mesut},
  pages     = {95--124},
  publisher = {Springer Berlin Heidelberg},
  title     = {Manufacturing System Modeling Using Petri Nets},
  year      = {2010},
  address   = {Berlin, Heidelberg},
  booktitle = {Production Engineering and Management under Fuzziness},
}

@InProceedings{vescio2015petri,
  author    = {Vescio, Giovanni and Riccobon, Paolo and Grasselli, Umberto and De Angelis, Francesco},
  booktitle = {2015 Annual Reliability and Maintainability Symposium (RAMS)},
  title     = {A Petri Net model for electrical power systems operating procedures},
  year      = {2015},
  pages     = {1-6},
}

@Article{julvez2005modelling,
  author  = {Jorge Júlvez and René Boel},
  journal = {IFAC Proceedings Volumes},
  title   = {MODELLING AND CONTROLLING TRAFFIC BEHAVIOUR WITH CONTINUOUS PETRI NETS},
  year    = {2005},
  issn    = {1474-6670},
  note    = {16th IFAC World Congress},
  number  = {1},
  pages   = {43-48},
  volume  = {38},
}

@InProceedings{liang2021modeling,
  author    = {Liang, Xin and Dang, Yuang and Hou, YiFan},
  booktitle = {2021 IEEE International Conference on Networking, Sensing and Control (ICNSC)},
  title     = {Modeling and Analysis of Urban Traffic System Based on Colored Petri Nets},
  year      = {2021},
  pages     = {1-6},
  volume    = {1},
}

@InProceedings{heiner2008petri,
  author    = {Heiner, Monika and Gilbert, David and Donaldson, Robin},
  booktitle = {Formal Methods for Computational Systems Biology},
  title     = {Petri Nets for Systems and Synthetic Biology},
  year      = {2008},
  address   = {Berlin, Heidelberg},
  editor    = {Bernardo, Marco and Degano, Pierpaolo and Zavattaro, Gianluigi},
  pages     = {215--264},
  publisher = {Springer Berlin Heidelberg},
}

@Article{mecheraoui2021petri,
  author  = {Khalil Mecheraoui and Irina A. Lomazova and Nabil Belala},
  journal = {Journal of Parallel and Distributed Computing},
  title   = {A Petri net extension for systems of concurrent communicating agents with durable actions},
  year    = {2021},
  issn    = {0743-7315},
  pages   = {14-23},
  volume  = {155},
}

@Article{ali2024modeling,
  author  = {Nazakat Ali and Sasikumar Punnekkat and Abdul Rauf},
  journal = {Journal of Systems and Software},
  title   = {Modeling and safety analysis for collaborative safety-critical systems using hierarchical colored Petri nets},
  year    = {2024},
  issn    = {0164-1212},
  pages   = {111958},
  volume  = {210},
}

@Book{david2010discrete,
  author    = {David, Ren{\'e} and Alla, Hassane},
  publisher = {Springer-Verlag},
  title     = {Discrete, {Continuous}, and {Hybrid} {Petri} {Nets}},
  year      = {2005},
  address   = {Berlin/Heidelberg},
}

@Article{alla1998continuous,
  author  = {ALLA, HASSANE and DAVID, REN\'{E}},
  journal = {Journal of Circuits, Systems and Computers},
  title   = {CONTINUOUS AND HYBRID PETRI NETS},
  year    = {1998},
  number  = {01},
  pages   = {159-188},
  volume  = {08},
}

@Article{demongodin1998differential,
  author  = {Demongodin, I. and Koussoulas, N.T.},
  journal = {IEEE Transactions on Automatic Control},
  title   = {Differential Petri nets: representing continuous systems in a discrete-event world},
  year    = {1998},
  number  = {4},
  pages   = {573-579},
  volume  = {43},
}

@Article{wang2005hybrid,
  author  = {Chan, Felix and Wang, Zheng and Zhang, Jie and Chan, Felix T.S.},
  journal = {Journal of Manufacturing Technology Management},
  title   = {A hybrid {Petri} nets model of networked manufacturing systems and its control system architecture},
  year    = {2005},
  month   = jan,
  number  = {1},
  pages   = {36--52},
  volume  = {16},
}

@Article{cavone2018hybrid,
  author  = {G. Cavone and M. Dotoli and N. Epicoco and M. Franceschelli and C. Seatzu},
  journal = {IFAC-PapersOnLine},
  title   = {Hybrid Petri Nets to Re-design Low-Automated Production Processes: the Case Study of a Sardinian Bakery},
  year    = {2018},
  issn    = {2405-8963},
  note    = {14th IFAC Workshop on Discrete Event Systems WODES 2018},
  number  = {7},
  pages   = {265-270},
  volume  = {51},
}

@Article{di2004modelling,
  author  = {Angela Di Febbraro and Nicola Sacco},
  journal = {Control Engineering Practice},
  title   = {On modelling urban transportation networks via hybrid Petri nets},
  year    = {2004},
  issn    = {0967-0661},
  note    = {Analysis and Design of Hybrid Systems},
  number  = {10},
  pages   = {1225-1239},
  volume  = {12},
}

@Article{mishra2023multi,
  author  = {Jyotismita Mishra and Pradyumna Kumar Behera and Monalisa Pattnaik and B. Chitti Babu},
  journal = {Sustainable Energy Technologies and Assessments},
  title   = {A multi-agent petri net model power management strategy for wind–solar-battery driven DC microgrid},
  year    = {2023},
  issn    = {2213-1388},
  pages   = {102859},
  volume  = {55},
}

@Article{brinkrolf2021vanesa,
  author  = {Christoph Brinkrolf and Lennart Ochel and Ralf Hofestädt},
  journal = {Biosystems},
  title   = {VANESA: An open-source hybrid functional Petri net modeling and simulation environment in systems biology},
  year    = {2021},
  issn    = {0303-2647},
  pages   = {104531},
  volume  = {210},
}

@Article{peleg2005using,
  author  = {Peleg, Mor and Rubin, Daniel and Altman, Russ B.},
  journal = {Journal of the American Medical Informatics Association},
  title   = {Using {Petri} {Net} {Tools} to {Study} {Properties} and {Dynamics} of {Biological} {Systems}},
  year    = {2005},
  issn    = {1067-5027},
  month   = mar,
  number  = {2},
  pages   = {181--199},
  volume  = {12},
}

@Article{kohler2023fault,
  author  = {Andreas Köhler and Ping Zhang},
  journal = {IFAC-PapersOnLine},
  title   = {Fault Diagnosis in Discrete Event Systems Modeled by Signal Interpreted Petri Nets},
  year    = {2023},
  issn    = {2405-8963},
  note    = {22nd IFAC World Congress},
  number  = {2},
  pages   = {4576-4581},
  volume  = {56},
}

@InProceedings{he2018modeling,
  author    = {He, Xudong},
  booktitle = {2018 IEEE International Conference on Software Quality, Reliability and Security Companion (QRS-C)},
  title     = {Modeling and Analyzing Cyber Physical Systems Using High Level Petri Nets},
  year      = {2018},
  pages     = {469-476},
}

@Article{jiang2024petri,
  author  = {Jiang, Zhongyuan and Wang, Huan and Wang, Wenjie},
  journal = {IEEE Access},
  title   = {A Petri Net Strategy for Fault Diagnosis and Location in Power Distribution Systems to Prevent Local Power Shortages},
  year    = {2024},
  pages   = {161038-161053},
  volume  = {12},
}

@Article{alzalab2021trust,
  author  = {Alzalab, Ebrahim Ali and El-Sherbeeny, Ahmed M. and El-Meligy, Mohammed A. and Rauf, Hafiz Tayyab},
  journal = {IEEE Access},
  title   = {Trust-Based Petri Net Model for Fault Detection and Treatment in Automated Manufacturing Systems},
  year    = {2021},
  pages   = {157997-158009},
  volume  = {9},
}

@Article{arichi2022fault,
  author  = {F. Arichi and B. Cherki and M. Djemai and S.M. Djouadi},
  journal = {ISA Transactions},
  title   = {Fault diagnosis for discrete events systems described by partially observed Petri nets},
  year    = {2022},
  issn    = {0019-0578},
  pages   = {220-228},
  volume  = {128},
}

@Article{coquand2023diagnosabilization,
  author  = {Camille Coquand and Yannick Pencolé and Audine Subias},
  journal = {IFAC-PapersOnLine},
  title   = {Diagnosabilization of Time Petri net for timed fault},
  year    = {2023},
  issn    = {2405-8963},
  note    = {22nd IFAC World Congress},
  number  = {2},
  pages   = {8648-8653},
  volume  = {56},
}

@Article{de2022online,
  author  = {Braian Igreja {de Freitas} and João Carlos Basilio},
  journal = {IFAC-PapersOnLine},
  title   = {Online Fault Diagnosis of Discrete Event Systems Modeled by Labeled Petri Nets Using Labeled Priority Petri Nets*},
  year    = {2022},
  issn    = {2405-8963},
  note    = {16th IFAC Workshop on Discrete Event Systems WODES 2022},
  number  = {28},
  pages   = {329-336},
  volume  = {55},
}

@Article{zhang2019robust,
  author  = {Zhang, Haigang},
  journal = {Information Technology and Control},
  title   = {Robust fault diagnosis for discrete-time switched system with unknown state delays subject to component faults},
  year    = {2019},
  number  = {1},
  pages   = {146--159},
  volume  = {48},
}

@Article{kazemi2018new,
  author  = {Mohammad Ghasem Kazemi and Mohsen Montazeri},
  journal = {International Journal of Electrical and Computer Engineering (IJECE)},
  title   = {A New Hybrid Robust Fault Detection of Switching Systems by Combination of Observer and Bond Graph Method},
  year    = {2018},
  number  = {4},
  volume  = {8},
}

@Article{du2022actuator,
  author  = {Du, Dongsheng and Wu, Yu and Yang, Yue and Mao, Zehui},
  journal = {Optimal Control Applications and Methods},
  title   = {Actuator fault detection for the discrete-time switched systems based on delta operator approach},
  year    = {2022},
  number  = {2},
  pages   = {476-494},
  volume  = {43},
}

@Article{el2025line,
  author  = {Rafika El Harabi and Manel Atitallah and Mohamed Naceur Abdelkrim},
  journal = {Measurement and Control},
  title   = {On-line switched robust fault detection framework for switched systems with unknown inputs},
  year    = {2025},
  pages   = {00202940251339800},
}

@Article{yahia2025actuator,
  author  = {Yahia, Salwa and Bedoui, Saida and Abderrahim, Kamel},
  journal = {IEEE Access},
  title   = {Actuator Fault Tolerant Control in Switched Systems: A Comprehensive Approach Integrating Clustering, Classification, and LMI-Based Compensation},
  year    = {2025},
  pages   = {44090-44106},
  volume  = {13},
}

@Article{song2018practical,
  author  = {Song, Jae-Hwan and Kim, Kyeong-Hwa},
  journal = {Energies},
  title   = {A Practical Approach to Localize Simultaneous Triple Open-Switches for a PWM Inverter-Fed Permanent Magnet Synchronous Machine Drive System},
  year    = {2018},
  issn    = {1996-1073},
  number  = {1},
  volume  = {11},
}

@Article{li2019sensor,
  author  = {Li, Jian and Pan, Kunpeng and Su, Qingyu and Zhao, Xiao-Qi},
  journal = {IEEE Access},
  title   = {Sensor Fault Detection and Fault-Tolerant Control for Buck Converter via Affine Switched Systems},
  year    = {2019},
  pages   = {47124-47134},
  volume  = {7},
}

@Article{li2020robust,
  author  = {Li, Jian and Pan, Kunpeng and Su, Qingyu},
  journal = {International Journal of Adaptive Control and Signal Processing},
  title   = {Robust fault detection and adaptive parameter identification for DC-DC converters via switched systems},
  year    = {2020},
  number  = {11},
  pages   = {1642-1657},
  volume  = {34},
}

@Article{waez2013survey,
  author  = {Md Tawhid Bin Waez and Juergen Dingel and Karen Rudie},
  journal = {Computer Science Review},
  title   = {A survey of timed automata for the development of real-time systems},
  year    = {2013},
  issn    = {1574-0137},
  pages   = {1-26},
  volume  = {9},
}

@Article{giua2018petri,
  author  = {Alessandro Giua and Manuel Silva},
  journal = {Annual Reviews in Control},
  title   = {Petri nets and Automatic Control: A historical perspective},
  year    = {2018},
  issn    = {1367-5788},
  pages   = {223-239},
  volume  = {45},
}

@Article{vardakis2025petri,
  author  = {Christoforos Vardakis and Ioannis Dimolitsas and Dimitrios Spatharakis and Dimitrios Dechouniotis and Anastasios Zafeiropoulos and Symeon Papavassiliou},
  journal = {Simulation Modelling Practice and Theory},
  title   = {A Petri Net-based framework for modeling and simulation of resource scheduling policies in Edge Cloud Continuum},
  year    = {2025},
  issn    = {1569-190X},
  pages   = {103098},
  volume  = {141},
}

@Article{lefebvre2015gradient,
  author  = {Dimitri Lefebvre and Edouard Leclercq and Fabrice Druaux and Philippe Thomas},
  journal = {International Journal of Systems Science},
  title   = {Gradient-based controllers for timed continuous Petri nets},
  year    = {2015},
  number  = {9},
  pages   = {1661--1678},
  volume  = {46},
}

@Article{HAMDI2009310,
  author  = {Fatiha Hamdi and Noureddine Manamanni and Nadhir Messai and Khier Benmahammed},
  journal = {Nonlinear Analysis: Hybrid Systems},
  title   = {Hybrid observer design for linear switched system via Differential Petri Nets},
  year    = {2009},
  issn    = {1751-570X},
  number  = {3},
  pages   = {310-322},
  volume  = {3},
}

@Article{scholkopf2001estimating,
  author  = {Sch{\"o}lkopf, Bernhard and Platt, John C. and Shawe-Taylor, John and Smola, Alex J. and Williamson, Robert C.},
  journal = {Neural Computation},
  title   = {Estimating the Support of a High-Dimensional Distribution},
  year    = {2001},
  number  = {7},
  pages   = {1443-1471},
  volume  = {13},
}

@Article{tax1999support,
  author  = {David M.J Tax and Robert P.W Duin},
  journal = {Pattern Recognition Letters},
  title   = {Support vector domain description},
  year    = {1999},
  issn    = {0167-8655},
  number  = {11},
  pages   = {1191-1199},
  volume  = {20},
}

@Article{rousseeuw1999fast,
  author  = {Peter J. Rousseeuw and Katrien Van Driessen},
  journal = {Technometrics},
  title   = {A Fast Algorithm for the Minimum Covariance Determinant Estimator},
  year    = {1999},
  number  = {3},
  pages   = {212--223},
  volume  = {41},
}

@Article{tax2004support,
  author  = {Tax, David M.J. and Duin, Robert P.W.},
  journal = {Machine Learning},
  title   = {Support {Vector} {Data} {Description}},
  year    = {2004},
  issn    = {1573-0565},
  month   = jan,
  number  = {1},
  pages   = {45--66},
  volume  = {54},
}

@Article{zeroual2019road,
  author  = {Abdelhafid Zeroual and Fouzi Harrou and Ying Sun},
  journal = {Sustainable Cities and Society},
  title   = {Road traffic density estimation and congestion detection with a hybrid observer-based strategy},
  year    = {2019},
  issn    = {2210-6707},
  pages   = {101411},
  volume  = {46},
}

@Article{ghomri2007modeling,
  author  = {Latefa Ghomri and Hassane Alla},
  journal = {Nonlinear Analysis: Hybrid Systems},
  title   = {Modeling and analysis using hybrid Petri nets},
  year    = {2007},
  issn    = {1751-570X},
  note    = {Nonlinear Hybrid Control Systems},
  number  = {2},
  pages   = {141-153},
  volume  = {1},
}

@Article{zeroual_piecewise_2017,
  author  = {Zeroual, Abdelhafid and Messai, Nadhir and Kechida, Sihem and Hamdi, Fatiha},
  journal = {International Journal of Automation and Computing},
  title   = {A piecewise switched linear approach for traffic flow modeling},
  year    = {2017},
  issn    = {1751-8520},
  month   = dec,
  number  = {6},
  pages   = {729--741},
  volume  = {14},
}

@InProceedings{hamdi2024ICSC,
  author    = {Fatiha, Hamdi and Abdelhafid, Zeroual},
  booktitle = {2024 12th International Conference on Systems and Control (ICSC)},
  title     = {Hybrid Timed Petri Net Framework for Switched Dynamical System Modelling},
  year      = {2024},
  pages     = {306-311},
}

@InProceedings{Hamdi2024AFROS,
  author    = {Fatiha, Hamdi and Abdelhafid, Zeroual},
  booktitle = {2024 International Conference of the African Federation of Operational Research Societies (AFROS)},
  title     = {Extended Timed Petri Net Modelling Approach for a Class of Hybrid Dynamic Systems},
  year      = {2024},
  pages     = {1-6},
}

@InProceedings{Zeroual2017ICEE-B,
  author    = {Abdelhafid, Zeroual and Harrou, Fouzi and Sun, Ying},
  booktitle = {2017 5th International Conference on Electrical Engineering - Boumerdes (ICEE-B)},
  title     = {An efficient statistical-based approach for road traffic congestion monitoring},
  year      = {2017},
  pages     = {1-5},
}

@InProceedings{zeroual2015calibration,
  author       = {Zeroual, A and Messai, N and Kechida, S and Hamdi, F},
  booktitle    = {2015 3rd International Conference on Control, Engineering \& Information Technology (CEIT)},
  title        = {Calibration and validation of a switched linear macroscopic traffic model},
  year         = {2015},
  organization = {IEEE},
  pages        = {1--5},
}

\end{document}